\documentclass[12pt,letterpaper]{csdissertation}
\usepackage{graphicx}
\usepackage{url}
\usepackage{cite}

\usepackage{lscape}
\usepackage{longtable}
\usepackage{booktabs}

\usepackage{multirow}
\usepackage{float}
\usepackage[utf8]{inputenc}
\usepackage{amsmath,amsthm}

\newtheorem{theorem}{Theorem}
\newtheorem{lemma}{Lemma}
\newtheorem{definition}{Definition}
\newtheorem{proposition}{Proposition}

\newcommand{\rpm}{\sbox0{$1$}\sbox2{$\scriptstyle\pm$}
  \raise\dimexpr(\ht0-\ht2)/2\relax\box2 }

\DeclareMathOperator*{\E}{\mathbb{E}}
\DeclareMathOperator*{\MI}{\mathcal{I}}
\DeclareMathOperator*{\MU}{\mathcal{U}}

\usepackage{amssymb}
\usepackage{mathtools}
\usepackage{algcompatible}
\usepackage{algorithm, algorithmicx}

\usepackage{subcaption}

\usepackage{cleveref}

\usepackage[square,numbers]{natbib} 

\title{Dissertation: On the Theoretical Foundation of  \\Model Comparison and Evaluation for Recommender System}
\author{Dong Li}
\degrees{
B.S., University of Science and Technology of China, 2015\\
M.S., Kent State University, 2017\\
Ph.D., Kent State University, 2023}
\submitdate{August 2023}
\advisor{Dr. Ruoming Jin}
\membera{Dr. Feodor Dragan}
\memberb{Dr. Qiang Guan}
\memberc{Dr. Omar De la Cruz Cabrera}
\memberd{Dr. Tsung-Heng Tsai}
\chair{Dr. Javed I. Khan}
\dean{Dr. Mandy Munro-Stasiuk}

\begin{document}
\frontmatter

\prefacesection{Acknowledgements}
  First, I would like to extend my thanks and appreciation to my advisor, Prof. Ruoming Jin. His encouragement and instruction are essential for me to complete this dissertation. The qualities like persistence and patience shining on him benefit me a lot.
Also, I would say thank you to all the professors in Computer Science Department. I shall thank all the professors who taught the courses I took during my years at KSU.
Computer Science Department has so many nice staffs. I would thank them for helping me solve many problems I met.

This work received partial support from a research agreement between Kent State University and iLambda, Inc., as well as from the National Science Foundation under Grant IIS-2142675.

\startthesis


\chapter{Introduction}\label{chpt:one}
\section{Introduction to Recommender System}

Recommender systems have become increasingly important with the rise of the web as a medium for electronic and business transactions. One of the key drivers of this technology is the ease with which users can provide feedback about their likes and dislikes through simple clicks of a mouse. This feedback is commonly collected in the form of ratings, but can also be inferred from a user's browsing and purchasing history. Recommender systems utilize users' historical data to infer customer interests and provide personalized recommendations. The basic principle of recommendations is that significant dependencies exist between user- and item-centric activity, which can be learned in a data-driven manner to make accurate predictions. Collaborative filtering is one family of recommendation algorithms that uses ratings from multiple users to predict missing ratings or uses binary click information to predict potential clicks. However, recommender systems can be more complex and incorporate auxiliary data such as content-based attributes, user interactions, and contextual information. 

\section{Goals of Recommender System}

It is important to understand the different ways in which the recommendation problem can be approached. There are two main models:

\noindent\textbf{Prediction version of the problem:} In this approach, the goal is to predict the rating value for a particular user-item combination, assuming that there is training data available that indicates user preferences for items. The data is represented as an incomplete matrix of m × n dimensions, where the observed values are used for training and the missing values are predicted using the learning algorithm. This is known as the matrix completion problem, as the algorithm completes the missing values in the matrix.

\noindent\textbf{Ranking version of the problem:} In practice, it may not be necessary to predict the ratings of users for specific items in order to make recommendations. Instead, a merchant may want to recommend the top-k items for a particular user or the top-k users for a particular item. The determination of the top-k items is more common than the determination of the top-k users, and this problem is known as the top-k recommendation problem. This is the ranking formulation of the recommendation problem, and throughout this book, we will focus on the determination of the top-k items as it is the more prevalent scenario.

In this dissertation, we would deal with the latter one - the ranking-based recommender system problem.

\section{Design of Offline Recommender System}

\subsection{Implicit Feedback Data}

\begin{figure}
    \centering
    \includegraphics[width = 0.9\linewidth]{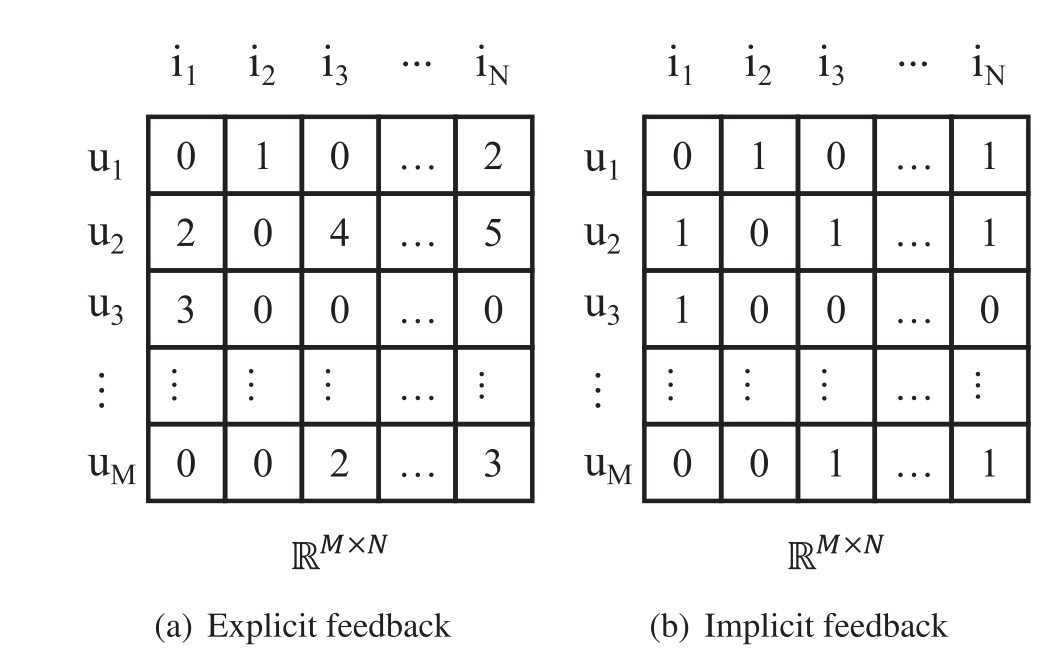}
    \caption{Interaction matrices of explicit and implicit feedback. Figure obtained from \cite{dnnrec}.}
    \label{fig:implicit}
\end{figure}

In \cref{fig:implicit}, (a) is the explicit feedback, for example, the user's ratings over the movies. (b) is the implicit feedback, we can interpret as a user $i$ rated/watched a movie $j$ if the $(i,j)$ term is 1. Throughout this dissertation, we would focus on only the implicit feedback and pure collaborative filtering models without leveraging other auxiliary information.

\subsection{Training Testing Data Splitting}

Online recommender systems typically utilize A/B testing to distinguish the performance of various recommendation models and select the best one. Randomly divide a large population into two or more groups and apply independent recommendation models in each group. By comparing the metrics between different groups, like Revenue, one can determine which one is the good recommendation model.

This is not feasible in an offline recommender system. To design an offline recommender system, researchers typically collect real-world data, for example, \cref{fig:implicit} (b). Then the data has to be split into training and testing two parts (sometimes, we also have extra validation data). 

Here is an example, the raw data is $X\in\mathbb{R}^{5\times 8}$:
\begin{equation*}
    X=\begin{bmatrix}
    1&0&1&0&1&1&0&1\\
    0&1&1&0&0&1&1&0\\
    1&1&0&1&0&1&1&1\\
    0&1&1&0&0&0&0&1\\
    0&0&0&1&1&0&0&1\\
    \end{bmatrix}
\end{equation*}
Each row is the historical clicked information for a user.
The first row means user $1$ clicked/liked item $1,3,5,6,8$, and other items remain unknown.
The splitting can be by ratio or by number.

\noindent\textbf{leave-one-out}: for each user, randomly pick one clicked item to form as test dataset and others as a training dataset.
\begin{equation*}
    X_{train}=\begin{bmatrix}
    1&0&1&0&1&1&0&0\\
    0&1&1&0&0&1&0&0\\
    1&1&0&1&0&1&0&1\\
    0&1&0&0&0&0&0&1\\
    0&0&0&0&1&0&0&1\\
    \end{bmatrix}\quad\quad 
       X_{test}=\begin{bmatrix}
    0&0&0&0&0&0&0&1\\
    0&0&0&0&0&0&1&0\\
0&0&0&0&0&0&1&0\\
  0&0&1&0&0&0&0&0\\
   0&0&0&1&0&0&0&0\\
    \end{bmatrix}
\end{equation*}

\textbf{A recommender system would learn from training data ($X_{train}$) and try to recommend the item in test data ($X_{test}$).}

\subsection{Evaluation of Offline Ranking-based Recommender System}
Taking the same example, a recommendation model $A$ would learn from the training data $X_{train}$ and recommend a personalized ranked item list to each user. For example, assume it would recommend a list $i_2>i_8>i_7>i_4$ to $u_1$, where $i_2$ is the item 2 and $u_1$ is the user 1. (It is guaranteed to recommend only unknown items which means clicked items won't appear in the recommended list).
Since the target item in the test set is $i_8$, we would see that the rank of the target item is $2$. For each user, there would be a personalized recommended item list and for each target item in the test set, it has a rank number. The offline evaluation would take these rank numbers to quantify the performance of the recommendation model. From the above, we know that a good recommendation model would likely rank the test item at the top of the list.

\subsection{Collaborative Filtering Model}

From the last section, we know that the mission of the recommendation model is to generate a personalized item ranking list for each user. There are multiple ways to approach it. The most straightforward way is Matrix Factorization:
\begin{equation*}
    X_{train}\approx V_u\cdot V^T_i = \widehat{X}
\end{equation*}
One can factorize the training matrix into two low-rank matrices $V_u$ and $V_i$, interpreted as the user representation/embedding matrix and item representation/embedding matrix. The values in $\widehat{X}$ can be used to rank the item naturally, thus one can form a ranking list for each user.
\section{Organization}

We organize the remaining of the dissertation as follows: In \cref{chpt:linear},  we discuss the linear models of recommendation \cite{jin@linear,ease,edlae}. Specifically, we reveal the similarities and differences between the matrix factorization approach and 
the regression-based approach; In \cref{chpt:ccl}, we unify the common objective functions from the perspective of contrastive learning; In \cref{chpt:sample-evaluation}, we theoretically and empirically study the sampling-based recommendation evaluation process \cite{jin2021towards,dong@2023aaai,li2022on,walid@sample,WalidRendle20}.



\chapter{Sampling-Based Evaluation}\label{chpt:sample-evaluation}

\section{Introduction}
\label{sec:sampling-evaluation-introduction}

Recommendation models are crucial in the AI-driven economy, and evaluating them rigorously is increasingly important. While online A/B tests are the ultimate criteria for discerning different recommendation models, running such a test often takes days or even weeks to draw a conclusion. Offline evaluations thus play a critical role in helping choose  promising/right recommendation models (based on historical data) for online testing. 

An offline evaluation first splits the user's historical data (for example, clicked items) into training and test sets (based on various data-splitting strategies~\cite{fayyaz2020recommendation}); it then learns the recommendation model parameters using the training data; finally, it performs a top-K metric calculation on the test set. 
The last step is often prohibitively expensive in computation because it needs to compute preference or likelihood scores of all items for each user to calculate the evaluation metrics (e.g., Recall and NDCG~\cite{gunawardana2009survey}).
To speed up the such calculation,  researchers/practitioners have resorted to the item-sampling strategy, i.e., calculating the top-$K$ evaluation metrics using only a small set of item samples~\cite{Koren08, CremonesiKT@10,he2017neural, ebesu2018collaborative,HuSWY18,krichene2018efficient,wang2018explainable,YangBGHE18,YangCXWB18}.

However, item sampling-based evaluation has given rise to a major controversy among the recommendation community, as \cite{rendle2019evaluation,krichene@kdd20} pointed out that commonly used (top-$K$) evaluation metrics, such as Recall (Hit-Ratio)/Precision, Average Precision (AP) and NDCG, (other than AUC), are all ``inconsistent'' with respect to the global metrics (even in expectation). They suggested a cautionary use (avoiding if possible) of the sampled metrics for recommendation evaluation. Due to the ubiquity of sampling methodology, it is not only of theoretical importance but also of practical interest to understand item-sampling evaluation. Indeed, since the number of items in any real-world recommendation system is typically quite large (easily in the order of tens of millions), efficient model evaluation based on item sampling can be very useful for recommendation researchers and practitioners.  

To address the discrepancy between item sampling results and the exact top-$K$ recommendation evaluation, \cite{krichene@kdd20,Li@KDD20} have proposed a few estimators in recovering global top-$K$ metrics from sampling. Despite these latest works on item-sampling estimation,  there remain some major gaps in making item-sampling reliable and accurate for top-$K$ metrics. To this end, we make the following contributions in this chapter: 
\begin{itemize}
    \item We reveal the alignment relationship between the sampling-based top-K Recall metric curve and global top-K Recall and propose a few mapping functions that help verify such an approximating linear relation.

    \item We propose two estimators $MES$ and $MLE$ that can estimate the global user rank distribution and also can help estimate the global true metric 
    
        \item We derive an optimal item-sampling estimator and highlight a subtle difference between the estimators derived by \cite{krichene@kdd20} and ours, and point out the potential issues of former estimator because it fails to link the user population size with the estimation variance.  To the best of our knowledge, this is the first estimator that directly optimizes the estimation errors. 

        \item  We address the limitation of the current item sampling approaches by proposing a new  adaptive sampling method. This provides a simple and effective remedy that helps avoid the sampling ``blind spot'', and significantly improves the accuracy of the estimated metrics with low sample complexity. 

        \item We study the effect of sampling from the perspective of users. We found that this is practical and efficient for a such dataset with much more users than items.

        \item We perform a thorough experimental evaluation of the proposed estimators. The experimental results confirm the statistical analysis and the superiority of our estimators.
    \end{itemize}

Our results help lay the theoretical foundation for adopting item sampling metrics for recommendation evaluation and offer a simple yet effective new adaptive sampling approach to help recommendation practitioners and researchers to apply item sampling-based approaches to speed up offline evaluation.   

We organize this chapter as follows: In \cref{sec:sampling-evaluation-introduction}, we introduce the background of the recommendation evaluation; In \cref{sec:sampling-evaluation-problem}, we briefly discuss the main problem of this chapter; In \cref{sec:map}, we take a glance at relationship between the global Top-K recall and sampling Top-K recall; In \cref{sec:fix_estimator_learning_rank}, we talk about various sampling-based estimators; In \cref{sec:adaptive_estimator}, we propose the new adaptive sampling and estimation method; In \cref{sec:user_sampling}, we discuss sampling from user perspective; Finally, in \cref{sec:sampling-evaluation-experiement}, we present the experimental results.

\section{Problem Formulation}
\label{sec:sampling-evaluation-problem}

\begin{table}
\label{tab:notations}

\begin{scriptsize}
 \begin{tabular}{|l|l|}
\hline
$U$	& the set of overall users, and $|U| = M$ \\
\hline
$I$	& the set of overall items, and $|I| = N$ \\
\hline
$M$ & total \# of users\\
\hline
$N$ & total \# of items \\
\hline
$i_u$	& the target item (to be ranked) for each user $u$\\
\hline 
$I_u$	& sampled test set for user $u$, consists of 1 target item $i_u$, $n-1$ sampled items \\
\hline 
$n$ & sample set size, $|I_u| = n$\\
\hline
$R_u$	& rank of item $i_u$ among all items $I$ for user $u$\\
\hline
$r_u$	& rank of item $i_u$ among $I_u$ for user $u$\\
\hline
$T$ & evaluation metric (for a recommendation model), e.g. $Recall$, $Recall@K$, etc \\
\hline
$\mathcal{F}$ & individual evaluation metric function (for each item rank), e.g. $Recall(\cdot)$, $Recall@K(\cdot)$, etc\\
\hline
$T^S$ & sampled evaluation metric \\
\hline
$\widehat{T}$ & estimated evaluation metric\\
\hline
  \end{tabular}
  \end{scriptsize}
  \caption{Notations for Leave-one-out evaluation setting.}
\end{table}

\subsection{Leave-one-out Recommendation Evaluation}\label{subsec:evaluation}
There are user set $U$, ($|U|=M$) and item set $I$, ($|I|=N$). Each user $u$ is associated with one and only one target item $i_u$ (hold out from training set for $u$). A recommendation model $A$ is trained on the training set and would compute a personalized rank $R_u =A(i_u| u, I)$ for item $i_u$ among all the items $I$, and $\{R_u\}^{M}_{u=1}$ is called \textit{global rank} distribution. Intuitively speaking, $R_u$ is the prediction rank of item $i_u$ among all items, $R_u\in [1, N]$.
In contrast, one can also compute the other type personalized rank $r_u = A(i_u|u, I_u)$, $r_u\in [1, n]$, for item $i_u$ among sampled set $I_u=\{\ i\sim I\backslash i_u\}\cup \{i_u\}$, ($|I_u|=n$). In this case, $\{r_u\}_{u=1}^M$ is called \textit{sampled rank} distribution. Given $\{R_u\}^{M}_{u=1}$ or $\{r_u\}^{M}_{u=1}$, the (global or sampling-based) evaluation metric can be computed according to a specific metric function.

\subsection{Top-K Evaluation Metrics}\label{subsec:evaluation_metrics}
A metric function $\mathcal{F}$ maps its integer input (any rank $R_u$ or $r_u$) to a real-valued score. The aggregation of these scores over all users is called metric $T$:

\begin{equation}\label{eq:def_metric_global}
    \begin{split}
        T=\frac{1}{M}\sum\limits_{u=1}^M{\mathcal{F}(R_u)}
    \end{split}
\end{equation}

It is worth noting that a metric function $\mathcal{F}$ is an operation on some integer rank input and a metric $T$ is a corresponding aggregation of the output, which is a real value and represents the performance of a recommendation model. When we talk about $Recall$, $NDCG$ or $Recall@K$, $NDCG@K$, they can be both metric functions and metrics, depending on the scenario.
Similar to \citep{krichene@kdd20}, we define the simplified top-K metric functions $\mathcal{F}$ in the following way:

\begin{equation}\label{eq:def_metric_fn_1}
    \begin{split}
        \mathcal{F}_{Recall@K}(x)&=\delta(x\le K)\\
        \mathcal{F}_{NDCG@K}(x)&=\delta(x\le K)\cdot \frac{1}{\log_2(x+1)}\\
        \mathcal{F}_{AP@K}(x) &= \delta(x\le K)\cdot \frac{1}{x}
    \end{split}
\end{equation}
where $\delta(x) = 1$ if $x$ is True, $0$ otherwise. 
And the corresponding (global) metrics $T$ for a given $\{R_u\}_{u=1}^M$  are:
\begin{equation}\label{eq:def_metric_1}
    \begin{split}
        T_{Recall@K}&=\frac{1}{M}\sum\limits_{u=1}^M{ \mathcal{F}_{Recall@K}{(R_u)}}= \frac{1}{M}\sum\limits_{u=1}^M{\delta(R_u\le K)}\\
        T_{NDCG@K}&=\frac{1}{M}\sum\limits_{u=1}^M{ \mathcal{F}_{NDCG@K}{(R_u)}}= \frac{1}{M}\sum\limits_{u=1}^M{\delta(R_u\le K)\cdot \frac{1}{\log_2(R_u+1)}}\\
        T_{AP@K} 
        &=\frac{1}{M}\sum\limits_{u=1}^M{ \mathcal{F}_{AP@K}{(R_u)}}= 
        \frac{1}{M}\sum\limits_{u=1}^M{ \delta(R_u\le K)\cdot \frac{1}{R_u}}
    \end{split}
\end{equation}
\subsection{Sampling-based Top-K Metrics Estimation Problem}
Its also a common choice \citep{Koren08, CremonesiKT@10,he2017neural, ebesu2018collaborative,HuSWY18,krichene2018efficient,wang2018explainable,YangBGHE18,YangCXWB18} that use sampling-based top-K metrics to evaluate recommendation models, denoted as $T^S$ in general:
\begin{equation}\label{eq:def_metric_sample}
    \begin{split}
        T^S \triangleq \frac{1}{M} \sum\limits_{u=1}^{M}\delta ({r_u\le K})\cdot \mathcal{F}(r_u)
    \end{split}
\end{equation}

It's obvious that $r_u$ and $R_u$ differ substantially. For instance, $r_u \in [1,n]$ whereas $R_u \in [1,N]$. Therefore, for the same $K$, the sampling-based top-K metric $T^S$ and the global top-K metric $T$ correspond to distinct measures (no direct relationship): $T  \neq T^S$ ($T_{Recall@K} \neq T^S_{Recall@K}$ for example). 
This problem is highlighted  in~\cite{krichene@kdd20, rendle2019evaluation}, referring to these two metrics being {\em inconsistent}.
From the perspective of statistical inference, the basic sampling-based top-$K$ metric $T^S$  is not a reasonable or good {\em estimator} ~\cite{theorypoint} of $T$.

We consider one fundamental question for item sampling-based evaluation: \textbf{Given the sampled rank distribution $\{r_u\}_{u=1}^M$, (one can compute $T^S$ according to \cref{eq:def_metric_sample} naturally), what would be the estimation value of $T$ defined in \cref{eq:def_metric_global}?}

\subsection{Efforts in Metric Estimation}
There is one recent work studying the general metric estimation problem using the item sampling metrics. Specifically, given the sampling ranked results in the test data set, $\{r_u\}^{M}_{u=1}$, how to infer/approximate the $T$, without the knowledge $\{R_u\}_{u=1}^M$? 

\noindent{\bf Krichene and Rendle's approaches:} In ~\cite{krichene@kdd20}, Krichene and Rendle develop a discrete function $\widehat{\mathcal{F}}(r)$ so that: 
\begin{small}
\begin{equation}
\begin{split}
    T&=\frac{1}{M}\sum_{u=1}^M \delta(R_u\le K)\cdot \mathcal{F}(R_u) \\
    &\approx \frac{1}{M}\sum_{u=1}^M \widehat{\mathcal F}(r_u) = \sum_{r=1}^n \tilde{P}(r) \widehat{\mathcal F}(r)=\widehat{T}
    \end{split}
\end{equation}
\end{small}
where $\tilde{P}(r)$ is the empirical rank distribution on the sampling data (~\cref{tab:notations}). They have proposed a few estimators based on this idea, including estimators that use the unbiased rank estimators,  minimize bias with monotonicity constraint ($CLS$), and utilize Bias-Variance ($BV$) tradeoff. Their study shows that only the last one ($BV$) is competitive ~\cite{krichene@kdd20}. They derive a solution based on the Bias-Variance ($BV$) tradeoff ~\cite{krichene@kdd20}:
\begin{equation} \label{eq:leastsquare}
\widehat{\mathcal{F}} = \Big((1.0-\gamma)A^TA+\gamma \text{diag}(\pmb{c})  \Big)^{-1}A^T\pmb{b}
\end{equation}
where 
\begin{equation}\label{eq:BV_param}
    \begin{split}
        &A\in \mathbb{R}^{N\times n},\quad A_{R,r} = \sqrt{P(R)}P(r|R)\\
        &\pmb{b}\in \mathbb{R}^N,\quad b_{R} = \sqrt{P(R)}\mathcal{f}(R)\\
        &\pmb{c}\in \mathbb{R}^n, \quad c_{r} = \sum\limits_R^{N}{P(R)P(r|R)}
    \end{split}
\end{equation}

\section{Top-K Recall Metric Estimation via Mapping Function}\label{sec:map}

\subsection{Overview}

\begin{figure}[htbp]
\includegraphics[width=\textwidth]{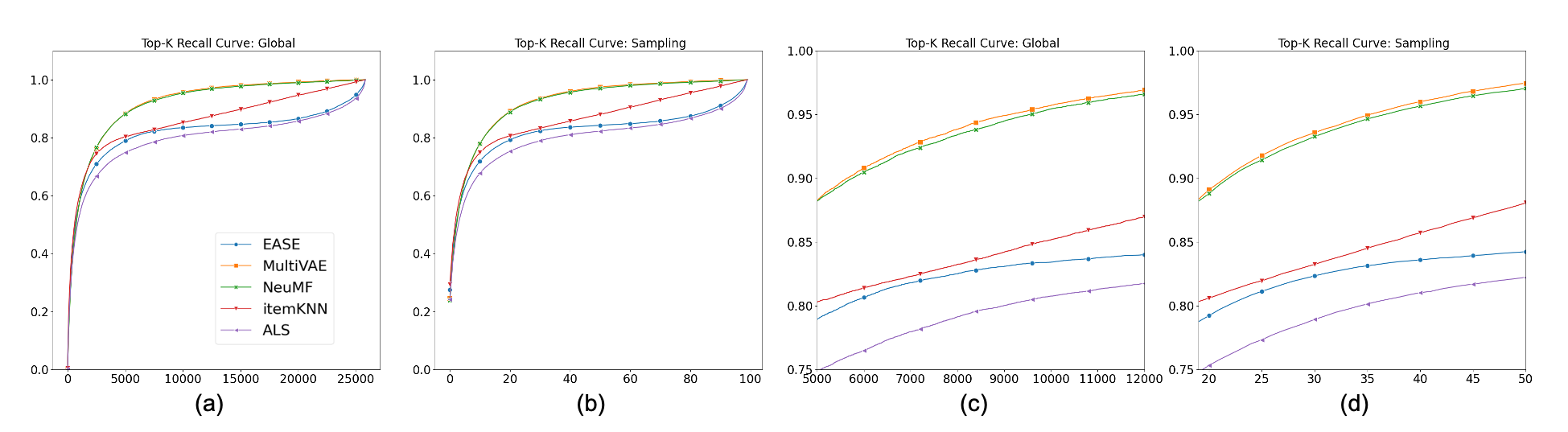}
\caption{Global vs Sampling Top-$K$ Hit-Ratio on $yelp$ dataset.
To display the details clearly, we zoom in (a) the global recall curve and (b) the sampling recall curve at different range scales, to $(c)$ and $(d)$ respectively. Comparing two figures, we can easily conclude that sampling evaluation maintains the same curve trend as global evaluation for different algorithms even at a small error range.}
\label{figure:yelp:correspondence}
\end{figure}

In this section, we would like to focus on one specific metric - "$Recall@K$". At a high level, we shall show that there is an approximately linear function $f(k)$ enables $T^S_{Recall@K}\approx T_{Recall@f(K)}$. For instance, in \cref{figure:yelp:correspondence}. Intuition: \cref{figure:yelp:correspondence} (a) and (b) plot two Recalll curves and they look similar in the whole trend. This suggests the mapping relation between the global Recall metric and sampling metric.

Let us consider the global top-$K$ Recall (or Hit-Ratio) metric (rewrite the \Cref{eq:def_metric_1}): 
\begin{equation}\label{eq:def_T_Recall@K}
\begin{split}
    T_{Recall@K} = \frac{1}{M}\sum\limits_{u = 1}^{M}{\delta(R_u \le K)} = \sum\limits_{R = 1}^{N}{\tilde{P}(R) \cdot \delta(R\le K)}
    \end{split}
\end{equation}
where $\tilde{P}(R)$ is the frequency of users with item $i_u$ rank in position $R$, also known as empirical global rank distribution:  
\begin{small}
\begin{equation}\label{eq: PR_empirical_def}
    \tilde{P}(R) = \frac{1}{M}\sum_{u = 1}^{M}{\delta (R_u = R)}
\end{equation}
\end{small}

Next, let us revisit the top-$k$ Recall (Hit-Ratio) under sampling in \Cref{eq:def_metric_sample}. For a given user $u$ and the relevant item $i_u$, we first sample $n - 1$ items from the entire set of items $I$, forming the subset $I_u$ (including $i_u$). The relative rank of $i_u$ among $I_n$ be denoted as $r_u = A(i_u|u, I_u)$. Note that $r_u$ is a random variable depending on the sampling set $I_u$. 

Given this, the sampling top-$k$ Recall metric can be written as: 
\begin{equation}\label{eq:def_TS_Recall@K}
    T^S_{Recall@K}=\frac{1}{M} \sum\limits_{u=1}^{M}{Z_u}, \ \ Z_u \sim Bernoulli(p_u=Pr(r_u\le K))
\end{equation}

where, $Z_u$ is a random variable for each user $u$, and follows a Bernoulli distribution with probability $p_u=Pr(r_u\le k)$. 

Now, recall that we are trying to study the relation between $T^S_{Recall@K}$ and $T_{Recall@K}$.

\subsection{Statistical View of Sampling Process}

We note that the population sum $\sum_{u=1}^M Z_u$ is a Poisson binomial distributed variable (a sum of $M$ independent Bernoulli distributed variables). Its mean and variance will simply be sums of the mean and variance of the $n$ Bernoulli distributions:

\begin{equation*}
    \mu = \sum_{u=1}^M{p_u}, \ \ \ \ 
    \sigma^2 = \sum_{u=1}^M p_u (1-p_u)
\end{equation*}

Given this, the expectation and variance of $T^S_{Recall@K}$ in \Cref{eq:def_TS_Recall@K}:

\begin{equation}\label{eq: E_TS_Recall@K_mean}
\begin{split}
    &\mathbb{E}[T^S_{Recall@K}] = \frac{1}{M} \sum\limits_{u=1}^{M}{p_u} 
    = \sum\limits_{R = 1}^{N}{\tilde{P}(R) \cdot P(R)}  
\end{split}
\end{equation}

\begin{equation}\label{eq: E_TS_Recall@K_var}
\begin{split}
    &Var[T^S_{Recall@K}] = \frac{1}{M^2} 
    \sum\limits_{u=1}^{M}{p_u(1-p_u)} 
    =\frac{1}{M} \sum\limits_{R = 1}^{N}{\tilde{P}(R) \cdot P(R)(1-P(R))}  
\end{split}
\end{equation}

where $P(R)$ is the probability for users who are in the same group ($R_u = R$), share the same $p^u$.

To define $p_u=Pr(r_u \le K)$ precisely, let us consider the two commonly used types of sampling strategy (with and without replacements). 

\noindent{\bf Sampling with replacement (Binomial Distribution):}

For a given user $u$, let $X_u$ denote the number of sampled items that are ranked in front of relevant item $i_u$: 

\begin{equation*}\label{eq:bern_dist}
    X_u = \sum\limits_{i=1}^{n-1}{X_{ui}}, \quad X_{ui} \sim Bernoulli(b_u = \frac{R_u - 1}{N-1}) 
\end{equation*}

where  $X_{ui}$ is a Bernoulli random variable for each sampled item $i$: $X_{ui}=1$ if item $i$ has rank range in $[1, R_u -1]$ ($b_u$ is the corresponding probability) and $X_{ui}=0$ if $i$ is located in $[R_u + 1, N]$. Thus, $X_u$ follows binomial distribution: 

\begin{equation}\label{eq:binom_dist}
\begin{split}
    \quad X_u \sim Binomial(n-1, b_u=\frac{R_u - 1}{N-1}) 
    \end{split}
\end{equation}

And the random variable $r_u=X_u +1$, and we have 

\begin{equation*}
\begin{split}
    p_u & =CDF(K;n-1,b_u)=Pr(r_u\le K) \\
    & = \begin{cases}
    \sum\limits_{l=0}^{k-1}{    {\binom{n-1}{l}} b_{u}^{l} (1-b_{u})^{n-1-l}    } &, R_u \geq K\\
    1 &, R_u < K
    \end{cases}
    \end{split}
\end{equation*}

\noindent{\bf Sampling without replacement (Hypergeometric Distribution):} 
If we sample $n-1$ items from the total $N-1$ items without replacement, and the total number of successful cases is $R_u-1$, then let $X_u$ be the random variable for the number of items appearing in front of relevant item $i_u$ ($r_u=X_u+1$):

\begin{equation*}
    X^u\sim Hypergeometric(N-1, R_u-1, n-1)
\end{equation*}

 \begin{equation*}
 \begin{split}
    p_u & =CDF(K;N-1, R_u-1, n-1)=Pr(r_u \le K) \\
    & =\begin{cases}
    \sum\limits_{l=0}^{k-1}{\frac{{\binom{R_u-1}{l}}{\binom{N - R_u}{n - 1 - l}} }{{\binom{N-1}{n -1}} }} &, R_u \geq K\\
    1 &, R_u < K
    \end{cases}
    \end{split}
\end{equation*}

It is well-known that, under certain conditions, the hypergeometric distribution can be approximated by a binomial distribution. We will focus on using binomial distribution for analysis, and we will validate the results of hypergeometric distribution experimentally.

\subsection{A Functional View of $T_{Recall@K}$ and $T^S_{Recall@K}$}

To better understand the relationship between $T_{Recall@K}$ (global top-K Recall or equivalently Hit-Ratio in \textit{leave-one-out} setting) and $T^S_{Recall@K}$ (the sampling version), it is beneficial to take a functional view of them. Let $\mathcal{R}$ be the random variable for the user's item rank, with probability mass function $Pr({\mathcal R}=R)$;  then, $T_{Recall@K}$ is simply the empirical cumulative distribution of ${\mathcal R}$ ($\widehat{Pr}$): 

\begin{equation}\label{eq:HR2}
T_{Recall@K} = \widehat{Pr}({\mathcal R} \leq K), \quad \tilde{P}(R) = \widehat{Pr}({\mathcal R}=R)
\end{equation}

For $T^S_{Recall@K}$, its direct meaning is more involved and will be examined below. For now, we note that $T^S_{Recall@K}$ is a function of $K$ varying from $1$ to $n$, where $n-1$ is the number of sampled items. 


\Cref{figure:yelp:correspondence}(a) displays the curves of functional fitting of empirical accumulative distribution $T_{Recall@K}$ (aka the global top-$K$ Hit- Ratio, varying $K$ from $1$ to $N=25815$), for $5$ representative recommendation algorithms ($3$ classical and $2$ deep learning methods), on the $yelp$ dataset. To observe the performance of these methods more closely, we zoom in the range $K$ from $5000$ to $12000$ in \Cref{figure:yelp:correspondence}(c). \Cref{figure:yelp:correspondence}(b) displays the curves of functional fitting of function $T^S_{Recall@K}$ (the sampling top-$k$ Hit-Ratio, varying $k$ from $1$ to $n=100$) with $n-1$ samples, under sampling with replacement, for the same $5$ representative recommendation algorithms on the same dataset. Similarly, we zoom in and highlight $K$ from $20$ to $50$ in \Cref{figure:yelp:correspondence}(d).

How can the sampling Hit-Ratio curves help to reflect what happened in the global curves? Before we consider the more detailed relationship between them, we introduce the following results:

\begin{theorem}[Sampling Theorem]
\label{theorem1}
Let us assume we have two global Recall curves (empirical cumulative distribution), $T^{(1)}_{Recall@K}$ and $T^{(2)}_{Recall@K}$, and assume one curve dominates the other one, i.e., $T^{(1)}_{Recall@K} \geq T^{(2)}_{Recall@K}$ for any $1\leq K \leq N$; then, for their corresponding sampling curve at any $k$ for any size of sampling, we have 
$$\mathbb{E}[T^{S, (1)}_{Recall@K}] \geq \mathbb{E}[T^{S, (2)}_{Recall@K}]$$ 
\end{theorem}

\begin{proof}
\sloppy Recall \Cref{eq: E_TS_Recall@K_mean}: $\mathbb{E}[T^{S}_{Recall@K}]=\sum_{R=1}^N \tilde{P}(R) \cdot P(R) =\sum_{u=1}^M Pr(r_u \leq K)$.  Let us assign each user $u$ the weight $Pr(r_u \leq K)$ for both curves, $T^{(1)}_{Recall@K}$ and $T^{(2)}_{Recall@K}$. Now, let us build a bipartite graph by connecting any $u$ in the $T^{(1)}_{Recall@K}$ with user $v$ in $T^{(2)}_{Recall@K}$, if $R_u \leq R_v$. We can then apply Hall's marriage theorem to claim there is a one-to-one matching between users in $T^{(1)}_{Recall@K}$ to users in $T^{(2)}_{Recall@K}$, such that $R_u \leq R_v$, and $Pr(r_u \leq K) \geq Pr(r_v \leq K)$.  (To see that, use the fact that $\sum_{R=1}^K \tilde{P}^{(1)}(R) \geq \sum_{R=1}^K \tilde{P}^{(2)}(R)$, where $\tilde{P}^{(1)}(R)$ and $\tilde{P}^{(2)}(R)$ are the empirical probability mass distributions of user-ranks, or equivalently, $\sum_{R=K}^N \tilde{P}^{(1)}(R) \leq \sum_{R=K}^N \tilde{P}^{(2)}(R)$. Thus, any subset in $T^{(1)}_{Recall@K}$ is always smaller than its neighbor set $N(T^{(1)}_{Recall@K})$ in $T^{(2)}_{Recall@K}$.  Given this, we can observe that the theorem holds.  
\end{proof}

The above theorem shows that, under the strict order of global Hit-Ratio curves (though it may be quite applicable for searching/evaluating better recommendation algorithms, such as in Figure~\ref{figure:yelp:correspondence}), sampling hit ratio curves can maintain such order. 

However, this theorem does not explain the stunning similarity, shapes and trends shared by the global and their corresponding sampling curves. Basically, the detailed performance differences among different recommendation algorithms seem to be well-preserved through sampling. 
However, unless $n\approx N$, $T^{S}_{Recall@K}$ does not correspond to $T_{Recall@K}$ (as in what is being studied by Rendle~\cite{rendle2019evaluation}).

Those observations hold on other datasets and recommendation algorithms as well, not only on this dataset. 
Thus, intuitively and through the above experiments, we may conjecture that it is the overall curve $T_{Recall@K}$ that is being approximated by $T^{S}_{Recall@K}$. Since these functions are defined on different domain sizes $N\ vs\ n$, we need to define such approximation carefully and rigorously.

\subsection{Mapping Function $f$}
To explain the similarity between the global and sampling top-$K$ Recall curves, we hypothesize that:

\textbf{there exists a function $f(K)$ such that the relation $T^S_{Recall@K} \approx T_{Recall@f(K)}$ holds for different ranking algorithms on the same dataset}. 

In a way, the sampling metric $T^S_{Recall@K}$ is like ``signal sampling''~\cite{rao2018signals}, where the global metrics between top $1$ to $N$ are sampled (and approximated) at only $f(1)<f(2)<\cdots<f(n)$ locations, which corresponds to $T^S_{Recall@K}$ ($k=1,2, \cdots, n$). In general, $f(k)\neq k$ (when $n << N$) (~\cite{rendle2019evaluation}).

In order to identify such a mapping function, let us take a look at the error between $T^S_{Recall@K}$ and $HR@ f(k)$ : 

\begin{equation}\label{eq:error}
    \begin{split}
    |T^S_{Recall@K} - T_{Recall@f(K)}|
 \le |T^S_{Recall@K} - \mathbb{E}[T^S_{Recall@K}]| + |\mathbb{E}[T^S_{Recall@K}] - T_{Recall@f(K)}|
    \end{split}
\end{equation}

Thanks to the Hoeffding's bound, we observe, 
\begin{small}
\begin{equation*}
    Pr(|T^S_{Recall@K} - \mathbb{E}[T^S_{Recall@K}]|\ge t) \le 2\exp(-2Mt^2)
\end{equation*}
\end{small}
This can be a rather tight bound, due to the large number of users in the population.
For example, if $M = 30K, t = 0.01$:
\begin{small}
\begin{equation*}
    Pr(|T^S_{Recall@K} - \mathbb{E}[T^S_{Recall@K}] \ge 0.01|\le 0.005
\end{equation*}
\end{small}
If we want to look at it more closely, we may use the law of large numbers and utilize the variance in equation ~\ref{eq: E_TS_Recall@K_var} for deducing the difference between $T^S_{Recall@K}$ and its expectation. 
Overall, for a large user population, the sampling top-$k$ Hit-Ratio will be tightly centered around its mean. Furthermore, if the user number is, indeed, small, an average of multiple sampling results can reduce the variance and error. In the publicly available datasets, we found that one set of samples is typically very close to the average of multiple runs. 

Given this, our problem is {\bf how to find the mapping function $f$, such as $|\mathbb{E}[T^S_{Recall@K}] - T_{Recall@f(K)}|$ can be minimized (ideally close to or equal to $0$).} Note that $f$ should work for all the $K$ (from $1$ to $n$), and it should be independent of algorithms on the same dataset.

\subsection{ Approximating Mapping Function $f$} 
\label{subsec: mapping_function}

\noindent{\bf Baseline:}
To start, we may consider the following naive mapping function.
We notice that for any $n$, 
\begin{equation*}
\begin{split}
    \mathbb{E}[X_u] &=(n-1) \cdot b_u= (n-1) \frac{R_u - 1}{N - 1}=\mathbb{E}[r_u]-1\\
    \mathbb{E}[r_u] & = 1 + (n-1) \frac{R_u - 1}{N - 1}
    \end{split}
\end{equation*}

When the sample set size $n$ is large, we simply use the indicator function $\delta(\mathbb{E}[r_u]) \leq K$ to approximate and replace $Pr(r_u \leq K)$. Thus,

\begin{equation}\label{eq:fk1}
\begin{split}
1 + (n-1)\frac{R_u-1}{N-1} &\le K\\
   R_u \leq \frac{K - 1}{n - 1}*(N-1) + 1 &\triangleq f(K)\\
   \end{split}
\end{equation}
To wrap to, this baseline function $f(K)$ enable us that: for any given $T^S_{Recall@K}$, we can obtain a approximation for $T_{Recall@K^\prime}$, where $K^\prime = f(K)$. This guarantees us that we can plot $T^S_{Recall@K}$ curve in the global range (1 to N) as well, which would help us directly observe the relation between $T^S_{Recall@K}$ and $T_{Recall@K}$, see \Cref{fig:unweighted_scalability}.

Since the indicator function, $\delta(\mathbb{E}[r_u] \leq K)$ is a rather crude estimation of the CDF of $r_u$ at $K$, this only serves as a baseline for our approximation of the mapping function $f$. 

\noindent{\bf Approximation Requirements:} Before we introduce more carefully designed approximations of the mapping function $f$, let us take a close look of the expectation of the sampling top-$K$ Recall $\mathbb{E}[T^S_{Recall@K}]$ and $T^S_{Recall@f(K)}$. \Cref{fig:unweighted_scalability} shows how the user empirical probability mass function $\tilde{P}(R)$ works with the step indicator function $\delta(R_u\leq f(K))$, and $b_u$ (assuming a hypergeometric distribution), to generate the global top-$K$ and sampling Recall.

\begin{figure}[h]
    \centering
        \includegraphics[width=0.7\linewidth]{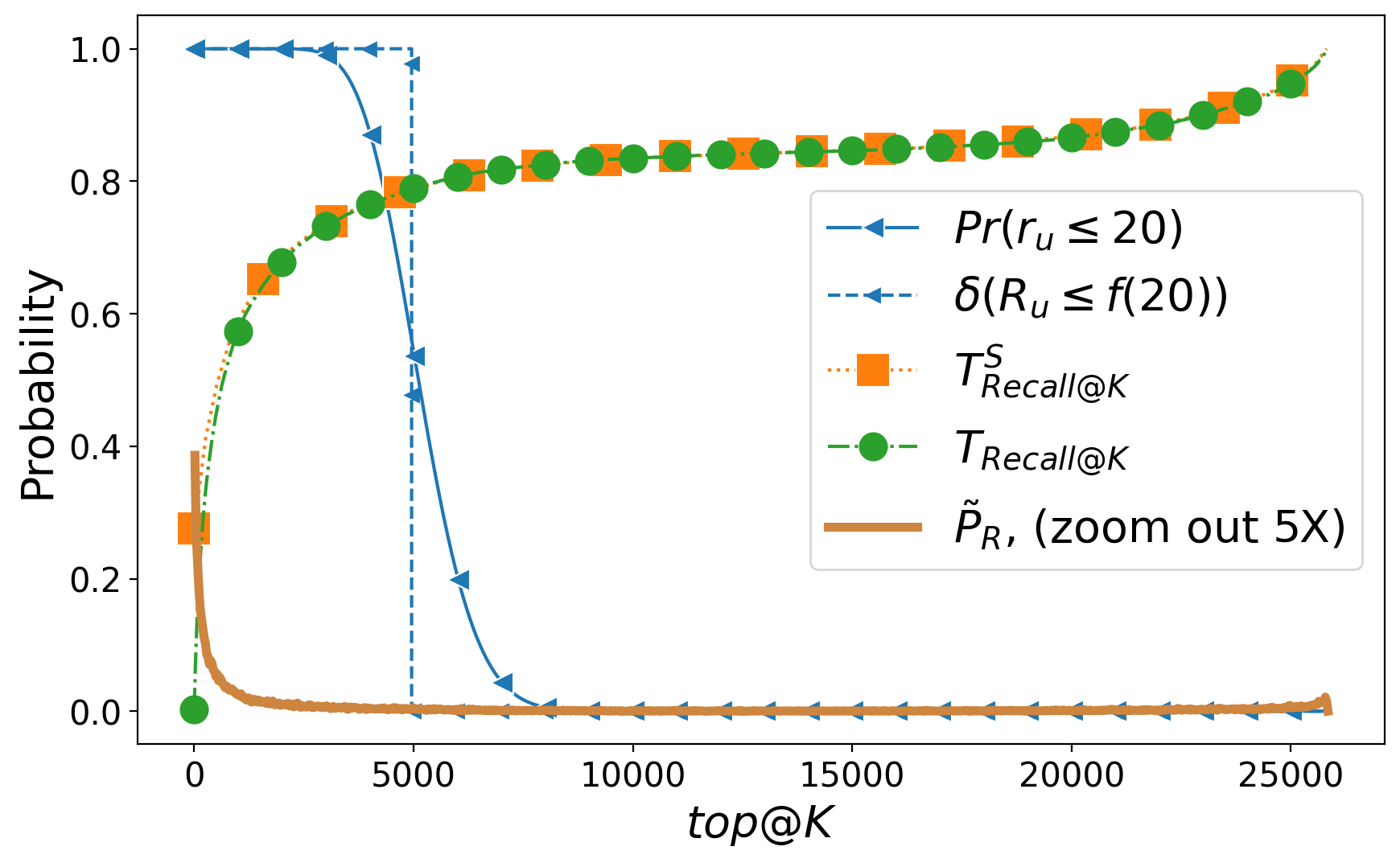}
    \caption{Curve Relationship of model $EASE$ on $yelp$ dataset. $T_{Recall@K}$ is the top-K global Recall curve;  $T^S_{Recall@K}$ is the sampling top-K Recall curve shown in global scale (by baseline \Cref{eq:fk1}); $\tilde{P}(R)(zoom out 5X)$ is the empirical user ranking distribution, where we make it 5 times larger for displaying purpose.}
    \label{fig:unweighted_scalability}
\end{figure}

We make the following observations (as well as requirements):

\begin{itemize}
    \item \textbf{Existence of mapping function $f$ for each individual $T_{Recall@K}$ curve:} Given any $K$, assuming $T_{Recall@f(K)}$ is a continuous cumulative distribution function (i.e., assuming that there is no jump/discontinuity on the CDF, and that $f(K)$ is a real value), then, there is $f(K)$ such that $T_{Recall@f(K)}=\mathbb{E}[T^S_{Recall@K}]$. 
In our problem setting, where $f(K)$ is integer-valued and ranges between $0$ and $N$, the best $f(K)$, theoretically, is 
\begin{equation*}
    f(K)=\arg\limits_f \min |T_{Recall@f(K)}-\mathbb{E}[T^S_{Recall@K}]|
\end{equation*}
\item \textbf{ Mapping function $f$ for different $T_{Recall@f(K)}$ curves:} Since our main purpose is for $T^S_{Recall@K}$ to be comparable across different recommendation algorithms, we prefer $f(K)$ to be the same for different $Recall$ curves (on the same dataset).
Thus, by comparing different $T^S_{Recall@K}$, we can infer  their corresponding Recall $T_{Recall}$ at the same $f(K)$ location. \Cref{figure:yelp:correspondence} and \Cref{fig:unweighted_scalability} show that the sampling $Recall$ curves are comparable with respect to their respective counterparts, and suggests that such a mapping function, indeed, may exist. 
\end{itemize}

But how does these the second requirement coexists with the first requirement of the minimal error of individual curves? We note that, for most of the recommendation algorithms, their overall $Recall$ curve $T_{Recall@K}$, and the empirical probability mass function $\tilde{P}(R)$ (~\Cref{eq: PR_empirical_def}), are actually fairly similar. 
From another viewpoint, if we allow individual curves to have different optimal $f(K)$, the difference (or shift) between them is rather small and does not affect the performance comparison between them, using the sampling curves $T^S_{Recall@K}$. In this section, we will focus on studying dataset-algorithm-independent mapping functions.



\subsection{Boundary Condition Approximation}
Consider that sampling with replacement, for any individual user, $X_u$ from \Cref{eq:binom_dist}, obeys binomial distribution. Apply the general case of bounded variables Hoeffding's inequality: 
\begin{equation*}
    Pr(|X_u - \mathbb{E}[X_u]|\ge t) \le 2e^{-\frac{2t^2}{n-1}}
\end{equation*}

since $r_u = X_u + 1$, and $\mathbb{E}[r_u] = \mathbb{E}[X_u]+1 = (n-1)b_u + 1$  :
 \begin{equation}\label{eq:ineq_lr0}
    \begin{cases}
    &Pr(r_u \ge (n-1)b_u + 1 + t)\le 2e^{-\frac{2t^2}{n-1}} \\
    &Pr(r_u \le(n-1)b_u +1 - t)\le 2e^{-\frac{2t^2}{n-1}}
    \end{cases}
\end{equation}

The above inequalities indicate that $r_u$ is restricted around its expectation within the range defined by $t$. The second term of error in equation \ref{eq:error} can be written as:
\begin{equation}\label{eq:bound1}
    \begin{split}
        \mathbb{E}[T^S_{Recall@K}]-T_{Recall@f(K)}
        = - \sum\limits_{R = 1}^{f(K)}\tilde{P}(R)\cdot Pr(r_R \ge k + 1)+ \sum\limits_{R = f(K)+ 1}^{N}\tilde{P}(R)\cdot Pr(r_R \le K)
    \end{split}
\end{equation}

where, $r_R=r_u$ for $R_u=R$. 
For some relatively large $t$ (compared to $\sqrt{n-1}$), the probability in \Cref{eq:ineq_lr0} can come extremely close to $0$. Based on this fact, if we would like to limit the first term $Pr(r_R\le K+1)$ to approach $0$, $K+1$ must be greater than $(n-1)b_u + 1+t$. And similar to the second term, we have:

 \begin{equation*}
    \begin{cases}
    r_u \geq K+1 \ge (n-1)\frac{R-1}{N-1} + 1 + t), & R = 1,\dots,f^{lower}(K) \\
   r_u \leq K \le(n-1)\frac{R-1}{N-1} +1 - t, & R =f^{upper}(k) +1,\dots,N
    \end{cases}
\end{equation*}

where $f^{lower}(K)$ and $f^{upper}(K)$ are the lower bound and upper bound for $f(K)$, respectively. Explicitly, 

 \begin{equation}\label{eq:bound_lr2}
    f^{lower}(K) \le (k-t)\cdot\frac{N-1}{n-1} +1, \quad 
   f^{upper}(K)\ge (k+t-1)\cdot\frac{N-1}{n-1}
\end{equation}

Given this, let the average of these two for $f$:
\begin{small}
\begin{equation}\label{eq:bound_mid}
     \boxed{f(K)= \lfloor\frac{f^{lower}(K)+f^{upper}(K)}{2} \rfloor=\lfloor (k-\frac{1}{2})\frac{N-1}{n-1}+\frac{1}{2} \rfloor}
\end{equation} 
\end{small}
Note that, although this formula appears similar to our baseline ~\Cref{eq:fk1}, the difference between them is actually pretty big ($\approx \frac{1}{2}\frac{N-1}{n-1}$).  
As we will show in the experimental results, this formula is remarkably effective in reducing the error $|\mathbb{E}[T^S_{Recall@K}]-T_{Recall@f(K)}|$. 

\subsection{Beta Distribution Approximation}

In this approach, we try to directly minimize $\mathbb{E}[T^S_{Recall@K}]-T_{Recall@f(K)}$, and this is equivalent to: 

\begin{equation}\label{eq:lr0}
\begin{split}
&\sum\limits_{R = 1}^{N}{\tilde{P}(R) \cdot \delta(R\le f(K)})= \sum\limits_{R = 1}^{N}{\tilde{P}(R)\cdot Pr(r_R\le K)}
\end{split}
\end{equation}

In order to get a closed-form solution of $f(K)$ from the above equation, we leverage the Beta distribution $Beta(a,1)$ to represent the user ranking distribution $\tilde{P}(R)$, inspired by ~\cite{LiMG@Linguistic}: 
\begin{equation*}
    \tilde{P}(R) = \frac{1}{\mathcal{B}(a, 1)}(\frac{R -1}{N-1})^ {a-1} \frac{1}{N-1}
\end{equation*}
where $a$ is a constant parameter and $\frac{1}{N-1}$ is the constant for discretized Beta distribution. Note that $\frac{R -1}{N-1}$ normalizes the user rank $R_u$ from $[1, N]$, to $[0, 1]$. Especially, when $a<1$, this distribution can represent exponential distribution, which can help provide fit for the $Recall$ distribution. \Cref{fig:beta_user} illustrates the Beta distribution fitting of $\tilde{P}(R)$.


\begin{figure}[h]
    \centering
    \includegraphics[ width=0.7\linewidth]{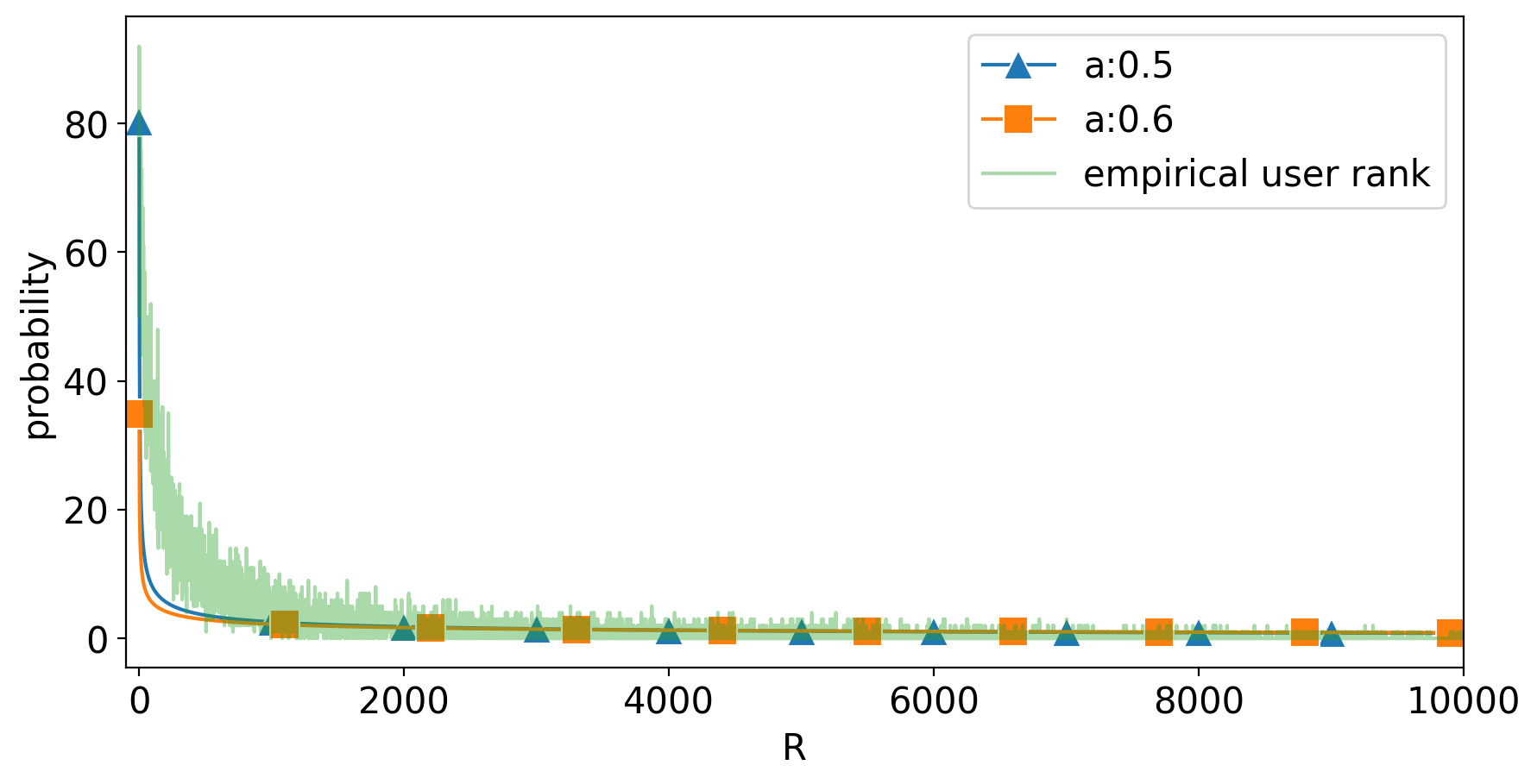}
    \caption{Beta distributions and empirical user rank distribution $\tilde{P}(R)$}
    \label{fig:beta_user}
\end{figure}

Let he left term of \Cref{eq:lr0} is denoted as $\mathcal{L}_k$:
\begin{equation*}\label{eq:left0}
\begin{split}
    \mathcal{L}_k&=  \sum\limits_{R = 1}^{N}{\tilde{P}(R) \cdot {\bf 1}_{R\le f(k)}}  = \sum\limits_{R = 1}^{f(k)}{\tilde{P}(R)}
    =\frac{1}{\mathcal{B}(a,1)}\sum\limits_{R = 1}^{f(k)}{(\frac{R -1}{N-1})^ {a-1}\cdot \frac{1}{N-1}} \\
    & = \frac{1}{\mathcal{B}(a,1)} \sum\limits_{x = 0}^{\frac{f(k) -1}{N -1} } {x^ {a-1}\cdot \Delta x} \quad \text{where}, x = \frac{R - 1}{N - 1}, \text{ and} \Delta x = \frac{1}{N-1}, \\
    &\approx \frac{1}{\mathcal{B}(a,1)}\int_{0}^{\frac{f(k) -1}{N -1} }x^{a-1}dx 
     = \frac{1}{a\mathcal{B}(a,1)}[\frac{f(k) -1}{N - 1}]^a
\end{split}
\end{equation*}

Considering sampling with replacement, then the right term of \Cref{eq:lr0} is denoted as: 

\begin{equation*} \label{eq:right0}
\begin{split}
&\mathcal{R}_k =  \sum\limits_{i = 0}^{k-1}{\binom{n-1}{i}\sum\limits_{R = 1}^N{\tilde{P}(R)(\frac{R-1}{N-1})^i (1-\frac{R-1}{N-1})^{n-i-1}} }
\end{split}
\end{equation*}

Calculate the difference:

\begin{equation*}\label{eq:right1}
\begin{split}
&\mathcal{R}_{k+1} -\mathcal{R}_k =  {\binom{n-1}{k}\sum\limits_{R = 1}^N{\tilde{P}(R)(\frac{R-1}{N-1})^k (1-\frac{R-1}{N-1})^{n-1-k}} }\\
&\approx \binom{n-1}{k}\frac{1}{\mathcal{B}(a, 1)}\int_{x = 0}^{1} { x^{a+k-1} (1-x)^{n-1-k} dx} \\ 
& = \binom{n-1}{k}\frac{1}{\mathcal{B}(a, 1)}\mathcal{B}(a+k, n-k)
 =\frac{1}{\mathcal{B}(a, 1)} \frac{\Gamma(n)}{\Gamma(n+a)}\frac{\Gamma(k+a)}{\Gamma(k+1)}
\end{split}
\end{equation*}

Based on above equations: $\mathcal{L}_{k+1}-\mathcal{L}_{k}\approx \mathcal{R}_{k+1}-\mathcal{R}_{k}$, we have 
(we denote the mapping function as $f(k;a)$ for parameter $a$).

\begin{equation}\label{eq:l2r}
 \begin{split}
        &[f(k+1;a)-1]^a-[f(k;a)-1]^a \\
        &=a[N-1]^a\binom{n-1}{k}\mathcal{B}(a+k, n-k)
    \end{split}
\end{equation}

Then we have the following recurrent formula:
\begin{equation}\label{eq:fk}
\boxed{
\begin{split}
    f(k+1;a)=&\Big[ a [N-1]^a \binom{n-1}{k}\mathcal{B}(a+k, n-k)\\
    &+ [f(k;a)-1]^a\Big]^{1/a} + 1
    \end{split}
}
\end{equation}
And $f(1)$ is given by $\mathcal{L}_1 = \mathcal{R}_1$:
\begin{equation}\label{eq:f1}
f(1;a) = (N-1)[a\mathcal{B}(a,n)]^{1/a} + 1 
\end{equation}

In \cref{app_subsec:property}, we enumerate a list of interesting properties of this recurrent formula of $f$ based on Beta distribution.

\begin{figure}
\centering
\begin{minipage}{.45\textwidth}
\centering
 \includegraphics[width =0.9\linewidth]{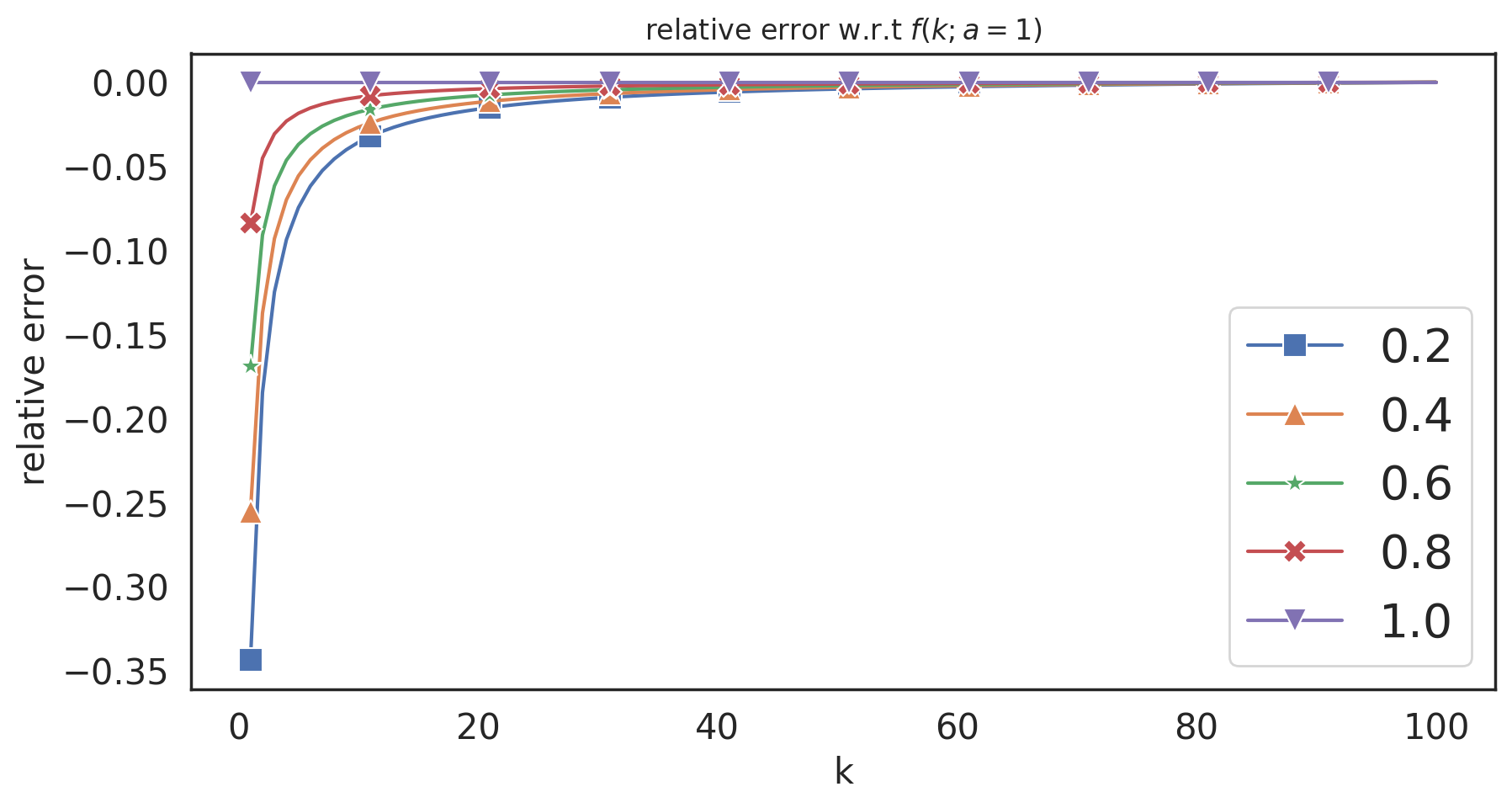}
    \captionof{figure}{Relative error w.r.t. $f(k;a=1)$. example on yelp dataset, n= 100.}
    \label{fig:rela_a1}
\end{minipage}%
\hfill
\begin{minipage}{.55\textwidth}
  \centering
  \includegraphics[width=0.95\linewidth]{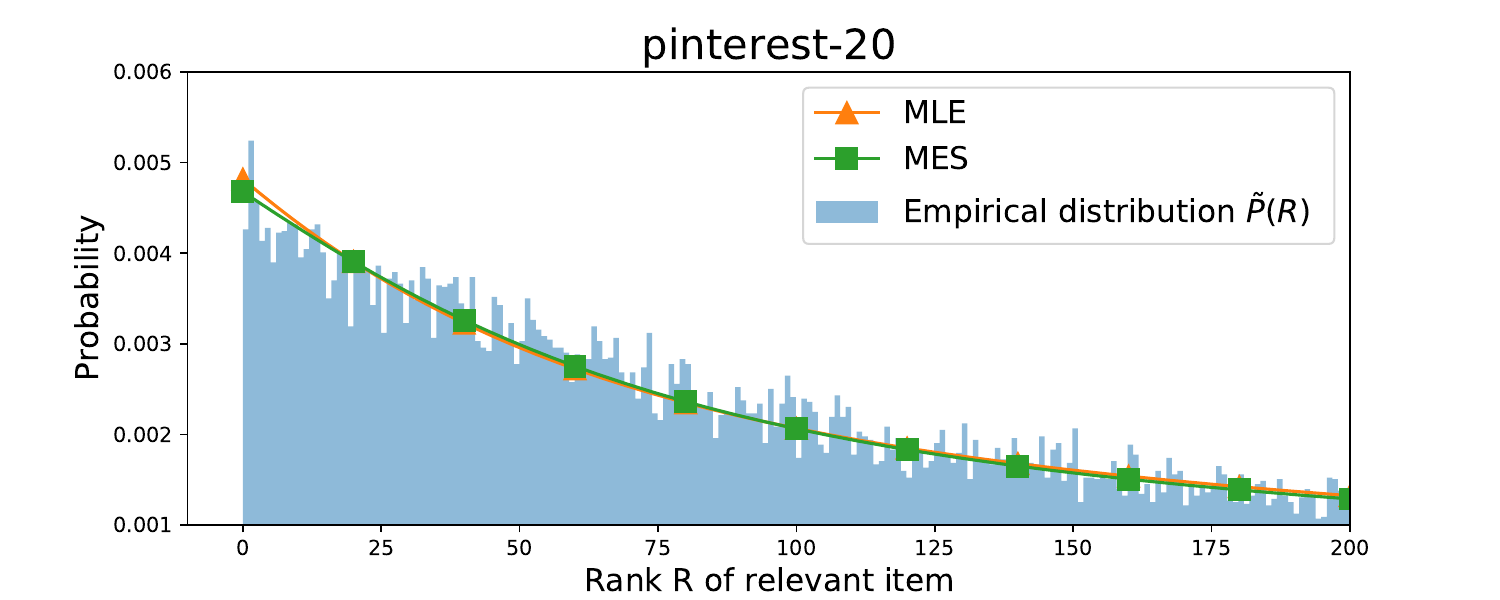}
  \captionof{figure}{Learning Empirical Rank Distribution $P(R)$}
 \label{fig:P(R)}
\end{minipage}
\end{figure}

\section{Top-K Metric Estimation Foundation via Global Rank Distribution Learning}\label{sec:fix_estimator_learning_rank}

\subsection{Learning Empirical Rank Distribution Problem}\label{subsec:learning_empirical}

In this section, our new proposed approach is based on the following observation: 

\begin{equation}\label{eq:metric_observe}
T=\frac{1}{M}\sum_{u=1}^M \mathcal{F}(R_u)=\sum_{R=1}^K \tilde{P}(R) \cdot \mathcal{F}(R) 
\end{equation}
Again, $T$ is any metric to quantify the quality of a recommendation model and $\mathcal{F}$ is the corresponding specific metric function that we discussed in \Cref{subsec:evaluation_metrics}.

Thus, if we can estimate $\widehat{P}(R)\approx P(R) \approx \tilde{P}(R)$, then we can derive the metric estimator as 

\begin{equation}\label{eq:metric_from_rank}
\widehat{T}=\sum_{R=1}^K \widehat{P}(R) \cdot \mathcal{F}(R)
\end{equation}

Given this, we introduce the problem of learning and estimating the empirical rank distribution $\{\tilde{P}(R)\}_{R=1}^N$ based on sampling $\{r_u\}_{r=1}^M$. To our best knowledge, this  problem has not been formally and explicitly studied before for sampling-based recommendation evaluation. 

We note that the importance of the problem is two-fold. 
On one side, the learned empirical rank distributions can directly provide estimators for any metric $T$; on the other side, since this question is closely related to the underlying mechanism of sampling for recommendation, tackling it can help better understand the power of sampling and help resolve the questions as to if and how we should use sampling for evaluating recommendation. 

Furthermore, since metric $T$ is the linear function of $\{\tilde{P}(R)\}_{R=1}^K$, the statistical properties of estimator $\widehat{P}(R)$ can be nicely preserved by $\widehat{T}$~\cite{theorypoint}. In addition, this approach can be considered as metric-independent: We only need to estimate the empirical rank distribution $\tilde{P}(R)$ once; then we can utilize it for estimating all the top-$K$ evaluation metrics $T$ (including for different $K$) based on \Cref{eq:metric_from_rank}. 

Finally, we note that we can utilize the $BV$ estimator \citep{krichene@kdd20} to estimate $P(R)$ as follows: Let $\widehat{\mathcal{F}}_{BV}(R)$ be the Recall metric function estimator from \citep{krichene@kdd20}. Then we have 

\begin{equation} \label{eq:BV_recall_pr}
\begin{split}
   \widehat{P}(R) &=\widehat{\mathcal{F}}_{BV}(R)-\widehat{\mathcal{F}}_{BV}(R-1) \\
   &= (\tilde{P}(r))_{r=1}^n \Big((1.0-\gamma)A^TA+\gamma \text{diag}(\pmb{c})  \Big)^{-1}A^T\cdot\pmb{b}_R
\end{split}
\end{equation}

where $\widehat{\mathcal{F}}_{BV}(R)$ is the $BV$ estimator for the $Recall@R$ metric function, $(\tilde{P}(r))_{r=1}^n$ is the row vector of empirical rank distribution over the sampling data, and $\pmb{b}_R$ has the $R$-th element as $b_R$ (\cref{eq:BV_param}) and other elements as $0$. 
We consider this as our baseline for learning the empirical rank distribution. 

In this section, we will introduce a list of estimators for the empirical rank distribution $\{P(R)\}_{R=1}^N$ based on sampling ranked data: $\{r_u\}_{r=1}^M$. 
\Cref{fig:P(R)} illustrates the different approaches of learning the empirical rank distribution $P(R)$, including the Maximal Likelihood Estimation (MLE), and the Maximal-Entropy-based approach (MES),  for $R \leq 200$ on \textit{pinterest-20}.


\subsection{Sampling Rank Distribution: Mixtures of Binomial Distributions}
Let us consider the sampling with replacement scheme (the results can be extended to sampling without replacement). Now, assume an item $i$ is ranked $R$ in the entire set of items $I$. Then there are $R-1$ items whose rank is higher than item $i$ and $N-R$ items whose rank is lower than $i$. Under the (uniform) sampling (sampling with replacement), we have $\theta \coloneqq \frac{R-1}{N-1}$ probability to pick up an item with higher rank than $R$. 
Let $x$ be the number of irrelevant items ranked in front of the relevant one, and $x=r-1$. Thus, 
the rank $r-1$ under sampling follows a binomial distribution: $r-1 \sim Binomial(n-1, \theta)$, and the conditional rank distribution $P(r|R)$ is
\begin{equation}
    \begin{split}
        &P(r|R)=Binomial(r-1;n-1,\theta)= \binom{n-1}{r-1}\theta^{r-1}(1-\theta)^{n-r} 
    \end{split}
\end{equation}

Given this, an interesting observation is that the sampling ranked data $\{r_u\}_{r=1}^M$ can be directly modeled as a mixture of binomial distributions. 
Let $\pmb{\Theta}=(\theta_1\dots,\theta_R,\dots,\theta_N)^T$ where 
\begin{equation}
    \theta_R \coloneqq \frac{R-1}{N-1},\quad R = 1, \dots, N
\end{equation}
Let the empirical rank distribution $\pmb{\tilde{P}}=\{\tilde{P}(R)\}_{R=1}^N$, then the sampling rank follows the distribution $P(r|\pmb{\tilde{P}})=$

\begin{equation}
\begin{split}
&\sum_{R=1}^N P(r|R)\cdot P(R)=\sum_{R=1}^N Bin(r-1; n-1, \theta_R) \cdot P(R)  \\ 
=&\sum_{R=1}^N  P(R) \binom{n-1}{r-1} \big(\frac{R-1}{N-1}\big)^{r-1}\big(1-\frac{R-1}{N-1}\big)^{n-r} 
\end{split}
\end{equation}

Thus, $P(R)$ can be considered as the parameters for the mixture of binomial distributions.

\subsection{Maximum Likelihood Estimation (MLE)}

The basic approach to learn the parameters of the mixture of binomial distributions ($MB$) given $\{r_u\}_{u=1}^M$ is based on maximal likelihood estimation (MLE). 
Let $\pmb{\Pi} = (\pi_1,\dots,\pi_R,\dots,\pi_N)^T$ be the parameters of the mixture of binomial distributions. Then we have $p(r_u|\pmb{\Pi})=\sum_{R=1}^N \pi_R \cdotp(r_u|\theta_R)$, where 
$p(r_u|\theta_R)=Binomial(r_u-1; n-1, \theta_R)$. 

Then MLE aims to find the particular $\pmb{\Pi}$, which maximizes the log-likelihood: 

\begin{equation}\label{eq:mb_mle}
    \log\mathcal{L} = \sum\limits_{u=1}^{M}{\log{p(r_u|\pmb{\Pi})}} = \sum\limits_{u=1}^{M}{ \log{\sum_{R=1}^N \pi_R p(r_u|\theta_R)}}
\end{equation}

By leveraging EM algorithm (details see \Cref{app_sec:fix-em}):
\begin{equation}\label{eq:em_pi_new}
    \pi^{new}_R = \frac{1}{M} \sum\limits_{u=1}^M \frac{{\pi}^{old}_R p(r_u|\theta_R)}{\sum\limits_{j=1}^{N}{ \pi^{old}_j p(r_u|\theta_j)}}
\end{equation} 

When \Cref{eq:em_pi_new} converges, we obtain $\pmb{\Pi}^*$ and use it to estimate $\pmb{P}$, i.e., $\widehat{P}(R)=\pi^*_R$. 
Then, we can use $\widehat{P}(R)$ in \Cref{eq:metric_from_rank} to estimate the desired metric $T$ in \Cref{eq:metric_observe}. 

\subsubsection{Speedup and Time Complexity}

To speedup the computation, we can further rewrite the updated formula \Cref{eq:em_pi_new} as 
\begin{equation}\label{eq:em_pi_speedup}
    \pi^{new}_R = \sum\limits_{r=1}^n \tilde{P}(r) \frac{{\pi}^{old}_R \cdot p(r|\theta_R)}{\sum\limits_{j=1}^{N}{\pi^{old}_j\cdot p(r|\theta_j)}}
\end{equation} 
where $\tilde{P}(r)=\frac{1}{M} \sum_{u=1}^M {\delta(r_u=r)}$ is the empirical rank distribution on the sampling data. 
Thus the time complexity improves to $O(kNn)$ (from $O(kNM)$ using \Cref{eq:em_pi_new}), where $k$ is the iteration number. This is faster than the least squares solver for the $BV$ estimator (\Cref{eq:leastsquare}) ~\cite{krichene@kdd20}, which is at least $O(n^2N)$.
Furthermore, we note $\widehat{P}(R)$ can be used for any metric $T$ for the same algorithm, whereas $BV$ estimator has to be performed for each metric $T$ separately.

\subsubsection{Weighted MLE}
If we are particularly interested in $\pi_R$ ( also known as $P(R)$) when $R$ is very small (such as $R<10$), then we can utilize the weighted MLE to provide more focus on those ranks. This is done by putting more weight on the sampling rank observation $r_u$ when $r_u$ is small. Specifically, the weighted MLE aims to find the $\pmb{\Pi}$, which maximizes the weighted log-likelihood: 

\begin{small}
\begin{equation}\label{eq:mb_mle}
    \log\mathcal{L} =  \sum\limits_{u=1}^{M}{w(r_u)\cdot \log{p(r_u|\pmb{\Pi})}} = \sum\limits_{u=1}^{M}{w(r_u) \cdot \log{\sum_{R=1}^N \pi_R \cdotp(r_u|\theta_R)}}
\end{equation}
\end{small}
where $w(r_u)$ is the weight for rank $r_u$. Note that the typical MLE (without weight) is the special case of \Cref{eq:mb_mle} ($w(r_u)=1$).

For weighted MLE, its updated formula is
\begin{equation}\label{eq:em_pi_speedup}
    \pi^{new}_R = \sum\limits_{r=1}^n  \frac{\tilde{P}(r) w(r)}{\sum_{r=1}^n \tilde{P}(r) w(r)}  \frac{{\pi}^{old}_R p(r|\theta_R)}{\sum\limits_{j=1}^{N}{\pi^{old}_j p(r|\theta_j)}}
\end{equation} 

For the weight $w(r)$, we can utilize any decay function (as $r$ becomes bigger, than $w(r)$ will reduce). We have experimented with various decay functions and found that the important/metric functions, such as $AP$ and $NDCG$,  $w(r)=\mathcal{F}_{AP}(r/C)$ and  $w(r)=\mathcal{F}_{NDCG}(r/C)$ ($C>1$ is a constant to help reduce the decade rate), obtain good and competitive results. We will provide their results in the experimental evaluation section.

\subsection{Maximal Entropy with Minimal Distribution Bias (ME)}

Another commonly used approach for estimating a (discrete) probability distribution is based on the principal of maximal entropy~\cite{elementsinfo}. 
Assume a random variable $x$ takes values in $(x_1, x_2, \cdots, x_n)$ with pmf: $p(x_1), p(x_2), \cdots, p(x_n)$. Typically, given a list of (linear) constraints in the form of $\sum_{i=1}^n p(x_i) f_k(x_i) \geq F_k$ ($k=1, \cdots m$), together with the equality constraint ($\sum_{i=1}^n p(x_i)=1$), it aims to maximize its entropy: 
\begin{equation}
    H(p)=-\sum_{i=1}^{n} p(x_i) \log p(x_i)
\end{equation}

In our problem, let the random variable $\mathcal{R}$ take on rank from $1$ to $N$. Assume its pmf is  $\pmb{\Pi}=(\pi_1,\dots,\pi_R,\dots,\pi_N)$, and the only immediate inequality constraint is $\pi_R \geq 0$ besides $\sum_{R=1}^N \pi_R=1$. 
Now, to further constrain $\pmb{\pi}$, we need to consider how they reflect and manifest on the observation data $\{r_u\}^M_{u=1}$. The natural solution is to simply utilize the (log) likelihood. However, combining them together leads to a rather complex non-convex optimization problem which will complicate the EM-solver. 

In this section, we introduce a method (to constrain the maximal entropy) that utilizes the squared distance between the learned rank probability (based on $\pmb{\Pi}$) and the empirical rank probability in the sampling data

\begin{equation}\label{eq:entropy}
\begin{split}
    \mathcal{E}&= \frac{1}{M} \sum\limits_{R = 1}^{M}{\Big(p(r_u|\pmb{\Pi}) - \tilde{P}(r_u)  \Big)^2 } \\
    &=\sum\limits_{r = 1}^{n}{\tilde{P}(r)\Big(\sum\limits_{R = 1}^{N}{P(r|R)\pi_R} - \tilde{P}(r)  \Big)^2 }
    \end{split}
\end{equation}
Again, $\tilde{P}(r)$ is the empirical rank distribution in the sampling data. 
Note that $\mathcal{E}$ can be considered to be derived from the  log-likelihood of independent Gaussian distributions if we assume the error term $p(r_u|\pmb{\Pi}) - \tilde{P}(r_u)$ follows the Gaussian distribution.  

Given this, we seek to solve the following optimization problem: 
\begin{equation} \label{eq:MEE}
     \pmb{\Pi}=\arg\max_{\pmb{\Pi}} \eta\cdot H(\pmb{\pi})-\mathcal{E}  
\end{equation}
where $\eta$ is hyperparameter and with constraints: 
\begin{equation*}
    \pi_R \ge 0 \ (1 \leq R \leq N), \quad \sum_R{\pi_R}=1
\end{equation*}
Note that this objective can also be considered as adding an entropy regularizer for the log-likelihood. The objective function: $\eta\cdot H(\pmb{\pi})-\mathcal{E}$ is concave (or its negative is convex). This can be easily observed as both negative of entropy and sum of squared errors are convex function. Given this, we can employ available convex optimization solvers~\cite{CVX} to identify the optimization solution. 
Thus, we have the estimator  $\widehat{P}(R)=\pi^*_R$, where $\Pi^*$ is the optimal solution for  \Cref{eq:MEE}.

\section{Top-K Metric Estimation Foundation via Optimized Bias-Variance Estimator}\label{sec:fix_estimator_learning_metric}

In this section, we introduce a new estimator which aims to directly minimize the expected errors between the item sampling-based top-$K$ metrics and the global top-$K$ metrics. Here, we consider a similar strategy as ~\cite{krichene@kdd20} though our objective function is different and aims to explicitly minimize the expected error. 
We aim to search for a {\em sampled metric} $\widehat{\mathcal F}(r)$  to approach $\widehat{T}\approx T$: 

\begin{equation*}
\begin{split}
    \widehat{T}&=\sum_{r=1}^n \tilde{P}(r) \cdot \widehat{\mathcal F}(r) = \frac{1}{M}\sum_{u=1}^M \widehat{\mathcal F}(r_u) \\ &\approx  \frac{1}{M}\sum_{u=1}^M {\mathcal F}(R_u) = \sum_{R=1}^N P(R)\cdot {\mathcal F}(R) = T 
    \end{split}
\end{equation*}

where $\tilde{P}(r) = \frac{1}{M}\sum\limits_{r=1}^M\delta(r_u=r)$ is the empirical sampled rank distribution
and $\widehat{\mathcal F}(r)$ is the adjusted discrete metric function. An immediate observation is: 
\begin{equation}\label{eq:t_et}
    \begin{split}
        \mathbb{E} [\widehat{T}]&=\sum_{r=1}^n \mathbb{E}[\tilde{P}(r)]\cdot \widehat{\mathcal F}(r) = \sum_{r=1}^n {P}(r) \cdot  \widehat{\mathcal F}(r)
    \end{split}
\end{equation}

Following the classical statistical inference~\cite{casella2002statistical}, the optimality of an estimator is measured by Mean Squared Error:

\begin{small}
\begin{equation}\label{eqn:last_1}
    \begin{split}
    &\mathbb{E}[T-\sum_{R=1}^N P(R) \cdot \mathcal{F}(R)]^2 
 =\mathbb{E}[\mathbb{E}[T]-\sum_{R=1}^N P(R) \cdot \mathcal{F}(R) + T - \mathbb{E}[T]]^2 \\
 =&\Big(\mathbb{E}[T]- \sum_{R=1}^N P(R) \cdot\mathcal{F}(R)\Big)^2 + \mathbb{E}[T-\mathbb{E}[T]]^2  \\
=&\Big(\sum_{r=1}^n P(r) \cdot \widehat{\mathcal F}(r)-\sum_{R=1}^N P(R) \cdot \mathcal{F}(R)\Big)^2 
 +\mathbb{E}[\sum_{r=1}^n \tilde{P}(r) \cdot \widehat{\mathcal F}(r)-\sum_{r=1}^n P(r) \cdot\widehat{\mathcal F}(r)]^2 \\
  =&  \Big(\sum_{r=1}^n \sum_{R=1}^N P(r|R)\cdot P(R) \widehat{\mathcal F}(r)-\sum_{R=1}^N P(R)\cdot \mathcal{F}(R)\Big)^2 \\
  &+ \mathbb{E}[\sum_{r=1}^n \sum_{R=1}^N \tilde{P}(r|R)\cdot P(R) \widehat{\mathcal F}(r)- \sum_{r=1}^n \sum_{R=1}^N  P(r|R)\cdot P(R) \widehat{\mathcal F}(r)]^2 
    \end{split}
\end{equation}
\end{small}

Remark that $\tilde{P}(r|R)$ is the empirical conditional sampling rank distribution given a global rank $R$. 
We next use Jensen's inequality to bound the first term in \Cref{eqn:last_1}. Specifically, we may treat $\sum_{r=1}^n P(r|R) \cdot \widehat{\mathcal F}(r)- \mathcal{F}(R)$ as a random variable and use $(\mathbb{E} X)^2 \leq \mathbb{E} X^2$ to obtain

\begin{equation*}
    \begin{split}
         &\Big(\sum_{r=1}^n \sum_{R=1}^N P(r|R)\cdot P(R) \cdot\widehat{\mathcal F}(r)-\sum_{R=1}^N P(R) \cdot\mathcal{F}(R)\Big)^2   \\
         &\leq 
    \sum_{R=1}^N P(R) \Big(\sum_{r=1}^n P(r|R) \cdot\widehat{\mathcal F}(r)- \mathcal{F}(R)\Big)^2
    \end{split}
\end{equation*}

Therefore, we have

\begin{equation*}
    \begin{split}
         \mathbb{E}[\widehat{T}-\sum_{R=1}^N P(R) \mathcal{F}(R)]^2
    &\leq  \underbrace{\sum_{R=1}^N P(R) \Big\{ \Big(\sum_{r=1}^n P(r|R)  \hat{\mathcal{F}}(r)- \mathcal{F}(R)\Big)^2}_{\mathcal L_1}\\
    &
    + \underbrace{\mathbb{E}[\sum_{r=1}^n \tilde{P}(r|R) \cdot\widehat{\mathcal F}(r)- \sum_{r=1}^n  P(r|R) \cdot\widehat{\mathcal F}(r)]^2 \Big\}}_{\mathcal L_2}.
    \end{split}
\end{equation*}

Let $\mathcal L = \mathcal L_1 + \mathcal L_2$, which gives an upper bound on the expected MSE. Therefore, our goal is to find $\widehat{\mathcal{F}}(r)$ to minimize $\mathcal L$. We remark that a seemingly innocent application of Jensen's inequality results in an optimization objective that possesses a range of interesting properties:

\noindent{\textbf{1. Statistical structure.}} The objective has a variance-bias trade-off interpretation, i.e., 
\begin{equation} \label{eq:L1}
\mathcal{L}_1=\sum\limits_{R=1}^N{P(R)\Big({\mathbb{E}(\widehat{\mathcal F}(r)|R) -\mathcal{F}(R)}\Big)^2} 
\end{equation}
\begin{equation}\label{eq:L2}
\mathcal{L}_2 
 = \sum\limits_{R=1}^{N}{\frac{1}{M} Var[\widehat{\mathcal F}(r)|R]}
\end{equation}

where $\mathcal L_1$ can be interpreted as a bias term and $\mathcal L_2$ can be interpreted as a variance term. Note that while \cite{krichene@kdd20} also introduces a variance-bias trade-off objective, their objective is constructed from heuristics and contains a hyper-parameter (that determines the relative weight between bias and variance) that needs to be tuned in an ad-hoc manner. Here, because our objective is constructed from direct optimization of the MSE, it is more principled and also removes dependencies on hyperparameters. See ~\Cref{sec:analysisl2} for proving ~\Cref{eq:L2} (~\Cref{eq:L1} is trivial) and ~\Cref{sec:closedform} for more comparison against estimators proposed in~\cite{krichene@kdd20}.

\noindent{\textbf{2. Algorithmic structure.}} while the objective is not convex, we show that the objective can be expressed in a compact manner using matrices and we can find the optimal solution in a fairly straightforward manner. In other words, Jensen's inequality substantially simplifies the computation at the cost of having a looser upper bound. See ~\Cref{sec:closedform}. 

\noindent{\textbf{3. Practical performance.}} Our experiments also confirm that the new estimator is effective (~\cref{sec:sampling-evaluation-experiement}), which suggests that Jensen's inequality makes only inconsequential and moderate performance impact to the estimator's quality.

\subsection{Analysis of ${\mathcal L}_2$}\label{sec:analysisl2}
To analyze ${\mathcal L}_2$, let us take a close look of $\tilde{P}(r|R)$. Formally, let $X_r$ be the random variable representing the number of items at rank $r$ in the item-sampling data whose original rank in the entire item set is $R$. Then, we rewrite  $\tilde{P}(r|R) = \frac{X_r}{M\cdot P(R)}$. Furthermore, it is easy to observe $(X_1, \cdots X_n)$ follows the multinomial distribution $Multi(P(1|R), \cdots,P(n|R))$. We also have:   
\begin{equation}
    \begin{split}        
        \mathbb{E}[X_r]&=M\cdot P(R)\cdot P(r|R) \\ 
        Var[X_r] &= M\cdot P(R)\cdot  P(r|R)(1-P(r|R))
    \end{split}
\end{equation}
Next, let us define a new random variable $\mathcal{B} \triangleq \sum\limits_r^{n}{\widehat{\mathcal F}(r) X_r}$, which is the weighted sum of random variables under a multinomial distribution. According to \Cref{eq:A_equation}, its variance is give by:
\begin{equation*}
    \begin{split}
        &Var[\mathcal{B}] =  \mathbb{E}[\sum_{r=1}^n {X_r}\widehat{\mathcal F}(r)- \sum_{r=1}^n  {\mathbb{E}X_r]} \widehat{\mathcal F}(r)]^2\\
        &=M\cdot P(R) \Big(\sum\limits_{r}^n{\widehat{\mathcal F}^2(r) P(r|R)}-\big(\sum\limits_{r}^n{\widehat{\mathcal F}(r) P(r|R)}\big)^2 \Big) 
    \end{split}
\end{equation*}
${\mathcal L}_2$ can be re-written (see \cref{sec:l2_re}) as:
\begin{equation*}
    \mathcal{L}_2 = \sum\limits_{R=1}^{N}{\frac{1}{M} Var[\widehat{\mathcal F}(r)|R]}
\end{equation*}

\subsection{Closed Form Solution and its Relationship to Bias-Variance Estimator}\label{sec:closedform}
We can rewrite $\mathcal L$ as a matrix format and optimizing it corresponding to a constraint least square optimization:
\begin{equation}\label{eq:new_loss}
\begin{split}
\mathcal{L}_1&=\sum\limits_{R=1}^N{P(R)\Big(\sum\limits_{r=1}^n{P(r|R) \widehat{\mathcal F}(r)-\mathcal{F}(R)}\Big)^2} = ||\sqrt{D}A\mathbf{x}-\sqrt{D}\mathbf{b}||^2_F\\
\mathcal{L}_2&=\sum\limits_{R=1}^N{P(R)\cdot \mathbb{E}[\sum_{r=1}^n \tilde{P}(r|R) \widehat{\mathcal F}(r)- \sum_{r=1}^n  P(r|R) \widehat{\mathcal F}(r)]^2 }\\
&=\sum\limits_{R=1}^N{P(R)\cdot \mathbb{E}[\sum_{r=1}^n \frac{X_r}{M\cdot P(R)} \widehat{\mathcal F}(r)- \sum_{r=1}^n  \frac{\mathbb{E}[X_r]}{M\cdot P(R)} \widehat{\mathcal F}(r)]^2 }\\
        &=\sum\limits_{R=1}^N{\frac{1}{M^2\cdot P(R)}\cdot \mathbb{E}[\sum_{r=1}^n {X_r}\widehat{\mathcal F}(r)- \sum_{r=1}^n  {\mathbb{E}[X_r]} \widehat{\mathcal F}(r)]^2 }\\
 &= \sum\limits_{R=1}^{N}{\frac{1}{M}\cdot \Big(\sum\limits_{r}^n{\widehat{\mathcal F}^2(r) P(r|R)}-\big(\sum\limits_{r}^n{\widehat{\mathcal F}(r)P(r|R)}\big)^2 \Big) }\\
        &=\frac{1}{M}\mathbf{x}^T\Lambda_1 \mathbf{x} - \frac{1}{M}||A\mathbf{x}||^2_F=\frac{1}{M}||\sqrt{\Lambda_1} \mathbf{x}||^2_F - \frac{1}{M}||Ax||^2_F\\
        \mathcal{L} &= ||\sqrt{D}A\mathbf{x}-\sqrt{D}\mathbf{b}||^2_F+\frac{1}{M}||\sqrt{\Lambda}_1\mathbf{x}||^2_F -\frac{1}{M} ||A\mathbf{x}||^2_F
    \end{split}
\end{equation}
and its solution:
\begin{equation}\label{eq:mn_closed}
    \begin{split}
        \mathbf{x}=\Big(A^TDA-\frac{1}{M}A^TA+\frac{1}{M}\Lambda_1\Big)^{-1}A^TD\mathbf{b}
    \end{split}
\end{equation}
where $M$ is the number of users and $\mathcal{M}$ is metric function, $diagM(\cdot)$ is a diagonal matrix:
\begin{small}
\begin{equation*}\label{eq:new_parameters}
\begin{split}
\mathbf{x} &= \begin{bmatrix}
        \widehat{\mathcal F}(r=1)\\
        \vdots\\
        \widehat{\mathcal F}(r=n)
        \end{bmatrix}\quad \mathbf{b}=\begin{bmatrix}
        \mathcal{F}(R=1)\\
        \vdots\\
        \mathcal{F}(R=N)
        \end{bmatrix}\in\mathbb{R}^N\\
        A_{R,r} &= P(r|R)\in \mathbb{R}^{N\times n}\quad
        D = diagM\big(P(R)\big)\in \mathbb{R}^{N\times N}\\
        \Lambda_1 &= diagM\big(\sum\limits_{R=1}^{N}P(r|R)\big)\in \mathbb{R}^{n\times n}\\
\end{split}
\end{equation*}
\end{small}

\noindent{\bf Relationship to the BV Estimator:}
The bias-variance trade-off is given by \cite{krichene@kdd20}:
\begin{small}
\begin{equation*}\label{eq:eq_bv}
\begin{split}
       \mathcal{L}_{BV}= \underbrace{\sum\limits_{R=1}^{N} P(R) (\mathbb{E}[\widehat{\mathcal F}(r)|R]-\mathcal{F}(R))^2}_{{\mathcal L}_1} + \underbrace{\sum\limits_{R=1}^{N} \gamma \cdot P(R) \cdot Var[\widehat{\mathcal F}(r)|R]}_{{\mathcal L}_2}
       \end{split}
\end{equation*}
\end{small}

We observe the main difference between the BV and our new estimator is on the ${\mathcal L}_2$ components (variance components): for our estimator, each $Var[\widehat{\mathcal F}(r)|R]$ is regularized by $1/M$ ($M$ is the number of testing users), where in BV, this term is regularized by $\gamma\cdot P(R)$. Our estimator reveals that as the number of users increases, the variance in the ${\mathcal L}_2$ components will continue to decrease, whereas the BV estimator does not consider this factor. Thus, as the user size increases, BV estimator still needs to deal with ${\mathcal L}_2$ or has to manually adjust $\gamma$. 

Finally, both BV and the new estimator rely on prior distribution $P(R)$, which is unknown. In \cite{krichene@kdd20}, the uniform distribution is used for the estimation purpose. In this cgapter, we propose to leverage the latest approaches in ~\cite{Jin@AAAI21} which provide a more accurate estimation of $P(R)$ for this purpose. The experimental results in \cref{sec:sampling-evaluation-experiement} will confirm the validity of using such distribution  estimations.

\section{Boosting Global Top-K Metric Estimation Accuracy via Adaptive Item Sampling}
\label{sec:adaptive_estimator}

\subsection{Blind Spot Issue and Adaptive Sampling}

\begin{figure}
\centering
\begin{minipage}{.48\textwidth}
  \centering
  \includegraphics[width=.9\linewidth]{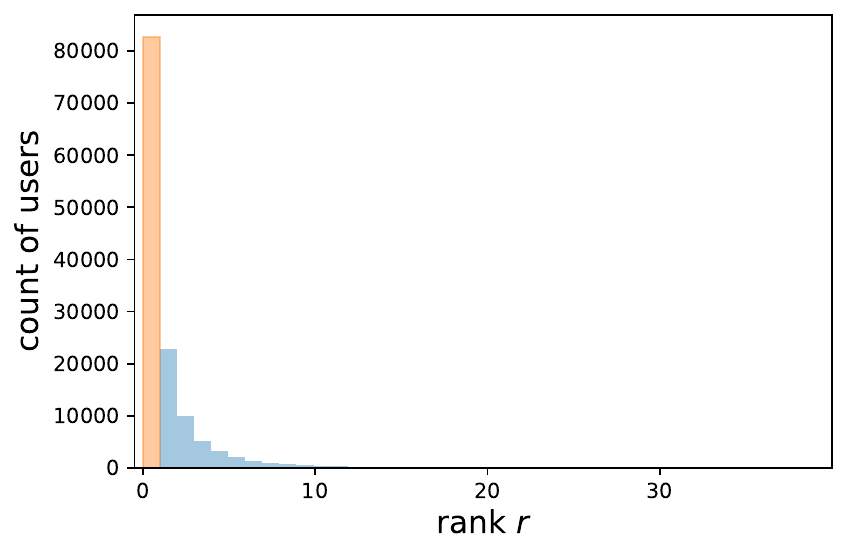}
  \captionof{figure}{Distribution of $r_u$ with sample set size $n=100$. Rank $r=1$ is highlighted.}
  \label{fig:dist-sample}
\end{minipage}%
\hfill
\begin{minipage}{.48\textwidth}
  \centering
  \includegraphics[width=.85\linewidth]{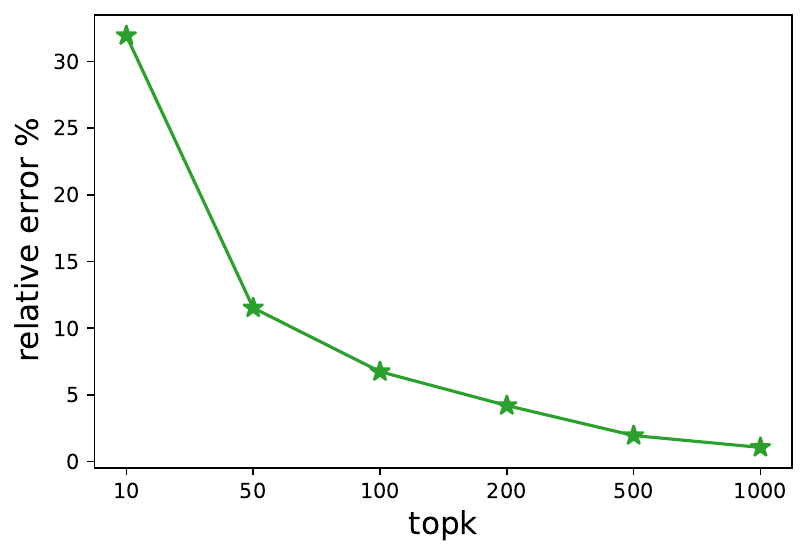}
  \captionof{figure}{the relative error of MLE estimator for different top-$K$. The result is obtained by EASE model \cite{Steck_2019} over ml-20m dataset.}
  \label{fig:blind-spot}
\end{minipage}
\end{figure}


In recommendation, top-ranked items are vital, thus it's more crucial to obtain an accurate estimation for these top items. However current sampling approaches treat all items equally and particularly have difficulty in recovering the global top-$K$ metrics when K is small. In \Cref{fig:dist-sample}, we plot the distribution of target items' rank in the sample set and observe that most target items rank top 1 (highlighted). This could lead to the "blind spot" problem - when $K$ gets smaller, the estimation of basic estimators is more inaccurate (see \Cref{fig:blind-spot}). Intuitively, when $r_u=1$, it does not mean its global rank $R_u$ is $1$, instead, its expected global rank may be around $100$ (assuming $N=10K$ and sample set size $n=100$) according to the analysis in \Cref{sec:map}.
And the estimation granularity is only at around 1\% ($1/n$) level. This blind spot effect brings a big drawback for current estimators. 

Based on above discussion, we propose an adaptive sampling strategy, which increases acceptable test sample size for users whose target item ranks top (say $r_u=1$) in the sampled data. When $r_u = 1$, we continue doubling the sample size until $r_u\neq 1$ or until sample size reaches a predetermined ceiling. See~\Cref{alg:adp}. Specifically, we start from an initial sample set size parameter $n_0$. We sample $n_0-1$ items and compute the rank $r_u$ for all users. For those users with $r_u>1$, we take down the sample set size $n_u=n_0$. For those with $r_u = 1$, we double the sample set size $n_1 = 2n_0$, in other words, we sample another set of $n_0$ items (since we already sample $n_0-1$). Consequently, we check the rank $r_u$ and repeat the process until $r_u\neq 1$ or sample set size is $n_{max}$. We will discuss how to determine $n_{max}$ later in ~\Cref{upperbound}. 

The benefits of this adaptive strategy are two folds: \textit{high granularity}, with more items sampled, the counts of $r_u = 1$ shall reduce, which could further improve the estimating accuracy; \textit{efficiency}, we iteratively sample more items for users whose $r_u=1$ and the empirical experiments confirm that small average adaptive sample size (compared to uniform sample size) is able to achieve significantly better performance.

\begin{algorithm}[t]
\caption{
Adaptive Sampling Process}\label{alg:adp}
\begin{flushleft}
        \textbf{INPUT:} 
        Recommender Model $RS$, test user set $\mathcal{U}$, initial size $n_0$, terminal size $n_{max}$
        \\
        \textbf{OUTPUT:} $\{(u, r_u, n_u)\}$
\end{flushleft}
\begin{algorithmic}[1]
\FORALL{$u \in \mathcal{U}$} 
    \STATE sampling $n_0 - 1$ items, form the sample set $I^s_u$
    
    \STATE $n_u = n_0$, $r_u = RS(i_u, I^s_u)$
    
    \WHILE{$r_u = 1$ and $n_u \neq n_{max}$}
    \STATE ${}$\hspace{2em} sampling extra $n_u$ items, form the new set $I^s_u$
    \STATE ${}$\hspace{2em} $n_u = 2n_u$, $r_u = RS(i_u, I^s_u)$
    \ENDWHILE
    \STATE record $n_u, r_u$ for user $u$
\ENDFOR
\end{algorithmic}
\end{algorithm}

\subsection{Maximum Likelihood Estimation by EM}\label{subsec: adaptive_MLE}

To utilize the adaptive item sampling for estimating the global top-$K$ metrics, we review two routes: 1) approaches from ~\cite{krichene@kdd20} and our aforementioned estimators in this chapter; 2)  methods based on MLE and EM in \Cref{sec:fix_estimator_learning_rank}. Since every user has a different number of item samples, we found the first route is hard to extend (which requires an equal sample size); but luckily the second route is much more flexible and can be easily generalized to this situation. 

To begin with, we note that for any user $u$ (his/her test item ranks $r_u$ in the sample set (with size $n_u$) and ranks $R_u$ (unknown)), its rank $r_u$ follows a binomial distribution: 
\begin{equation}
\begin{split}
    P(r=r_u|R=R_u; n_u)=Binomial(r_u -1;n_u -1,\theta_u)\\
    \end{split}
\end{equation}

Given this, let $\pmb{\Pi} = (\pi_1,\dots,\pi_R,\dots,\pi_N)^T$ be the parameters of the mixture of binomial distributions, $\pi_R$ is the probability for user population ranks at position $R$ globally. And then we have $p(r_u|\pmb{\Pi})=\sum_{R=1}^N \pi_R \cdot p(r_u|\theta_R; n_u)$, where 
$p(r_u|\theta_R; n_u)=Bin(r_u-1; n_u-1, \theta_R)$. 
We can apply the maximal likelihood estimation (MLE) to learn the parameters of the mixture of binomial distributions ($MB$), which naturally generalizes the EM procedure (details see \Cref{sec:app_em}) used in ~\cite{Jin@AAAI21}, where each user has the same $n$ samples: 

\begin{equation*}
    \begin{split}
        \phi(R_{uk}) &= P(R_u = k|r_u;\boldsymbol{\pi}^{old})\\
\pi^{new}_k&=\frac{1}{M}\sum\limits_{u=1}^M\phi(R_{uk})   
    \end{split}
\end{equation*}

When the process converges, we obtain $\pmb{\Pi}^*$ and use it to estimate $\pmb{P}$, i.e., $\widehat{P}(R)=\pi^*_R$. 
Then, we can use $\widehat{P}(R)$ in \Cref{eq:metric_observe} to estimate the desired metric $T$. 
The overall time complexity is linear with respect to the sample size $O(t\sum{n_u})$ where $t$ is the iteration number. 

\subsection{Sampling Size UpperBound}
\label{upperbound}

\begin{figure}
    \centering
    \includegraphics[ width=0.6\linewidth]{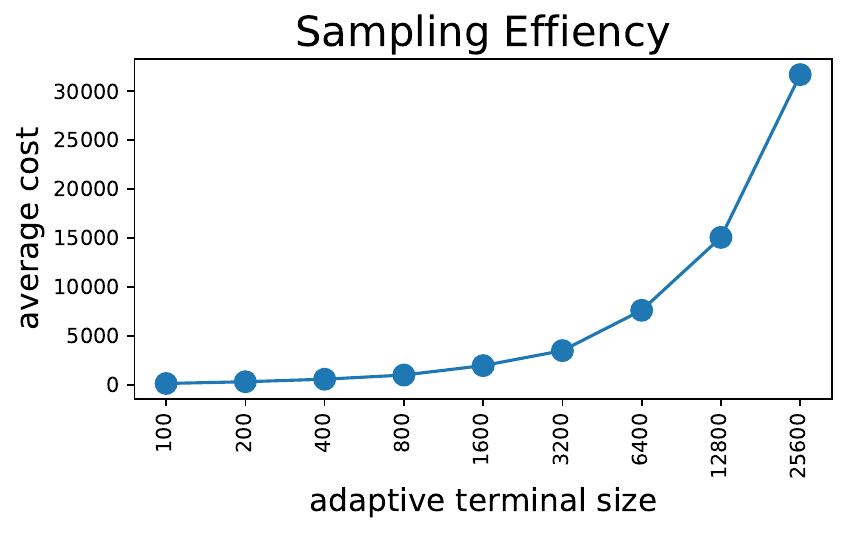}
    \caption{Sample efficiency w.r.t terminal size. The illustration result is obtained by EASE model \cite{Steck_2019} over yelp dataset while it consistently be observed in other datasets.}
    \label{fig:sample_efficiency}
\end{figure}

Now, we consider how to determine the terminal size $n_{max}$. We take the post-analysis over the different terminal sizes and investigate the average sampling cost, which introduces the concept \textit{sampling efficiency}, see \Cref{fig:sample_efficiency}. Formally, we  first select a large number $n_{max} \approx N$ and repeat the aforementioned adaptive sampling process. For each user, his/her sampling set size could be one of $\{n_0, n_1=2n_0, n_2=4n_0,\dots, n_t=n_{max}\}$. And there are $m_j$ users whose sample set size is $n_j$ $(j=0,1,\dots, t)$.
The average sampling cost for each size $n_j$ can be defined heuristically:
\begin{equation}\label{eq:ave_cost}
    \begin{split}
        C_j &= \frac{ (M-\sum\limits_{p=0}^{j-1}{m_p}) \times (n_j-n_{j-1}) }{m_j}\quad j\neq 0, t\\
        C_0 &= \frac{M\times n_0}{m_0}
    \end{split}
\end{equation}
The intuition behind \Cref{eq:ave_cost} is: at $j$-th iteration, we independently sample $n_j-n_{j-1}$ items for total $M-\sum\limits_{p=0}^{j-1}{m_p}$ users, and there are $m_j$ users whose rank $r_u> 1$. $C_j$ is the average items to be sampled to get a user whose $r_u> 1$, which reflects sampling efficiency. In ~\Cref{fig:sample_efficiency}, we can see that when the sample reaches $3200$, the sampling efficiency will reduce quickly (the average cost $C_j$ increases fast). Such post-analysis provides insights into how to balance the sample size and sampling efficiency. In this case, we observe $n_{max}=3200$ can be a reasonable choice. Even though different datasets can pickup different thresholds, we found in practice $3200$  can serve as a default choice to start and achieve pretty good performance for the estimation accuracy.

\section{User Sampling}\label{sec:user_sampling}
In real recommendation scenarios, the number of users is usually much larger than that of items. For example, there are millions of items in an online shopping portal while the user can be as many as billions. This section examines the sampling effect for the user side. While sampling users appears to be a natural strategy to speed up evaluation, there has been a lack of study from statistical analysis. 
In ~\citep{gunawardana2009survey}, the authors briefly reviewed the approach to compare two models $A$ and $B$ using the sign test ~\cite{demvsar2006statistical} and potentially more sophisticated Wilconxon signed rank test. However, they do not discuss the statistical nature of the commonly used top-$K$ evaluation metrics based on user-sampling, and how to use these user-sampling based metrics to draw a right (statistically rigorous) decision. 
Note that this section assumes the test data set can be used for evaluating the performance of recommendation models (where problems like data leakage has been solved~\cite{tamm2021quality}).
Finally, this section does not introduce new techniques. Instead, we focus on applying the available/right statistical tools to help quantify user-sampling based evaluation metrics for recommendation. Our analysis aims to offer principled guidelines for the practitioners in adopting the sampling based approaches to speedup their offline evaluation.


\subsection{Statistical Analysis for One Model}\label{sec:one_model}
First, we would like to point out that the top-K evaluation metrics on testing data itself is often considered as a special case of user-sampling (for instance, the common practice will split the data into $80-20$). Thus we hope to use testing user-sampling to approximate the overall population: 
$$T_{\mathcal{F}@K} =\sum_{R=1}^K \tilde{P}(R) \cdot {\mathcal F}(R) \approx \sum_{R=1}^K {{P}}(R) \cdot {\mathcal F}(R)=\mathbb{E}_{R} [{\mathcal F}(R)]$$
where ${\tilde{P}}(R)$ and $P(R)$ are the empirical rank distributions on the testing user population and entire user population, respectively.  $\mathcal{F}$ is any metric function (like Recall) and $T_{\mathcal{F}@K}$ is the corresponding metric result.


Let us consider the top-K Recall metric \cref{eq:def_metric_fn_1}, and it can be written as:
\begin{equation}\label{eq:recall_proportion}
    \begin{split}
        \frac{1}{M}\sum_{u=1}^M \delta({R_u\leq K}) \triangleq \frac{1}{M}\cdot Q
    \end{split}
\end{equation}
where the summation is denoted as a symbol $Q$. We assume the $R_u$ for any user $u$ follows the i.i.d. distribution, and thus $\delta({R_u\leq K})$ can be treat as a random variable of the Bernoulli distribution with some specific probability $p_K$ such that $\delta({R_u\leq K}) = 1$. Consequently 
\begin{equation}\label{eq:recall_binomial}
    \begin{split}
       Q \sim Binomial(M, p_K)
    \end{split}
\end{equation}

This is in fact the widely known point estimation for bionomial distributions~\cite{casella2002statistical}.
When $M$ (number of user) is sufficiently large enough, we assume that the top-k recall metric aka $\frac{Q}{M}$ is a good estimator of the underlying probability $p_K$. Clearly, we could estimate the underlying probability with much smaller samples ($m\ll M$), and we can also infer the sampling $m$ with respect to the margin of error (MOE) $e$: 
\begin{equation}\label{eq:m_moe}
    \begin{split}
        m = p_K(1-p_K)\cdot(\frac{z_{\frac{\alpha}{2}}}{e})^2
    \end{split}
\end{equation}
where $z_{\frac{\alpha}{2}}$ is the critical value for the corresponding confidence level. One may wonder how to determine the sample size in practice. $Recall@K$ aka $p_K$ is between $0$ and $1$, saying $Recall@30 = 0.3$ for instance. We could also take $p_K=0.5$ to get the largest sample size. If we set MOE to be $3\%$ and $1\%$  in $95\%$ confidence level:
\begin{equation}\label{eq:moe_example}
    m = 1.96^2\cdot\frac{0.5^2}{0.03^2}\approx 1067\quad  m = 1.96^2\cdot\frac{0.5^2}{0.01^2}\approx 9604
\end{equation}
Note that since most of the recommendation models with $Recall@K$ (say K=50) is often higher than say $0.4$, this can also effectively give us an estimation on the relative error estimation. In fact, this suggests $10K$ number of users can be a good rule-of-thumb for user-sampling for Recall metrics. 
In experiment \cref{sec:sampling-evaluation-experiement}, we empirically investigate the effect of different user sample sizes. 

Note that for the top-$K$ metrics, like AP and NDCG, the individual user metric is always bounded (actually between $0$ and $1$), then we may adopt Hoeffding's inequality  (we can also alternatively use a central limit theorem):

\begin{equation}\label{eq:hoeff_bound}
    \begin{split}
        Pr\big(|\frac{1}{m}\sum\limits_{u}^m{\Big(\delta({R_u\leq K})\cdot \mathcal{F}(R_u) - \mathbb{E} [\mathcal{F}(R)]\Big)}|\ge t\big)\le 2\exp(-2mt^2)
    \end{split}
\end{equation}
Given this, we can also infer the sample with targeted bound of error ($t$). For instance, when $t=0.02$ (absolute error), and sample size $10K$, the confidence bound is higher than $99.9\%$.

\subsection{Statistical Analysis of Multiple Models}\label{sec:two_model}
First, we note that since there are typically multiple models, the above analysis on the sample size and confidence interval analysis should be revised to support the statistical results for all the models are hold true. In this case, the Bonferroni correction (or Bonferroni inequality) can be leveraged to remedy this situation. This will lead the sample size to be multiplied.




Secondly, as we need to compare any two models or pickup the winners from a list of models, the statistical toolbox would require us to reply on hypothesis testing. For instance, a two-sample z-test is used to test the difference between the Recall metrics between two models, which are population proportions $p_1$ and $p_2$: 

\begin{equation}\label{eq:std_z}
    z = \frac{(\hat{p}_1-\hat{p}_2)}{\sqrt{\frac{2\bar{p}\bar{q}}{m}}}
\end{equation}
for the two hypotheses: 
\begin{equation}
\begin{array}{l}
H_{0}: p_{1}-p_{2}=0 \\
H_{a}: p_{1}-p_{2}<0
\end{array}
\end{equation}







Alternatively, we can even derive the sample size from based on the confidence interval for $\hat{p}_1-\hat{p}_2$: 
\begin{equation}
    \begin{split}
        (\hat{p}_1-\hat{p}_2)\pm z_{\frac{\alpha}{2}}\cdot \sqrt{\frac{\hat{p}_1(1-\hat{p}_1)}{m_1}+\frac{\hat{p}_2(1-\hat{p}_2)}{m_2}}
    \end{split}
\end{equation}
where $z_{\frac{\alpha}{2}}$ is the critical value for the standard normal curve with area $C$ between $-z_{\frac{\alpha}{2}}$ and $z_{\frac{\alpha}{2}}$. Set the sample size $m_1=m_2=m$ and the upper bound proportions $\hat{p}_1=\hat{p}_2\triangleq p_m=0.5$, we are able to derive the sample size for a given error range $e$ at specific $C$ confidence:
\begin{equation}\label{eq:m_two}
\begin{split}
        e &= z_{\frac{\alpha}{2}}\cdot \sqrt{\frac{2p_m(1-p_m)}{m}}\\
        m&=2p_m(1-p_m)\cdot (\frac{z_{\frac{\alpha}{2}}}{e})^2
\end{split}
\end{equation}
Compared with single model, its twice of which in \Cref{eq:m_moe}

Finally, for the general metrics, like AP and NDCG, which can not be represented as population proportion, we can resort to the two-sample t-test to decide if one model is better than the other. Furthermore, if we consider multiple comparisons at the same time, Bonferroni inequality again has to be used.

\section{experiment}\label{sec:sampling-evaluation-experiement}

In this section, we investigate the experimental results of the mapping function proposed in \Cref{sec:map}, Top-K metric estimators proposed in \Cref{sec:fix_estimator_learning_rank,sec:fix_estimator_learning_metric}, and also adaptive estimator in \Cref{sec:adaptive_estimator}, user-based sampling in \Cref{sec:user_sampling}. Specifically, we aim to answer following questions: 

\begin{itemize}
    \item {\textbf{Q1.}} How do various approximating mapping functions f (\Cref{subsec: mapping_function}) help align $T^S_{Recall@K}$ with respect to $T_{Recall@f(K)}$
    \item {\textbf{Q2.}} How do these estimators (\Cref{sec:fix_estimator_learning_rank,sec:fix_estimator_learning_metric}) perform compared to baseline BV \citep{krichene@kdd20} on estimating the top-K metrics based on sampling?
    \item \textbf{Q3.} How effective and efficient the adaptive item-sampling evaluation method - adaptive MLE (\Cref{sec:adaptive_estimator}) is, compared with the basic (non-adaptive) item-sampling methods?
    \item {\textbf{Q4.}} How accurately can these estimators (\Cref{sec:adaptive_estimator,sec:fix_estimator_learning_rank,sec:fix_estimator_learning_metric}) find the best model (in terms of the global top-K metric) among a list of recommendation models?
    \item {\textbf{Q5.}} How effective is the user-sampling based evaluation method (\Cref{sec:user_sampling})?
\end{itemize}

\subsection{Experimental Setting}

\subsubsection{Datasets}
We conduct experiments on 3 widely-used relatively large datasets (\textit{pinterest-20, yelp, ml-20m}) for recommendation system research. \Cref{tab:datasets} shows the information of these three datasets. 
\begin{table}[]
  \caption{Dataset Statistics}
  \label{tab:datasets}
  \begin{tabular}{lcccc}
    \toprule
    \textbf{Dataset}&
    \textbf{Interactions}&
    \textbf{Users}&
    \textbf{Items}&\textbf{Sparsity}\\
    \midrule
    pinterest-20& 1,463,581 & 55,187&9,916&99.73$\%$\\
    yelp& 696,865&25,677 &25,815&99.89$\%$\\
    ml-20m&9,990,682&136,677&20,720&99.65$\%$\\
    \bottomrule
  \end{tabular}
\end{table}
\subsubsection{Recommendation Models}
 We use five widely-used recommendation algorithms, including three non-deep-learning methods (itemKNN~\cite{DeshpandeK@itemKNN}, 
ALS~\cite{hu2008collaborative}, and EASE~\cite{Steck_2019}) 
and two deep learning ones (NeuMF~\cite{he2017neural} and MultiVAE~\cite{liang2018variational}). Selection of models try to enable varied performance and advantage in different datasets \citep{krichene@kdd20}.
\subsubsection{Evaluation Metrics}
Three most popular top-K metrics (\Cref{eq:def_metric_fn_1}): $Recall$, $NDCG$ and $AP$ are utilized for evaluating the recommendation models. 
\subsubsection{Evaluating and Estimating Procedure} Recalling that there are $M$ users and $N$ items. 
Each user $u$ is associated with a target item $i_u$ (\textit{leave-one-out}). The learned recommendation algorithm/model $A$ would compute the ranks $\{R_u\}_{u=1}^M$ among all items called global ranks and the ranks $\{r_u\}_{u=1}^M$ among sampled items called sampled ranks. Without the knowledge of $\{R_u\}_{u=1}^M$, the estimator tries to estimate the global metric $T$ defined in \Cref{eq:def_metric_global,eq:def_metric_fn_1,eq:def_metric_1} based on sampled set test results $\{r_u\}_{u=1}^M$. We repeat experiments $100$ times, deriving $100$ distinct $\{r_u\}_{u=1}^M$ results. Below reported experimental results are displayed with mean and standard deviation over these $100$ repeats.
\subsubsection{Item-sampling based estimators} 
$\textbf{BV}$ (Bias-Variance Estimator)\citep{krichene@kdd20}; $\textbf{MLE}$ (Maximal Likelihood Estimation) from \Cref{subsec:learning_empirical} \cite{Jin@AAAI21}; $\textbf{MES}$ (Maximal Entropy with Squared distribution distance) from \Cref{subsec:learning_empirical} \cite{Jin@AAAI21}; $\textbf{BV\_MLE}$, $\textbf{BV\_MES}$ (\Cref{eq:leastsquare} with $P(R)$ obtained from $MLE$ and $MES$; basically, we consider combining the two approaches from \textbf{BV} \cite{krichene@kdd20} and $\textbf{MLE}$/$\textbf{MES}$ \cite{Jin@AAAI21}); the new multinomial distribution based estimator with different prior, short as $\textbf{MN\_MLE}$, $\textbf{MN\_MES}$, \cref{eq:mn_closed} with prior $P(R)$ obtained from $MLE$ and $MES$ estimators.

\subsection{Aligning Sampling and Global Top-K Recall Metric (Q1)}

\begin{figure}
    \centering
    \includegraphics[width=0.95\linewidth]{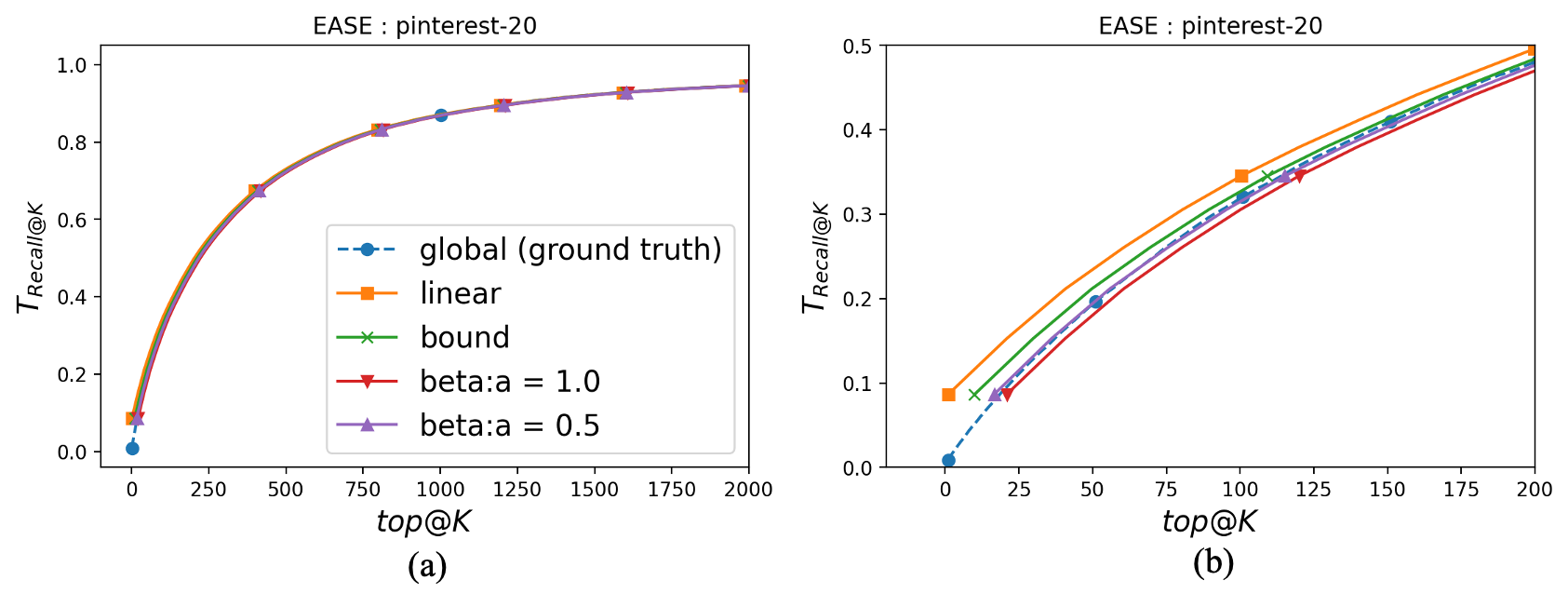}
    \caption{$T^S_{Recall@k}$ curves alignment with $T_{Recall@K}$ by different mapping function. (a) shows the results K from 1 to 2000 and (b) zoom out shows the details from 1 to 200. All the mapping functions exhibit promising approximating especially bound and $beta@0.5$. An example of model $EASE$ conducted in \textit{pinterest-20} dataset.}
    \label{fig:aligning}
\end{figure}

In this section, we provide a global view of the alignment between sampling and global Recall (curves). Here, we report an example the \texttt {EASE} results on dataset \texttt{pinterest-20} in \Cref{fig:aligning}. We use four different approximating mapping functions from \Cref{subsec: mapping_function}, the linear, bound, $beta@1$ and $beta@0.5$, for the curve alignment. \Cref{subsec: mapping_function} maps the sampling curve $T^S_{Recall@k}$ to the global top-$K$ view by mapping from location $K$ to location $f(K)$ in the global top $K$ view, and compares them with the global $T_{Recall@K}$ (ground truth curve). We observe that both bound and $beta@0.5$ achieve the superior results, which further validate our claim of mapping functions.

\subsection{Estimation Accuracy of Estimators (Q2)}

Here, we aim to answer Question 2: How do these estimators proposed in this chapter perform compared to baseline BV\citep{krichene@kdd20} on estimating the top-K metrics based on sampling?
Here we would try to quantify the accuracy of each estimator in terms of relative error, leading to a more rigorous and reliable comparison. Specifically, we compute the true global $T_{metric@K}$ ( $K$ from $1$ to $50$), then we average the absolute relative error between the estimated $\hat{T}_{metric@k}$ from each estimator and the true one.
 
The estimators include $BV$ (with the tradeoff parameter $\gamma=0.01$) from \cite{krichene@kdd20}, $MLE$ (Maximal Likelihood Estimation), $MES$ (Maximal Entropy with Squared distribution distance, where $\eta=0.001$) from \Cref{sec:fix_estimator_learning_metric,sec:fix_estimator_learning_rank} \cite{Jin@AAAI21} . \Cref{tab:recall_error_100_cp} presents the average relative error of the estimators in terms of $T_{Recall@K}$ ($k$ from $1$ to $50$). We highlight the most and the second-most accurate estimator. For instance, for model $EASE$ in dataset \textit{pinterest-20} (line $1$ of \Cref{tab:recall_error_100_cp}), the estimator $MN\_MES$ is the most accurate one with $5.00\%$ average relative error compared to its global $T_{Recall@K}$ ($K$ from $1$ to $50$).
 
Overall, we observe from \Cref{tab:recall_error_100_cp} that $MN\_MES$ and $MN\_MLE$ are among the most, or the second-most accurate estimators.  And in most cases, they outperform the others significantly. Meantime, they have smaller deviation compared to their prior estimators $MES$ and $MLE$. In addition, we also notice that the estimators with the knowledge of some reasonable prior distribution ($BV\_MES$, $MN\_MES$, $BV\_MLE$, $MN\_MLE$) could achieve more accurate results than the others. This indicate that these estimators could better help the distribution to converge.

\begin{table*}[]
\centering
\caption{The average relative errors between estimated $\widehat{T}_{Recall@K}$ ($K$ from $1$ to $50$) and the true ones ${T}_{Recall@K}$. Unit is $\%$. In each row, the smallest two results are highlighted in bold, indicating the most accurate results. Sample set size $n=100$.}
\label{tab:recall_error_100_cp}
\resizebox{0.9\textwidth}{!}{%
\begin{tabular}{|c|c|rrrrrrr|}
\hline
\multirow{2}{*}{dataset} &
  \multirow{2}{*}{Models} &
  \multicolumn{7}{c|}{sample set size 100} \\ \cline{3-9} 
 &
   &
  \multicolumn{1}{c|}{MES} &
  \multicolumn{1}{c|}{MLE} &
  \multicolumn{1}{c|}{BV} &
  \multicolumn{1}{c|}{BV\_MES} &
  \multicolumn{1}{c|}{BV\_MLE} &
  \multicolumn{1}{c|}{MN\_MES} &
  \multicolumn{1}{c|}{MN\_MLE} \\ \hline
\multirow{5}{*}{pinterest-20} &
  EASE &
  \multicolumn{1}{r|}{5.86$\rpm$2.26} &
  \multicolumn{1}{r|}{5.54$\rpm$1.85} &
  \multicolumn{1}{r|}{8.11$\rpm$2.00} &
  \multicolumn{1}{r|}{\textbf{5.05$\rpm$1.46}} &
  \multicolumn{1}{r|}{5.14$\rpm$1.46} &
  \multicolumn{1}{r|}{\textbf{5.00$\rpm$1.39}} &
  5.10$\rpm$1.34 \\ \cline{2-9} 
 &
  MultiVAE &
  \multicolumn{1}{r|}{4.17$\rpm$2.91} &
  \multicolumn{1}{r|}{3.34$\rpm$2.07} &
  \multicolumn{1}{r|}{\textbf{2.75$\rpm$1.61}} &
  \multicolumn{1}{r|}{2.89$\rpm$1.74} &
  \multicolumn{1}{r|}{\textbf{2.88$\rpm$1.74}} &
  \multicolumn{1}{r|}{\textbf{2.75$\rpm$1.66}} &
  \textbf{2.75$\rpm$1.68} \\ \cline{2-9} 
 &
  NeuMF &
  \multicolumn{1}{r|}{5.17$\rpm$2.74} &
  \multicolumn{1}{r|}{4.28$\rpm$1.95} &
  \multicolumn{1}{r|}{4.23$\rpm$1.79} &
  \multicolumn{1}{r|}{3.83$\rpm$1.59} &
  \multicolumn{1}{r|}{3.84$\rpm$1.72} &
  \multicolumn{1}{r|}{\textbf{3.60$\rpm$1.50}} &
  \textbf{3.76$\rpm$1.44} \\ \cline{2-9} 
 &
  itemKNN &
  \multicolumn{1}{r|}{5.90$\rpm$2.20} &
  \multicolumn{1}{r|}{5.80$\rpm$1.60} &
  \multicolumn{1}{r|}{8.93$\rpm$1.70} &
  \multicolumn{1}{r|}{\textbf{5.11$\rpm$1.22}} &
  \multicolumn{1}{r|}{5.31$\rpm$1.25} &
  \multicolumn{1}{r|}{\textbf{5.09$\rpm$1.15}} &
  5.26$\rpm$1.14 \\ \cline{2-9} 
 &
  ALS &
  \multicolumn{1}{r|}{4.19$\rpm$2.37} &
  \multicolumn{1}{r|}{3.44$\rpm$1.68} &
  \multicolumn{1}{r|}{3.17$\rpm$1.34} &
  \multicolumn{1}{r|}{3.05$\rpm$1.39} &
  \multicolumn{1}{r|}{3.07$\rpm$1.42} &
  \multicolumn{1}{r|}{\textbf{2.86$\rpm$1.27}} &
  \textbf{2.90$\rpm$1.28} \\ \hline
\multirow{5}{*}{yelp} &
  EASE &
  \multicolumn{1}{r|}{8.08$\rpm$4.94} &
  \multicolumn{1}{r|}{7.89$\rpm$4.70} &
  \multicolumn{1}{r|}{18.60$\rpm$2.78} &
  \multicolumn{1}{r|}{6.10$\rpm$3.74} &
  \multicolumn{1}{r|}{6.56$\rpm$3.90} &
  \multicolumn{1}{r|}{\textbf{4.84$\rpm$2.17}} &
  \textbf{5.61$\rpm$2.30} \\ \cline{2-9} 
 &
  MultiVAE &
  \multicolumn{1}{r|}{9.33$\rpm$6.61} &
  \multicolumn{1}{r|}{7.67$\rpm$4.94} &
  \multicolumn{1}{r|}{9.70$\rpm$3.22} &
  \multicolumn{1}{r|}{6.84$\rpm$4.10} &
  \multicolumn{1}{r|}{6.80$\rpm$4.04} &
  \multicolumn{1}{r|}{\textbf{4.30$\rpm$1.27}} &
  \textbf{4.35$\rpm$1.31} \\ \cline{2-9} 
 &
  NeuMF &
  \multicolumn{1}{r|}{15.09$\rpm$6.24} &
  \multicolumn{1}{r|}{15.47$\rpm$5.55} &
  \multicolumn{1}{r|}{22.40$\rpm$3.17} &
  \multicolumn{1}{r|}{\textbf{13.14$\rpm$4.55}} &
  \multicolumn{1}{r|}{13.92$\rpm$4.70} &
  \multicolumn{1}{r|}{\textbf{13.46$\rpm$2.43}} &
  14.50$\rpm$2.45 \\ \cline{2-9} 
 &
  itemKNN &
  \multicolumn{1}{r|}{9.25$\rpm$4.87} &
  \multicolumn{1}{r|}{9.62$\rpm$4.88} &
  \multicolumn{1}{r|}{23.24$\rpm$2.16} &
  \multicolumn{1}{r|}{\textbf{7.69$\rpm$4.09}} &
  \multicolumn{1}{r|}{8.15$\rpm$4.17} &
  \multicolumn{1}{r|}{\textbf{7.74$\rpm$2.08}} &
  8.75$\rpm$2.08 \\ \cline{2-9} 
 &
  ALS &
  \multicolumn{1}{r|}{14.31$\rpm$3.96} &
  \multicolumn{1}{r|}{13.68$\rpm$3.51} &
  \multicolumn{1}{r|}{15.14$\rpm$1.86} &
  \multicolumn{1}{r|}{13.43$\rpm$3.16} &
  \multicolumn{1}{r|}{13.26$\rpm$3.08} &
  \multicolumn{1}{r|}{\textbf{11.68$\rpm$0.88}} &
  \textbf{11.57$\rpm$0.83} \\ \hline
\multirow{5}{*}{ml-20m} &
  EASE &
  \multicolumn{1}{r|}{10.45$\rpm$1.03} &
  \multicolumn{1}{r|}{11.52$\rpm$1.03} &
  \multicolumn{1}{r|}{36.59$\rpm$0.31} &
  \multicolumn{1}{r|}{\textbf{8.99$\rpm$0.74}} &
  \multicolumn{1}{r|}{9.86$\rpm$0.77} &
  \multicolumn{1}{r|}{\textbf{9.07$\rpm$0.61}} &
  10.09$\rpm$0.69 \\ \cline{2-9} 
 &
  MultiVAE &
  \multicolumn{1}{r|}{9.93$\rpm$0.38} &
  \multicolumn{1}{r|}{\textbf{9.48$\rpm$0.22}} &
  \multicolumn{1}{r|}{22.24$\rpm$0.37} &
  \multicolumn{1}{r|}{9.85$\rpm$0.36} &
  \multicolumn{1}{r|}{\textbf{9.50$\rpm$0.22}} &
  \multicolumn{1}{r|}{9.82$\rpm$0.28} &
  9.53$\rpm$0.14 \\ \cline{2-9} 
 &
  NeuMF &
  \multicolumn{1}{r|}{4.35$\rpm$1.50} &
  \multicolumn{1}{r|}{6.05$\rpm$1.35} &
  \multicolumn{1}{r|}{28.27$\rpm$0.42} &
  \multicolumn{1}{r|}{\textbf{3.67$\rpm$1.14}} &
  \multicolumn{1}{r|}{4.81$\rpm$1.14} &
  \multicolumn{1}{r|}{\textbf{3.64$\rpm$1.05}} &
  4.79$\rpm$1.08 \\ \cline{2-9} 
 &
  itemKNN &
  \multicolumn{1}{r|}{15.31$\rpm$1.18} &
  \multicolumn{1}{r|}{17.19$\rpm$1.15} &
  \multicolumn{1}{r|}{36.63$\rpm$0.42} &
  \multicolumn{1}{r|}{\textbf{14.02$\rpm$0.75}} &
  \multicolumn{1}{r|}{15.24$\rpm$0.83} &
  \multicolumn{1}{r|}{\textbf{14.16$\rpm$0.68}} &
  15.41$\rpm$0.77 \\ \cline{2-9} 
 &
  ALS &
  \multicolumn{1}{r|}{36.17$\rpm$0.83} &
  \multicolumn{1}{r|}{\textbf{35.21$\rpm$0.64}} &
  \multicolumn{1}{r|}{36.39$\rpm$0.21} &
  \multicolumn{1}{r|}{36.50$\rpm$0.74} &
  \multicolumn{1}{r|}{35.75$\rpm$0.62} &
  \multicolumn{1}{r|}{36.32$\rpm$0.56} &
  \textbf{35.60$\rpm$0.48} \\ \hline
\end{tabular}%
}
\end{table*}

\subsection{Efficiency and Effectiveness of Adaptive Estimator (Q3)}

Here, we aim to answer Question 3: How effective and efficient the adaptive item-sampling evaluation method - adaptive MLE (Section 7) is, compared with the basic (non-adaptive) item-sampling methods, such as $BV\_MES$, $MN\_MES$, $BV\_MLE$, $MN\_MLE$ ?

 \Cref{tab:recall_adaptive_cp}  presents the average relative error of the estimators in terms of $T_{Recall@K}$ ($k$ from $1$ to $50$). We highlight the most accurate estimator. For the basic item sampling, we choose $500$ sample size for datasets \textit{pinterest-20} and $yelp$, and $1000$ sample size for dataset \textit{ml-20m}. The upper bound threshold $n_{max}$ is set at $3200$. 
 
 We observe that adaptive sampling uses much less sample size (typically $200-300$ vs $500$ on the first two datasets and $700-800$ vs $1000$ on the last dataset. Particularly, the relative error of the adaptive sampling is significantly less than that of the basic sampling methods. On the first (\textit{pinterest-20}) and third (\textit{ml-20m}) datasets, the relative errors have reduced to less than $2\%$. In other words, the adaptive method has been much more effective (in terms of accuracy) and efficient (in terms of sample size).  This also confirms the benefits in addressing the ``blind spot'' issue, which provides higher resolution to recover global $K$ metrics for small $K$ ($K\leq 50$ here).

\begin{table*}[]
\centering
\caption{Comparison of Adaptive estimators with basic ones in terms of $Recall$. The average relative errors between estimated $\widehat{T}_{Recall@K}$ ($K$ from $1$ to $50$) and the true ones. Unit is $\%$. In each row, the smallest relative error is highlighted, indicating the most accurate result.}
\label{tab:recall_adaptive_cp}
\resizebox{0.9\textwidth}{!}{%
\begin{tabular}{|cc|llll|cc|}
\hline
\multicolumn{1}{|c|}{\multirow{2}{*}{Dataset}}      & \multirow{2}{*}{Models} & \multicolumn{4}{c|}{Fix Sample}                                                                                                                  & \multicolumn{2}{c|}{Adaptive Sample}                        \\ \cline{3-8} 
\multicolumn{1}{|c|}{}                              &                         & \multicolumn{1}{c|}{BV\_MES}      & \multicolumn{1}{c|}{BV\_MLE}      & \multicolumn{1}{c|}{MN\_MES}              & \multicolumn{1}{c|}{MN\_MLE} & \multicolumn{1}{c|}{average size}   & adaptive MLE          \\ \hline
\multicolumn{2}{|c|}{}                                                        & \multicolumn{4}{c|}{sample set size n = 500}                                                                                                     & \multicolumn{2}{c|}{}                                       \\ \hline
\multicolumn{1}{|c|}{\multirow{5}{*}{pinterest-20}} & EASE                    & \multicolumn{1}{l|}{2.54$\rpm$0.85}  & \multicolumn{1}{l|}{2.68$\rpm$0.87}  & \multicolumn{1}{l|}{2.78$\rpm$1.05}          & 2.83$\rpm$1.06                  & \multicolumn{1}{c|}{307.74$\rpm$1.41} & \textbf{1.69$\rpm$0.60} \\ \cline{2-8} 
\multicolumn{1}{|c|}{}                              & MultiVAE                & \multicolumn{1}{l|}{2.17$\rpm$1.08}  & \multicolumn{1}{l|}{2.13$\rpm$1.09}  & \multicolumn{1}{l|}{2.60$\rpm$1.30}          & 2.55$\rpm$1.35                  & \multicolumn{1}{c|}{286.46$\rpm$1.48} & \textbf{1.95$\rpm$0.65} \\ \cline{2-8} 
\multicolumn{1}{|c|}{}                              & NeuMF                   & \multicolumn{1}{l|}{2.45$\rpm$1.15}  & \multicolumn{1}{l|}{2.44$\rpm$1.15}  & \multicolumn{1}{l|}{2.76$\rpm$1.37}          & 2.80$\rpm$1.38                  & \multicolumn{1}{c|}{259.77$\rpm$1.28} & \textbf{2.00$\rpm$0.81} \\ \cline{2-8} 
\multicolumn{1}{|c|}{}                              & itemKNN                 & \multicolumn{1}{l|}{2.49$\rpm$0.97}  & \multicolumn{1}{l|}{2.59$\rpm$0.94}  & \multicolumn{1}{l|}{2.79$\rpm$1.12}          & 2.79$\rpm$1.20                  & \multicolumn{1}{c|}{309.56$\rpm$1.31} & \textbf{1.63$\rpm$0.51} \\ \cline{2-8} 
\multicolumn{1}{|c|}{}                              & ALS                     & \multicolumn{1}{l|}{2.65$\rpm$1.04}  & \multicolumn{1}{l|}{2.63$\rpm$1.06}  & \multicolumn{1}{l|}{3.02$\rpm$1.32}          & 2.98$\rpm$1.33                  & \multicolumn{1}{c|}{270.75$\rpm$1.22} & \textbf{2.00$\rpm$0.73} \\ \hline
\multicolumn{2}{|c|}{}                                                        & \multicolumn{4}{c|}{sample set size n = 500}                                                                                                     & \multicolumn{2}{c|}{}                                       \\ \hline
\multicolumn{1}{|c|}{\multirow{5}{*}{yelp}}         & EASE                    & \multicolumn{1}{l|}{4.68$\rpm$2.43}  & \multicolumn{1}{l|}{4.56$\rpm$2.35}  & \multicolumn{1}{l|}{\textbf{3.47$\rpm$1.79}} & 3.49$\rpm$1.78                  & \multicolumn{1}{c|}{340.79$\rpm$2.03} & 3.48$\rpm$1.40          \\ \cline{2-8} 
\multicolumn{1}{|c|}{}                              & MultiVAE                & \multicolumn{1}{l|}{6.14$\rpm$3.48}  & \multicolumn{1}{l|}{6.07$\rpm$3.46}  & \multicolumn{1}{l|}{4.68$\rpm$2.27}          & \textbf{4.67$\rpm$2.28}         & \multicolumn{1}{c|}{288.70$\rpm$2.24} & 5.08$\rpm$2.14          \\ \cline{2-8} 
\multicolumn{1}{|c|}{}                              & NeuMF                   & \multicolumn{1}{l|}{6.59$\rpm$2.38}  & \multicolumn{1}{l|}{6.73$\rpm$2.35}  & \multicolumn{1}{l|}{5.48$\rpm$1.43}          & 5.68$\rpm$1.42                  & \multicolumn{1}{c|}{290.62$\rpm$2.11} & \textbf{4.01$\rpm$1.51} \\ \cline{2-8} 
\multicolumn{1}{|c|}{}                              & itemKNN                 & \multicolumn{1}{l|}{3.94$\rpm$1.94}  & \multicolumn{1}{l|}{3.95$\rpm$1.92}  & \multicolumn{1}{l|}{\textbf{2.92$\rpm$1.60}} & 2.96$\rpm$1.57                  & \multicolumn{1}{c|}{369.16$\rpm$2.51} & 3.25$\rpm$1.59          \\ \cline{2-8} 
\multicolumn{1}{|c|}{}                              & ALS                     & \multicolumn{1}{l|}{10.00$\rpm$3.47} & \multicolumn{1}{l|}{10.31$\rpm$3.65} & \multicolumn{1}{l|}{9.29$\rpm$2.03}          & 9.80$\rpm$2.23                  & \multicolumn{1}{c|}{297.07$\rpm$2.29} & \textbf{5.25$\rpm$2.38} \\ \hline
\multicolumn{2}{|c|}{}                                                        & \multicolumn{4}{c|}{sample set size n = 1000}                                                                                                    & \multicolumn{2}{c|}{}                                       \\ \hline
\multicolumn{1}{|c|}{\multirow{5}{*}{ml-20m}}       & EASE                    & \multicolumn{1}{l|}{1.39$\rpm$0.21}  & \multicolumn{1}{l|}{1.69$\rpm$0.28}  & \multicolumn{1}{l|}{1.81$\rpm$0.46}          & 1.73$\rpm$0.46                  & \multicolumn{1}{c|}{899.89$\rpm$1.90} & \textbf{1.07$\rpm$0.24} \\ \cline{2-8} 
\multicolumn{1}{|c|}{}                              & MultiVAE                & \multicolumn{1}{l|}{2.23$\rpm$0.58}  & \multicolumn{1}{l|}{2.91$\rpm$0.72}  & \multicolumn{1}{l|}{3.55$\rpm$1.23}          & 2.98$\rpm$1.50                  & \multicolumn{1}{c|}{771.26$\rpm$1.84} & \textbf{1.10$\rpm$0.39} \\ \cline{2-8} 
\multicolumn{1}{|c|}{}                              & NeuMF                   & \multicolumn{1}{l|}{0.82$\rpm$0.30}  & \multicolumn{1}{l|}{0.85$\rpm$0.28}  & \multicolumn{1}{l|}{1.51$\rpm$0.66}          & 1.69$\rpm$0.70                  & \multicolumn{1}{c|}{758.45$\rpm$1.61} & \textbf{0.78$\rpm$0.27} \\ \cline{2-8} 
\multicolumn{1}{|c|}{}                              & itemKNN                 & \multicolumn{1}{l|}{1.84$\rpm$0.24}  & \multicolumn{1}{l|}{2.13$\rpm$0.27}  & \multicolumn{1}{l|}{1.97$\rpm$0.42}          & 2.17$\rpm$0.49                  & \multicolumn{1}{c|}{725.72$\rpm$1.49} & \textbf{1.17$\rpm$0.28} \\ \cline{2-8} 
\multicolumn{1}{|c|}{}                              & ALS                     & \multicolumn{1}{l|}{9.41$\rpm$0.97}  & \multicolumn{1}{l|}{12.83$\rpm$1.27} & \multicolumn{1}{l|}{10.63$\rpm$2.53}         & 10.57$\rpm$3.18                 & \multicolumn{1}{c|}{705.76$\rpm$1.56} & \textbf{4.29$\rpm$1.05} \\ \hline
\end{tabular}
}
\end{table*}

\subsection{Estimating Winner of Recommender Models (Q4)}

Besides on the estimation accuracy that we studied, we also care about could the estimator correctly find the best recommendation model(s)? Initially, the reason that we compare the performance of various recommendation models (among \textit{EASE, ALS, itemKNN, NeuMF, MultiVAE} in this chapter) is try to find the best model(s) to deploy. The best model is determined by the global top-K metric in \Cref{eq:def_metric_1} ($\{ R_u\}$). Thus it is also meaningful to validate whether our estimated metric $\widehat{T}$ could find the correct "winner" (best model) as original metric $T$.
\Cref{tab:adaptive_winner_all} indicates the results of among the $100$ repeats, how many times that a estimator could find the best recommendation algorithm for a given metric (Recall, NDCG, AP). From the \Cref{tab:adaptive_winner_all}, We observe that new proposed estimators in \Cref{sec:fix_estimator_learning_metric,sec:fix_estimator_learning_rank} could generally achieve competitive or even better results to the baseline BV\citep{krichene@kdd20}. For adaptive estimator, it could achieve best accuracy in most case while cost less in average.

\begin{table}[]
\centering
\caption{Accuracy of estimating winner (of the recommendation models). Values in the table is the number of correcting predict the winner out of 100 repeat test. We highlight for adaptive estimator if it achieve best performance w.r.t each metric (each row).}
\label{tab:adaptive_winner_all}
\resizebox{0.8\textwidth}{!}{%

\begin{tabular}{|c|c|c|ccccccc|c|}
\hline
\multirow{2}{*}{Dataset}      & \multirow{2}{*}{Top-K} & \multirow{2}{*}{Metric} & \multicolumn{7}{c|}{Fix Sample}                                                                                                                                                       & Adaptive Sample   \\ \cline{4-11} 
                              &                        &                                 & \multicolumn{1}{c|}{MES} & \multicolumn{1}{c|}{MLE} & \multicolumn{1}{c|}{BV}  & \multicolumn{1}{c|}{BV\_MES} & \multicolumn{1}{c|}{BV\_MLE} & \multicolumn{1}{c|}{MN\_MES} & MN\_MLE & adaptive MLE      \\ \hline
                              &                        &                                 & \multicolumn{7}{c|}{sample set size n = 500}                                                                                                                                          & size 260$\sim$310 \\ \hline
\multirow{9}{*}{pinterest-20} & \multirow{3}{*}{5}     & RECALL                          & \multicolumn{1}{c|}{53}  & \multicolumn{1}{c|}{56}  & \multicolumn{1}{c|}{61}  & \multicolumn{1}{c|}{54}      & \multicolumn{1}{c|}{57}      & \multicolumn{1}{c|}{53}      & 53      & \textbf{69}       \\ \cline{3-11} 
                              &                        & NDCG                            & \multicolumn{1}{c|}{51}  & \multicolumn{1}{c|}{54}  & \multicolumn{1}{c|}{60}  & \multicolumn{1}{c|}{52}      & \multicolumn{1}{c|}{54}      & \multicolumn{1}{c|}{51}      & 52      & \textbf{71}       \\ \cline{3-11} 
                              &                        & AP                              & \multicolumn{1}{c|}{51}  & \multicolumn{1}{c|}{53}  & \multicolumn{1}{c|}{58}  & \multicolumn{1}{c|}{51}      & \multicolumn{1}{c|}{53}      & \multicolumn{1}{c|}{51}      & 51      & \textbf{60}       \\ \cline{2-11} 
                              & \multirow{3}{*}{10}    & RECALL                          & \multicolumn{1}{c|}{66}  & \multicolumn{1}{c|}{66}  & \multicolumn{1}{c|}{69}  & \multicolumn{1}{c|}{69}      & \multicolumn{1}{c|}{73}      & \multicolumn{1}{c|}{67}      & 69      & \textbf{78}       \\ \cline{3-11} 
                              &                        & NDCG                            & \multicolumn{1}{c|}{55}  & \multicolumn{1}{c|}{58}  & \multicolumn{1}{c|}{65}  & \multicolumn{1}{c|}{58}      & \multicolumn{1}{c|}{59}      & \multicolumn{1}{c|}{58}      & 60      & \textbf{84}       \\ \cline{3-11} 
                              &                        & AP                              & \multicolumn{1}{c|}{53}  & \multicolumn{1}{c|}{55}  & \multicolumn{1}{c|}{61}  & \multicolumn{1}{c|}{54}      & \multicolumn{1}{c|}{57}      & \multicolumn{1}{c|}{52}      & 52      & \textbf{68}       \\ \cline{2-11} 
                              & \multirow{3}{*}{20}    & RECALL                          & \multicolumn{1}{c|}{69}  & \multicolumn{1}{c|}{69}  & \multicolumn{1}{c|}{75}  & \multicolumn{1}{c|}{69}      & \multicolumn{1}{c|}{73}      & \multicolumn{1}{c|}{70}      & 74      & \textbf{81}       \\ \cline{3-11} 
                              &                        & NDCG                            & \multicolumn{1}{c|}{69}  & \multicolumn{1}{c|}{69}  & \multicolumn{1}{c|}{78}  & \multicolumn{1}{c|}{69}      & \multicolumn{1}{c|}{73}      & \multicolumn{1}{c|}{68}      & 73      & \textbf{79}       \\ \cline{3-11} 
                              &                        & AP                              & \multicolumn{1}{c|}{55}  & \multicolumn{1}{c|}{58}  & \multicolumn{1}{c|}{62}  & \multicolumn{1}{c|}{57}      & \multicolumn{1}{c|}{60}      & \multicolumn{1}{c|}{54}      & 56      & \textbf{69}       \\ \hline
                              &                        &                                 & \multicolumn{7}{c|}{sample set size n = 500}                                                                                                                                          & size 280$\sim$370 \\ \hline
\multirow{9}{*}{yelp}         & \multirow{3}{*}{5}     & RECALL                          & \multicolumn{1}{c|}{75}  & \multicolumn{1}{c|}{94}  & \multicolumn{1}{c|}{97}  & \multicolumn{1}{c|}{95}      & \multicolumn{1}{c|}{94}      & \multicolumn{1}{c|}{100}     & 100     & 96                \\ \cline{3-11} 
                              &                        & NDCG                            & \multicolumn{1}{c|}{73}  & \multicolumn{1}{c|}{89}  & \multicolumn{1}{c|}{97}  & \multicolumn{1}{c|}{95}      & \multicolumn{1}{c|}{94}      & \multicolumn{1}{c|}{100}     & 100     & 84                \\ \cline{3-11} 
                              &                        & AP                              & \multicolumn{1}{c|}{71}  & \multicolumn{1}{c|}{87}  & \multicolumn{1}{c|}{97}  & \multicolumn{1}{c|}{94}      & \multicolumn{1}{c|}{94}      & \multicolumn{1}{c|}{100}     & 100     & 80                \\ \cline{2-11} 
                              & \multirow{3}{*}{10}    & RECALL                          & \multicolumn{1}{c|}{88}  & \multicolumn{1}{c|}{95}  & \multicolumn{1}{c|}{100} & \multicolumn{1}{c|}{98}      & \multicolumn{1}{c|}{97}      & \multicolumn{1}{c|}{100}     & 100     & \textbf{100}      \\ \cline{3-11} 
                              &                        & NDCG                            & \multicolumn{1}{c|}{82}  & \multicolumn{1}{c|}{94}  & \multicolumn{1}{c|}{98}  & \multicolumn{1}{c|}{96}      & \multicolumn{1}{c|}{96}      & \multicolumn{1}{c|}{100}     & 100     & \textbf{100}      \\ \cline{3-11} 
                              &                        & AP                              & \multicolumn{1}{c|}{76}  & \multicolumn{1}{c|}{94}  & \multicolumn{1}{c|}{97}  & \multicolumn{1}{c|}{95}      & \multicolumn{1}{c|}{95}      & \multicolumn{1}{c|}{100}     & 100     & 94                \\ \cline{2-11} 
                              & \multirow{3}{*}{20}    & RECALL                          & \multicolumn{1}{c|}{100} & \multicolumn{1}{c|}{100} & \multicolumn{1}{c|}{100} & \multicolumn{1}{c|}{100}     & \multicolumn{1}{c|}{100}     & \multicolumn{1}{c|}{100}     & 100     & \textbf{100}      \\ \cline{3-11} 
                              &                        & NDCG                            & \multicolumn{1}{c|}{94}  & \multicolumn{1}{c|}{98}  & \multicolumn{1}{c|}{100} & \multicolumn{1}{c|}{100}     & \multicolumn{1}{c|}{100}     & \multicolumn{1}{c|}{100}     & 100     & \textbf{100}      \\ \cline{3-11} 
                              &                        & AP                              & \multicolumn{1}{c|}{82}  & \multicolumn{1}{c|}{94}  & \multicolumn{1}{c|}{99}  & \multicolumn{1}{c|}{97}      & \multicolumn{1}{c|}{96}      & \multicolumn{1}{c|}{100}     & 100     & 98                \\ \hline
                              &                        &                                 & \multicolumn{7}{c|}{sample set size n = 1000}                                                                                                                                         & size 700$\sim$900 \\ \hline
\multirow{9}{*}{ml-20m}       & \multirow{3}{*}{5}     & RECALL                          & \multicolumn{1}{c|}{100} & \multicolumn{1}{c|}{100} & \multicolumn{1}{c|}{100} & \multicolumn{1}{c|}{100}     & \multicolumn{1}{c|}{100}     & \multicolumn{1}{c|}{100}     & 100     & \textbf{100}      \\ \cline{3-11} 
                              &                        & NDCG                            & \multicolumn{1}{c|}{96}  & \multicolumn{1}{c|}{100} & \multicolumn{1}{c|}{100} & \multicolumn{1}{c|}{100}     & \multicolumn{1}{c|}{100}     & \multicolumn{1}{c|}{98}      & 98      & \textbf{100}      \\ \cline{3-11} 
                              &                        & AP                              & \multicolumn{1}{c|}{91}  & \multicolumn{1}{c|}{100} & \multicolumn{1}{c|}{100} & \multicolumn{1}{c|}{100}     & \multicolumn{1}{c|}{100}     & \multicolumn{1}{c|}{96}      & 96      & \textbf{100}      \\ \cline{2-11} 
                              & \multirow{3}{*}{10}    & RECALL                          & \multicolumn{1}{c|}{100} & \multicolumn{1}{c|}{100} & \multicolumn{1}{c|}{100} & \multicolumn{1}{c|}{100}     & \multicolumn{1}{c|}{100}     & \multicolumn{1}{c|}{100}     & 100     & \textbf{100}      \\ \cline{3-11} 
                              &                        & NDCG                            & \multicolumn{1}{c|}{100} & \multicolumn{1}{c|}{100} & \multicolumn{1}{c|}{100} & \multicolumn{1}{c|}{100}     & \multicolumn{1}{c|}{100}     & \multicolumn{1}{c|}{100}     & 100     & \textbf{100}      \\ \cline{3-11} 
                              &                        & AP                              & \multicolumn{1}{c|}{98}  & \multicolumn{1}{c|}{100} & \multicolumn{1}{c|}{100} & \multicolumn{1}{c|}{100}     & \multicolumn{1}{c|}{100}     & \multicolumn{1}{c|}{100}     & 100     & \textbf{100}      \\ \cline{2-11} 
                              & \multirow{3}{*}{20}    & RECALL                          & \multicolumn{1}{c|}{100} & \multicolumn{1}{c|}{100} & \multicolumn{1}{c|}{100} & \multicolumn{1}{c|}{100}     & \multicolumn{1}{c|}{100}     & \multicolumn{1}{c|}{100}     & 100     & \textbf{100}      \\ \cline{3-11} 
                              &                        & NDCG                            & \multicolumn{1}{c|}{100} & \multicolumn{1}{c|}{100} & \multicolumn{1}{c|}{100} & \multicolumn{1}{c|}{100}     & \multicolumn{1}{c|}{100}     & \multicolumn{1}{c|}{100}     & 100     & \textbf{100}      \\ \cline{3-11} 
                              &                        & AP                              & \multicolumn{1}{c|}{100} & \multicolumn{1}{c|}{100} & \multicolumn{1}{c|}{100} & \multicolumn{1}{c|}{100}     & \multicolumn{1}{c|}{100}     & \multicolumn{1}{c|}{100}     & 100     & \textbf{100}      \\ \hline
\end{tabular}
}
\end{table}

\subsection{Sucess of User Sampling (Q5)}
In this subsection, we empirically show the results of user sampling. \cref{tab:recall_sample_user,tab:ndcg_sample_user} compares user sampling method with the estimator $MN\_MES$ in terms of $T_{Recall@K}$ and $T_{NDCG@K}$ ($K$ from $1$ to $50$). In general, we could conclude that even with small portion of users, for example 1000 ($0.7\%$) sampled users compared to its total 137,000 users for \textit{ml-20m} dataset, the user-sampling based method could achieve pretty accurate results ($4\%\sim 8\%$ relative errors for $Top-K$ from 1 to 50). In addition, as the size of sampled user increase, it could be significantly close to the true results. For instance, with $10000$ ($7\%$) users be sampled for \textit{ml-20m} dataset, it can achieve as small as $1\%$ relative errors. This consistent empirical results together with demo example in \Cref{eq:moe_example} indicates the effectiveness of user sampling.
Noting that, according to \Cref{eq:m_moe}, the accuracy is not quite related with user size $M$, which suggest that for some very huge dataset (for example, $M>>1$ million and $M>>N$, this is quite common in e-commerce), user-sampling based estimation can be more practical and fundamentally efficient than item-sampling based estimation.

\begin{table}[]
\centering
\caption{The average relative errors between estimated $\hat{T}_{Recall@K}$ ($K$ from $1$ to $50$) and the ground true ones. Unit is $\%$. $item 100$ and  $item 500$ are the results of item-sampling based estimator $MN\_MES$ with sample set size $n=100$ and $n=500$. $user 1K$, etc are the results of user-sampling based unbiased average estimation (from \Cref{sec:user_sampling}) with $1K$ user sampled.}
\label{tab:recall_sample_user}
\resizebox{0.9\textwidth}{!}{%
\begin{small}
\begin{tabular}{|c|c|llccc|}
\hline
\multirow{2}{*}{dataset}      & \multirow{2}{*}{model} & \multicolumn{5}{c|}{Recall}                                                                                                                                \\ \cline{3-7} 
                              &                        & \multicolumn{1}{c|}{item 100}     & \multicolumn{1}{c|}{item 500}     & \multicolumn{1}{c|}{user 1K}      & \multicolumn{1}{c|}{user 5K}     & user 10K    \\ \hline
\multirow{5}{*}{pinterest-20} & EASE                   & \multicolumn{1}{l|}{5.00$\rpm$1.39}  & \multicolumn{1}{l|}{2.78$\rpm$1.05}  & \multicolumn{1}{c|}{9.04$\rpm$4.32}  & \multicolumn{1}{c|}{3.85$\rpm$1.76} & 2.65$\rpm$1.33 \\ \cline{2-7} 
                              & MultiVAE               & \multicolumn{1}{l|}{2.75$\rpm$1.66}  & \multicolumn{1}{l|}{2.60$\rpm$1.30}  & \multicolumn{1}{c|}{9.54$\rpm$4.41}  & \multicolumn{1}{c|}{4.34$\rpm$1.94} & 3.18$\rpm$1.35 \\ \cline{2-7} 
                              & NeuMF                  & \multicolumn{1}{l|}{3.60$\rpm$1.50}  & \multicolumn{1}{l|}{2.76$\rpm$1.37}  & \multicolumn{1}{c|}{10.43$\rpm$4.56} & \multicolumn{1}{c|}{4.66$\rpm$2.17} & 3.11$\rpm$1.39 \\ \cline{2-7} 
                              & itemKNN                & \multicolumn{1}{l|}{5.09$\rpm$1.15}  & \multicolumn{1}{l|}{2.79$\rpm$1.12}  & \multicolumn{1}{c|}{8.91$\rpm$4.27}  & \multicolumn{1}{c|}{3.64$\rpm$1.52} & 2.65$\rpm$1.22 \\ \cline{2-7} 
                              & ALS                    & \multicolumn{1}{l|}{2.86$\rpm$1.27}  & \multicolumn{1}{l|}{3.02$\rpm$1.32}  & \multicolumn{1}{c|}{10.24$\rpm$4.96} & \multicolumn{1}{c|}{4.19$\rpm$2.05} & 3.25$\rpm$1.39 \\ \hline
\multirow{5}{*}{yelp}         & EASE                   & \multicolumn{1}{l|}{4.84$\rpm$2.17}  & \multicolumn{1}{l|}{3.47$\rpm$1.79}  & \multicolumn{1}{c|}{12.21$\rpm$5.98} & \multicolumn{1}{c|}{4.97$\rpm$2.16} & 3.62$\rpm$1.84 \\ \cline{2-7} 
                              & MultiVAE               & \multicolumn{1}{l|}{4.30$\rpm$1.27}  & \multicolumn{1}{l|}{4.68$\rpm$2.27}  & \multicolumn{1}{c|}{15.70$\rpm$6.78} & \multicolumn{1}{c|}{6.58$\rpm$2.51} & 4.45$\rpm$1.83 \\ \cline{2-7} 
                              & NeuMF                  & \multicolumn{1}{l|}{13.46$\rpm$2.43} & \multicolumn{1}{l|}{5.48$\rpm$1.43}  & \multicolumn{1}{c|}{12.45$\rpm$6.75} & \multicolumn{1}{c|}{5.69$\rpm$2.96} & 4.23$\rpm$2.00 \\ \cline{2-7} 
                              & itemKNN                & \multicolumn{1}{l|}{7.74$\rpm$2.08}  & \multicolumn{1}{l|}{2.92$\rpm$1.60}  & \multicolumn{1}{c|}{11.62$\rpm$6.01} & \multicolumn{1}{c|}{4.59$\rpm$1.77} & 3.33$\rpm$1.63 \\ \cline{2-7} 
                              & ALS                    & \multicolumn{1}{l|}{11.68$\rpm$0.88} & \multicolumn{1}{l|}{9.29$\rpm$2.03}  & \multicolumn{1}{c|}{15.36$\rpm$6.79} & \multicolumn{1}{c|}{6.54$\rpm$2.18} & 4.46$\rpm$1.71 \\ \hline
\multirow{5}{*}{ml-20m}       & EASE                   & \multicolumn{1}{l|}{9.07$\rpm$0.61}  & \multicolumn{1}{l|}{2.31$\rpm$0.44}  & \multicolumn{1}{c|}{4.48$\rpm$2.07}  & \multicolumn{1}{c|}{1.92$\rpm$0.91} & 1.33$\rpm$0.68 \\ \cline{2-7} 
                              & MultiVAE               & \multicolumn{1}{l|}{9.82$\rpm$0.28}  & \multicolumn{1}{l|}{4.60$\rpm$0.99}  & \multicolumn{1}{c|}{5.97$\rpm$2.57}  & \multicolumn{1}{c|}{2.49$\rpm$1.04} & 1.80$\rpm$0.73 \\ \cline{2-7} 
                              & NeuMF                  & \multicolumn{1}{l|}{3.64$\rpm$1.05}  & \multicolumn{1}{l|}{1.63$\rpm$0.82}  & \multicolumn{1}{c|}{5.58$\rpm$2.35}  & \multicolumn{1}{c|}{2.35$\rpm$1.10} & 1.67$\rpm$0.80 \\ \cline{2-7} 
                              & itemKNN                & \multicolumn{1}{l|}{14.16$\rpm$0.68} & \multicolumn{1}{l|}{3.61$\rpm$0.61}  & \multicolumn{1}{c|}{5.27$\rpm$2.73}  & \multicolumn{1}{c|}{2.21$\rpm$1.11} & 1.54$\rpm$0.76 \\ \cline{2-7} 
                              & ALS                    & \multicolumn{1}{l|}{36.32$\rpm$0.56} & \multicolumn{1}{l|}{19.33$\rpm$1.93} & \multicolumn{1}{c|}{7.70$\rpm$2.92}  & \multicolumn{1}{c|}{3.13$\rpm$1.13} & 2.22$\rpm$0.85 \\ \hline
\end{tabular}
\end{small}
}
\end{table}

\begin{table}[]
\centering
\caption{The average relative errors between estimated $\widehat{T}_{NDCG@K}$ ($K$ from $1$ to $50$) and the ground true ones. Unit is $\%$. $item 100$ and  $item 500$ are the results of item-sampling based estimator $MN\_MES$ with sample set size $n=100$ and $n=500$. $user 1K$, etc are the results of user-sampling based unbiased average estimation (from \Cref{sec:user_sampling}) with $1K$ user sampled.}
\label{tab:ndcg_sample_user}
\resizebox{0.9\textwidth}{!}{%
\begin{scriptsize}
\begin{tabular}{|c|c|llccc|}
\hline
\multirow{2}{*}{dataset}      & \multirow{2}{*}{model} & \multicolumn{5}{c|}{NDCG}                                                                                                                                  \\ \cline{3-7} 
                              &                        & \multicolumn{1}{c|}{item 100}     & \multicolumn{1}{c|}{item 500}     & \multicolumn{1}{c|}{user 1K}      & \multicolumn{1}{c|}{user 5K}     & user 10K    \\ \cline{1-7}
\multirow{5}{*}{pinterest-20} & EASE                   & \multicolumn{1}{l|}{9.35$\rpm$3.09}  & \multicolumn{1}{l|}{4.17$\rpm$2.45}  & \multicolumn{1}{c|}{11.02$\rpm$6.81} & \multicolumn{1}{c|}{4.21$\rpm$2.72} & 3.05$\rpm$1.99 \\ \cline{2-7}
                              & MultiVAE               & \multicolumn{1}{l|}{3.13$\rpm$2.08}  & \multicolumn{1}{l|}{3.26$\rpm$2.14}  & \multicolumn{1}{c|}{10.48$\rpm$6.52} & \multicolumn{1}{c|}{4.88$\rpm$2.90} & 3.68$\rpm$2.06 \\ \cline{2-7}
                              & NeuMF                  & \multicolumn{1}{l|}{4.27$\rpm$2.44}  & \multicolumn{1}{l|}{3.24$\rpm$2.30}  & \multicolumn{1}{c|}{11.99$\rpm$6.56} & \multicolumn{1}{c|}{5.13$\rpm$3.11} & 3.55$\rpm$2.04 \\ \cline{2-7}
                              & itemKNN                & \multicolumn{1}{l|}{9.69$\rpm$2.74}  & \multicolumn{1}{l|}{4.23$\rpm$2.47}  & \multicolumn{1}{c|}{10.46$\rpm$6.72} & \multicolumn{1}{c|}{3.92$\rpm$2.42} & 2.96$\rpm$1.68 \\ \cline{2-7}
                              & ALS                    & \multicolumn{1}{l|}{3.70$\rpm$2.00}  & \multicolumn{1}{l|}{3.90$\rpm$2.24}  & \multicolumn{1}{c|}{11.29$\rpm$7.46} & \multicolumn{1}{c|}{4.54$\rpm$3.06} & 3.57$\rpm$1.98 \\ \cline{1-7}
\multirow{5}{*}{yelp}         & EASE                   & \multicolumn{1}{l|}{5.36$\rpm$2.40}  & \multicolumn{1}{l|}{4.03$\rpm$2.53}  & \multicolumn{1}{c|}{12.83$\rpm$8.00} & \multicolumn{1}{c|}{5.56$\rpm$3.32} & 4.02$\rpm$2.74 \\ \cline{2-7}
                              & MultiVAE               & \multicolumn{1}{l|}{4.31$\rpm$1.90}  & \multicolumn{1}{l|}{5.77$\rpm$3.87}  & \multicolumn{1}{c|}{16.69$\rpm$9.69} & \multicolumn{1}{c|}{7.37$\rpm$4.14} & 4.77$\rpm$2.62 \\ \cline{2-7}
                              & NeuMF                  & \multicolumn{1}{l|}{22.50$\rpm$2.33} & \multicolumn{1}{l|}{8.43$\rpm$4.07}  & \multicolumn{1}{c|}{14.24$\rpm$9.33} & \multicolumn{1}{c|}{6.78$\rpm$4.86} & 5.08$\rpm$3.35 \\ \cline{2-7}
                              & itemKNN                & \multicolumn{1}{l|}{10.53$\rpm$2.14} & \multicolumn{1}{l|}{3.65$\rpm$2.28}  & \multicolumn{1}{c|}{12.79$\rpm$7.81} & \multicolumn{1}{c|}{4.83$\rpm$2.66} & 3.56$\rpm$2.54 \\ \cline{2-7}
                              & ALS                    & \multicolumn{1}{l|}{16.91$\rpm$3.33} & \multicolumn{1}{l|}{12.57$\rpm$5.46} & \multicolumn{1}{c|}{16.10$\rpm$9.03} & \multicolumn{1}{c|}{6.91$\rpm$3.06} & 4.60$\rpm$2.37 \\ \cline{1-7}
\multirow{5}{*}{ml-20m}       & EASE                   & \multicolumn{1}{l|}{18.98$\rpm$0.89} & \multicolumn{1}{l|}{5.59$\rpm$1.49}  & \multicolumn{1}{c|}{4.94$\rpm$3.27}  & \multicolumn{1}{c|}{2.24$\rpm$1.40} & 1.57$\rpm$1.04 \\ \cline{2-7}
                              & MultiVAE               & \multicolumn{1}{l|}{16.28$\rpm$1.56} & \multicolumn{1}{l|}{7.01$\rpm$2.18}  & \multicolumn{1}{c|}{6.31$\rpm$3.64}  & \multicolumn{1}{c|}{2.77$\rpm$1.55} & 1.90$\rpm$1.09 \\ \cline{2-7}
                              & NeuMF                  & \multicolumn{1}{l|}{5.67$\rpm$1.24}  & \multicolumn{1}{l|}{2.10$\rpm$1.31}  & \multicolumn{1}{c|}{5.74$\rpm$3.26}  & \multicolumn{1}{c|}{2.68$\rpm$1.80} & 1.91$\rpm$1.27 \\ \cline{2-7}
                              & itemKNN                & \multicolumn{1}{l|}{28.66$\rpm$1.00} & \multicolumn{1}{l|}{7.40$\rpm$1.60}  & \multicolumn{1}{c|}{6.04$\rpm$4.00}  & \multicolumn{1}{c|}{2.43$\rpm$1.60} & 1.70$\rpm$1.17 \\ \cline{2-7}
                              & ALS                    & \multicolumn{1}{l|}{52.19$\rpm$2.50} & \multicolumn{1}{l|}{26.31$\rpm$3.81} & \multicolumn{1}{c|}{7.99$\rpm$3.90}  & \multicolumn{1}{c|}{3.26$\rpm$1.59} & 2.33$\rpm$1.20 \\ \cline{1-7}
\end{tabular}
\end{scriptsize}
}
\end{table}

\section{Summary}

In this chapter, we holistically discuss the story of sampling-based top-K recommendation evaluation. Starting from the "inconsistent" phenomenon that was first discovered in \citep{rendle2019evaluation}, we \citep{Li@KDD20} observe and propose the alignment theory in terms of Recall curve in \Cref{sec:map}, at the same time of \citep{krichene@kdd20}. Then to we propose two estimators $MES$ and $MLE$ in \Cref{sec:fix_estimator_learning_rank}, which can not only estimate the global user rank distribution $P(R)$ but also help estimate the global true metric \citep{dong@2023aaai}. Consequently, we propose first item-sampling estimators in  \Cref{sec:fix_estimator_learning_metric} which explicitly optimizes its mean square error with respect to the ground truth assuming the prior of rank distribution is given. And we highlight the subtle difference between the estimators from \cite{krichene@kdd20} and ours and point out the potential issue of the former - failing to link the user size with the variance. \citep{dong@2023aaai} Furthermore, we address the limitation of the current item sampling approach, which typically does not have sufficient granularity to recover the top $K$ global metrics when $K$ is small. We then propose an effective adaptive item-sampling method in \Cref{sec:adaptive_estimator}. In the end, we discuss another sampling evaluation strategy from the user perspective in \Cref{sec:user_sampling}. The experimental evaluation demonstrates the effectiveness of the estimators. Our results provide a solid step towards making item sampling as well as user sampling available for recommendation research and practice.

\chapter{Towards a Better Understanding of Linear Models
for Recommendation}\label{chpt:linear}

\section{Introduction}\label{sec:linear-intro}

Over the last 25 years, we have witnessed a blossom of recommendation algorithms being proposed and developed\cite{charubook,zhang2019deep}. Though the number of (top-n) recommendation algorithms is fairly large, the approaches can be largely classified as neighborhood approaches (including the regression-based approaches), matrix factorization (or latent factors) approaches,  more recent deep learning based approaches, their probabilistic variants, and others\cite{charubook,zhang2019deep}. 
However, there have been some interesting debates on the results being reported by recent deep learning-based approaches:  the experimental results show most of these methods seem to achieve sub-par results when compared with their simple/nonlinear counterparts~\cite{RecSys19Evaluation}. This issue also relates to selecting and tuning baselines ~\cite{Steffen@19DBLP} as well as the choice of evaluation metrics ~\cite{WalidRendle20}, among others. But recent studies also seem to confirm the state-of-the-art linear models, such as SLIM~\cite{slim01} and EASE~\cite{Steck_2019} do obtain rather remarkable results when compared to the more sophisticated counterparts~\cite{DacremaBJ21}.

Intuitively, SLIM and EASE search an item-to-item similarity matrix $W$ so that the user-item interaction (denoted as matrix $X$ with rows and columns corresponding to users and items, respectively) can be recovered by the matrix product: $XW$. In fact, they can be considered as simplified linear auto-encoders~\cite{DBLP:conf/nips/Steck20}, where $W$ serves as both encoder and decoder, denoted as $L_W$, such that $L_W(x_u)=x_u W$ can be used to recover $x_u$ ($x_u$ is the $u$-th row vector of user-item interaction matrix $X$).  
In the meantime, the matrix factorization methods, such as ALS~\cite{hu2008collaborative} and SVD-based approaches~\cite{mfsurvey} have been heavily favored and widely adopted in the industry for recommendation. They aim to discover user and item embedding matrices $P$ and $Q$, where $p_u$, $q_i$ represents the latent factors of user $u$ and item $i$, respectively, such that the user item interaction $x_{ui}$ can be approximated by $q_i^T p_u$.  
Furthermore, there has been a list of {\em low-rank} regression approaches~\cite{DBLP:conf/pakdd/ChristakopoulouK14,LFlow16,DBLP:conf/nips/Steck20} which aim to factorize the similarity matrix $W$ as $AB^T$, such that $XAB^T$ can be used to recover $X$. Here, $XA$ and $B$ also introduce the user and item  matrices, respectively (similar to matrix factorization).

Thus, on the surface, we can see the (reduced-rank) regression is like a special case of matrix factorization~\cite{fism13,LRec16}, and it also has seemingly smaller number of parameters (the size of similarity matrix $W$ is typically much smaller than the user and item latent factor matrix as the number of users tend to be much larger than items). Therefore, the expectation is the regression models are more restricted, and thus less flexible (and expressive) than the matrix factorization approaches. However, the recent results seem to indicate the regression approaches tend to perform better than the matrix factorization approaches in terms of commonly used evaluation criteria, such as Recall and nDCG~\cite{Steck_2019,DBLP:conf/nips/Steck20,DacremaBJ21}.

Is there any underlying factor/reason for the regression approaches to perform better than the matrix factorization approaches?  If and how these two different factorization (low-rank regression vs matrix factorization) relate to one another? To seek the connection and be able to compare/analyze their inherent advantage/weakness, can we unify them under the same framework? 
As most of the deep learning, probabilistic and non-linear approaches all have the core representation  from either factorization~\cite{he2017neural} or auto-encoders~\cite{liang2018variational}, the answer to these questions will not only help understand two of the (arguably) most important recommendation methodologies: neighborhood vs matrix factorization, but also help design more sophisticated deep learning based approaches.  To our surprise, the theoretical analyses of these approaches are still lacking and the aforementioned questions remain unanswered. 

In this chapter, by analyzing two basic (low-rank) regression and matrix factorization models, we are able to derive and compare their closed-form solutions in a unified framework. We reveal that both methods essentially ``scale-down'' their singular values of user-item interaction matrix using slightly different mechanisms. The (low-rank) regression mechanism allows the use of more {\em principal components} (latent dimensions) of the user-item matrix $X$ than the matrix factorization approaches. Thus, this potentially provides an inherent advantage to the former methods over the latter. Another surprising discovery is that although the matrix factorization seems to have more model parameters with both user and item latent factor matrices, its optimal solution for the simplified problem suggests that it is actually only dependent on the item matrix. Thus, it is actually more restricted and less flexible than the regression based approaches.  This again indicates the potential disadvantage of the matrix factorization approaches. 

To help further understand how the singular values of user-item interaction matrix can be adjusted (at individual level), we introduce a novel learning algorithm that can search through high dimensional continuous (hyper)parameter space. This learning algorithm also enables us to perform the post-model fitting exploration ~\cite{guan2018post} for existing linear recommendation models. Our approach is to augment (existing) linear models with additional parameters to help further improve the model's accuracy. The resulting models remain as linear models, which can be considered as {\em nearby} models with respect to the existing models. 
This approach indeed shares a similar spirit of our recently proposed ``next-door analysis'' from statistical learning~\cite{hastie@statisticallearning} though our approaches and targets are quite different. 
To the best of our knowledge, this is the first work to study post-model fitting exploration for recommendation models. Such a study can not only help better evaluate the optimality of the learned models but also (potentially) produce an additional boost for the learned models. 

As pointed out by ~\cite{DacremaBJ21}, a major problem in existing recommendation research is that the authors tend to focus on developing new methods or variants of recommendation algorithms, and then validate based on ``hyper-focus on abstract metrics" with often weak or not-fully-tuned baselines to ``prove'' the progress. Though using better baselines, datasets, and evaluation metrics can help address part of the problem, a better understanding of how, why, and where the improvement over existing ones are being made is equally important. We hope the theoretical analysis, the new learning tool for (hyper)parameter search, and the post-model analysis on linear recommendation models can be  part of the remedy for the aforementioned problem.

To sum up, in this chapter, we made the following contribution: 
\begin{itemize}
\item (\cref{pca}) We theoretically investigate the relationship between the reduced-rank regression (neighborhood) approaches and the popular matrix factorization approach using the closed-form solutions, and reveal how they connect with one another naturally (through the lens of well-known principal component analysis and SVD). We also discover some potential factors which may provide a benefit for the regression-based methods.   
\item (\cref{learning}) We introduce a new learning algorithm to help search the high-dimension (hyper)parameter space for the closed-form solution from \cref{pca}. We further apply the learning algorithm to perform post-model exploration analysis on the existing linear models by augmenting them with additional parameters (as nearby models).   
\item (\cref{experiments}) We experimentally validate the closed-form solution for the basic regression and matrix factorization models, and show their (surprising) effectiveness and accuracy comparing against the state-of-the-art linear models; we also experimentally validate the effectiveness of the learning algorithms for the closed-form solutions and identifying nearby models. We show nearby models can indeed boost the existing models in certain datasets. 
 \end{itemize}

\section{Background}\label{problem}
Let the training dataset consists of $m$ users and $n=|I|$ items, where $I$ is the entire set of items. In this chapter, we will focus on the implicit setting for recommendation. Compared with the explicit settings, the implicit setting has more applications in ecommerce, content recommendation, advertisement, among others. It has also been the main subjects for recent top-$n$ recommendation~\cite{Steck_2019,hu2008collaborative,slim01, RecSys19Evaluation,zhang2019deep}.
Here, the user-item interaction matrix $X$ can be considered as a binary matrix, where $x_{ui}=1$ represents there is an interaction between user $u$ and item $i$. If there is no interaction between $u$ and $i$, then $x_{ui}=0$. 
Let $X_u^+=\{j: x_{uj}>0\}$ denote the item set that user $u$ has interacted with, and $X_u^-=I-X_u^+$ to be the item set that $u$ has not interacted with. 

\subsection{Regression Models for Neighborhood-based Recommendation}
It is well-known that there are user-based and item-based neighborhood based collaborative filtering, and the item-based approach has shown to be more effective and accurate compared with user-based approaches~\cite{DeshpandeK@itemKNN}. Thus, most of the linear models are item-based collaborative filtering (ICF). 

Intuitively, the model-based ICF aims to predict $x_{ui}$ (user $u$'s likelihood of interaction with and/or preference of item $i$) based on user $u$'s past interaction with other items $X_u^+$:
\begin{equation}
\hat{x}_{ui}=\sum_{j \in X_u^+} s_{ji} x_{uj}, 
\end{equation}
where $s_{ji}$ denotes the similarity between item $j$ and $i$. 

The initial neighborhood approach uses the statistical measures, such as Pearson correlation and cosine similarity ~\cite{charubook} between the two columns $X_{*i}$ and $X_{*j}$ from items $i$ and $j$. The more recent approaches have been aiming to use a regression approach to directly learn the weight matrix $W$ (which can be considered as the inferred similarity matrix)  so that $||X-XW||_F^2$ ($||\cdot||_F$ denotes the Frobenius norm) is minimized. Clearly, in this formulation, the default solution $W=I$ should be avoided for generalization purpose, and the difference of different approaches lie in the constraints and regularization putting on $W$. Recent studies have shown these approaches achieve comparable or even better performance compared with the state-of-the-art deep learning based approaches ~\cite{DBLP:conf/nips/Steck20,DacremaBJ21}. 

\noindent{\bf SLIM:}
SLIM~\cite{slim01} is one of the first regression-based approach to infer the weight matrix $W$. It considers $W$ to be nonnegative, and regularizing it with $L_1$ and $L_2$ norm (thus {\em ElasticNet})~\cite{elasticnet05}. In addition, $W$ requires to be zero diagonal:
\begin{equation}
\begin{split}
    W & =\arg \min_{W}   \frac{1}{2}||X-XW||^2_F+ \lambda_1||W||_1 + \lambda_2 ||W||_F^2 \\ 
  & s.t.\ \  W\geq 0, diag(W)=0, \\
\end{split}
\end{equation} 
where $||\cdot||_1$ denotes the $L_1$ matrix norm, and $diag(\cdot)$ denotes the diagonal (vector) of the corresponding matrix. 
Since no closed-form solution for $W$, the solver of ElasticNet is used to optimize $W$, and $\lambda_1$ and $\lambda_2$ are the correspondingly regularization hyperparameters. 

There are a quite few variants of SLIM being proposed, including HOLISM~\cite{hoslim14} which extends SLIM to capture higher-order relationship, and LRec~\cite{LRec16} which considers a non-linear logistic loss (instead of squared loss) with no zero diagonal, no negative and $L1$ constraints, among others.  

\noindent{\bf EASE:}
EASE~\cite{Steck_2019} is a recent regression-based model which has shown to improve over SLIM with both speed and accuracy, and quite competitive again the state-of-the-art deep learning models. It simplifies the constraint and regularization enforced by SLIM by removing non-negative and $L1$ constraints: 
\begin{equation}
\begin{split}
W & =\arg \min_{W}   \frac{1}{2}||X-XW||^2_F+ \lambda ||W||_F^2 \\ 
& s.t. \quad diag(W)=0\\
\end{split}
\end{equation} 
Empirical study ~\cite{Steck_2019} basically confirms that the non-negative constraint and $L_1$ norm on matrix $W$ may not be essential (or have negative impact) on the performance. 
Particularly, EASE has a closed-form solution~\cite{Steck_2019}. 

\noindent{\bf DLAE and EDLAE:}
The latest extension of EASE, the DLAE (Denoising linear autoencoder) ~\cite{DBLP:conf/nips/Steck20} utilizes a drop-out induced the $L_2$ norm to replace the standard $L_2$ norm without zero diagonal constraints: 
\begin{equation} \label{eq:DLAE}
W =\arg\min_{W}\frac{1}{2}||X-XW||^2_F+ ||\Lambda^{1/2}W||_F^2 
\end{equation} 
where $\Lambda=\frac{p}{1-p} diagM(diag(X^T X))$ ($diagM(\cdot)$ denotes the diagonal matrix) and $p$ is the dropout probability. 


Another variant EDLAE would explicitly enforce the zero diagonal constraints: 
\begin{equation}
\begin{split}
    W & =\arg \min_{W}   \frac{1}{2}||X-XW||^2_F+ ||\Lambda^{1/2}W||_F^2 \\ 
  & s.t. \ \  diag(W)=0, \\
  \end{split}
\end{equation}
Both DLAE and EDLAE have closed-form solutions~\cite{DBLP:conf/nips/Steck20}. 

\subsubsection{Low-Rank Regression}
There have been a number of interesting studies ~\cite{fism13,LRec16,DBLP:conf/nips/Steck20} on using low-rank regression to factorize the weight/similarity matrix $W$. The latest work ~\cite{DBLP:conf/nips/Steck20} shows a variety of low-rank regression constraints which have been (or can be) used for this purpose: 
\begin{equation}
\begin{split}
||X-XAB^T||_F^2 & +\lambda (||A||_F^2+||B^T||_F^2) \\
    ||X-XAB^T||_F^2 & +\lambda ||AB^T||_F^2 \\
   ||X-XAB^T||_F^2 & + ||(\Lambda+\lambda I)AB^T||_F^2 
\end{split}
\end{equation}
where $A_{n\times k}$,  $B_{n\times k}$, and thus $rank(AB) \leq k$. 
The reduced-rank EDLAE ~\cite{DBLP:conf/nips/Steck20} further enforces zero diagonal constraints for generalization purpose.  

We note that interestingly, the reduced-rank regression solution naturally introduces a $k$-dimensional vector embedding for each user from $XA$ ($m \times k$) , and a $k$-dimensional vector embedding for each item via $B$. This immediately leads to an important question: how such embedding differs from the traditional matrix factorization (MF) approaches which aim to  explicitly decompose $X$ into two latent factor matrices (for users and items). Note that in the past, the MF methods are more popular and widely used in industry, but the recent researches ~\cite{DacremaBJ21} seem to suggest an edge based on the regression-based (or linear autoencoder) approaches over the MF approaches. Before we formalize our question, let us take a quick review of  MF approaches. 

\subsection{Matrix Factorization Approaches}
Matrix factorization has been been widely studied for recommendation and is likely the most popular recommendation method (particularly showing great success in the Netflix competition~\cite{mfsurvey}). 
The basic formula to estimate the rating is 
\begin{equation}
    \hat{x}_{ui}= p_u \cdot q_i = q_i^T p_u, 
\end{equation}
where $p_u$ and $q_i$ are the corresponding $k$-dimensional latent vectors of user $u$ and item $i$, respectively. 
Below, we review several well-known matrix factorization approaches for implicit settings; they differ on how to treat known vs missing interactions and regularization terms, among others.

\noindent{\bf WMF/ALS:}
The implicit Alternating Least Square (ALS) method~\cite{hu2008collaborative} is basically a weighted matrix factorization (WMF): 
\begin{equation}\label{eq:WMFALS}
     \arg \min_{P,Q} ||C \odot (X- PQ^T)||_F^2 + \lambda (||P||_F^2 + ||Q||_F^2), 
\end{equation}
where $P_{m \times k}$ records all the $k$-dimensional latent vectors for users and $Q_{n \times k}$ records all the item latent vectors, and $\lambda$ regularize  the squared Frobenius norm. $C$ is the weight matrix (for binary data, the known score in $X$ typically has $\alpha$ weight and the missing value has $1$), and $\odot$ is the element-wise product. For the general weight matrix, there is no closed-form solution; the authors thus propose using alternating least square solver to optimize the objective function. 

\noindent{\bf PureSVD:}
The Matrix factorization approach is not only closely related to SVD (singular value decomposition), it is actually inspired by it~\cite{charubook}. In the PureSVD approach, the interaction matrix $X$ is factorized using SVD (due to Eckart-Young theorem)~\cite{CremonesiKT@10}: 
\begin{equation}
    \arg \min_{U_k,\Sigma_k,V_k} ||X-U_k \Sigma_k V_k^T||_F^2,
\end{equation}
where $U_k$ is a $m \times k$ orthonormal matrix, $V_k$ is a $n \times k$ orthonormal matrix, and $\Sigma_k$ is a $k \times k$ diagonal matrix containing the first $k$ singular values. 
Thus the user factor matrix can be defined as $P=U_k \Sigma_k$ and the item factor matrix is $Q=V_k$. 

\noindent{\bf SVD++:}
SVD++ ~\cite{Koren08} is another influential matrix factorization approach which also integrate the neighborhood factor. It nicely combines the formulas of  factorization and neighborhood approaches with generalization. It targets only positive user-item ratings and typically works on explicit rating prediction. 

\subsection{The Problem}
As we mentioned earlier, a few recent studies ~\cite{DBLP:conf/nips/Steck20,DacremaBJ21} seem to indicate the regression based approach (or linear autoencoder approach) seem to have better performance than the popular matrix factorization approach, such as ALS~\cite{HuSWY18}. However, if we look at the reduced rank regression approach, we observe its solution can be considered a special case of matrix factorization. Another interesting question is on the regularization hyperparameter, $\lambda$: ALS typically use a much smaller regularization penalty compared with the one used in the regression based approach, such as EASE and low-rank version. The latter's  $\lambda$ value is typically very large, in the range of thousands and even tens of thousands or ~\cite{Steck_2019,DBLP:conf/nips/Steck20}. Note that both aim to regularize the squared Frobenius matrix norm. What results in such discrepancy?  

Another interesting problem is about model complexity of these two approaches. The regression-based (linear auto-encoder) approach uses the similarity matrix $W$ (which has $n\times n$ parameters), and when using low-rank regression, its parameters will be further reduced to $O(n \times k)$ where $k$ is the reduced rank. The MF has both user and item matrices, and thus has $O((m+n)k)$ parameters. This seems to indicate the MF approach should be more flexible than the regression approaches as it tends to have much more parameters (due to number of users is typically much larger than the number of items). But is this the case?   


In this study, our focus is not to experimentally compare these two types of recommendation approaches, but instead to have a better theoretical understanding their differences as well as their connections. Thus, we hope to  understand why and how if any approach maybe more advantageous than the other and along this, we will also investigate why the effective range of their regularization hyper-parameters are so different.
We will also investigate how to learn high dimensional (hyper)parameters and apply it to help perform post-model exploration to learn "nearby" models.

\section{Theoretical Analysis} 
\label{pca}
In this section, we will theoretically investigate the regression and matrix factorization models, and explore their underlying relationships, model complexity, and explain their discrepancy on regularization parameters. 

\subsection{Low-Rank Regression (LRR) Models}
To facilitate our discussion, we consider the following basic low-rank regression models: 
\begin{equation}
\label{eq:lowrank}
    W=\arg \min_{rank(W)\leq k} ||X -XW||_F^2+ 
    ||\mathbf{\Gamma} W||_F^2,
\end{equation}
where  $\mathbf{\Gamma}$ Matrix regularizes the squared Frobenius norm of $W$. (This can be considered as the generalized ridge regression, or multivariant Tikhonov regularization)~\cite{vanwieringen2020lecture}.
For the basic case, $\mathbf{\Gamma}^T \mathbf{\Gamma} =\lambda I$, and $\mathbf{\Gamma}^T \mathbf{\Gamma} =\Lambda=\frac{p}{1-p} diagM(diag(X^T X))$ for DLAE (eq ~\ref{eq:DLAE}) ~\cite{DBLP:conf/nips/Steck20}.
Note that this regularization does not include the zero diagonal requirement for $W$. As we will show in Section~\ref{experiments}, enforcing it only provides minor improvement and thus the basic model can well capture  the essence of (low-rank) regression based recommendation.

To help derive the closed-form solution for above problem, let us represent it as a standard regression problem. 
\begin{equation*}
\overline{Y}= \begin{bmatrix} X  \\ 0    \end{bmatrix} \ \ \ \ \overline{X}=\begin{bmatrix} X  \\ \boldsymbol{\Gamma}   \end{bmatrix}
\end{equation*}

Given this, the original problem can be rewritten as:
\begin{equation*} 
\begin{split}
    \min_{rank(W)\leq k}& ||\overline{Y}-\overline{X}W||_F^2 = \\
    \min_{rank(W) \leq k}& ||\overline{Y}-\overline{X} W^*||_F^2 +
     ||\overline{X}W^*-\overline{X}W||_F^2, 
\end{split}
\end{equation*}
where $W^*=\arg\min ||\overline{Y}-\overline{X} W^*||_F^2$. Basically, the initial loss $||\overline{Y}-\overline{X}W||_F^2$ is decomposed into two parts: $||\overline{Y}-\overline{X} W^*||_F^2$ (no rank constraint), and $||\overline{X}W^*-\overline{X} W||_F^2$. 
Note this holds since the vector ( $\overline{Y}-\overline{X}W^*$) is orthogonal to $\overline{X}W^*-\overline{X}W=\overline{X}(W^*-W)$ (The optimality of Ordinary Least-Square estimation~\cite{vanwieringen2020lecture}).  

Now, the original problem can be broken into two subproblems: 

\noindent{\bf (Subproblem 1:) item-weighted Tikhonov regularization:}
\begin{equation*}
\begin{split}
       W^*&=\arg\min_W ||\overline{Y}-\overline{X}W^*||_F^2=\arg \min_W ||X -XW||_F^2+||\mathbf{\Gamma} W||_F^2 \\
       &= (\overline{X}^T \overline{X})^{-1}\overline{X}^T \overline{Y}  \\
       &= (X^TX+ \Gamma^T \Gamma)^{-1} X^T X
\end{split}
\end{equation*}

\noindent{\bf (Subproblem 2:) low-rank matrix approximation: }
\begin{equation*}
\begin{split}
 \hat{W} & =\arg\min_{rank(W) \leq k} ||\overline{X}W^*-\overline{X}W||_F^2  \\
   &= \arg \min_{rank(W) \leq k} ||X W^* - XW||_F^2+ 
    ||\mathbf{\Gamma} (W^*- W)||_F^2
\end{split}
\end{equation*}
Let $\overline{Y}^*=\overline{X}W^*$, and based on the well-known {\em Eckart-Young} theorem~\cite{eckart1936approximation}, we have the best rank $k$ approximation of  $\overline{Y}^*$ in Frobenius norm is best represented by SVD. 
Let $\overline{Y}^*=P \Sigma Q^T$ ($P,Q$ are orthogonal matrices and $\Sigma$ is the singular value diagonal matrix, and then the best rank $k$ approximation of $\overline{Y}^*$, denoted as $\overline{Y}^*(k)$ is 
\begin{equation}
    \overline{Y}^*(k) =P_k \Sigma_k Q^T_k,
\end{equation} 
where $M_k$ takes the first $k$ rows of matrix $M$. 
We also have the following equation: 
\begin{equation*}
    P \Sigma Q^T (Q_k Q_k^T) = P_k \Sigma_k Q^T_k
\end{equation*}
Given this, we notice that 
\begin{equation*} 
\begin{split}
    \overline{Y}^*(k) &  = P_k \Sigma_k Q^T_k  =P \Sigma Q^T (Q_k Q_k^T) \\
    & = \overline{X} W^* (Q_k Q_k^T)
      =  \overline{X} W
\end{split}
\end{equation*}
Thus, we have
\begin{equation*}
\widehat{W}=W^* (Q_k Q_k^T)=(X^TX+ \Gamma^T \Gamma)^{-1} X^T X (Q_k Q_k^T), 
\end{equation*}
and the complete estimator for $XD$ (interaction/rating inference) is written as: 
\begin{equation} 
\boxed{
\widehat{W}= (X^TX+ \Gamma^T \Gamma)^{-1} X^T X (Q_k Q_k^T)
}
\end{equation}

Next, let us further simplify it using SVD which can better reveal its ``geometric'' insight. 

\subsubsection{Eigen Simplification}
First, let the SVD of $X$ as 
\begin{equation*}
X=U \Sigma V 
\end{equation*}
When $\Gamma=\Lambda^{1/2} V^T $ where $\Lambda$ is a diagonal matrix, we can observe: 
\begin{proposition}
\begin{equation*}
Q_k=V_k
\end{equation*}
\end{proposition}

\begin{equation*}
\overline{Y}^*=\overline{X} W^*= \overline{X} V(\Sigma^{2}+\Lambda)^{-1}\Sigma^2 V^T 
\end{equation*}
Then, from 
\begin{equation*}
\begin{split}
(\overline{Y}^*)^T \overline{Y}^* = V \Sigma^{-2} (\Sigma^{2}+\Lambda) V^T \overline{X}^T \overline{X} V(\Sigma^{2}+\Lambda)^{-1}\Sigma^2 V^T  \\ 
=V (\Sigma^2 + \Lambda) V^T 
\end{split}
\end{equation*}

Then we have the following: 
\begin{equation*}
\begin{split}
\widehat{W} &= (X^TX+ \Gamma^T \Gamma)^{-1} X^T X (Q_k Q_k^T) \\
& = V (\Sigma^2+\Lambda)^{-1}  \Sigma^2 V^T (V_k V_k^T) \\
& = V diag(\frac{\sigma^2_1}{\sigma^2_1+\lambda_1},\dots,\frac{\sigma^2_k}{\sigma^2_n+\lambda_k}) V^T (V_k V_k^T) 
\end{split}
\end{equation*}

Thus, we have the following closed-form solution:  
\begin{equation}
\label{eq:regressionlambda}
\begin{split}
\boxed{
\widehat{W} = V_k  diag(\frac{\sigma^2_1}{\sigma^2_1+\lambda_1},\dots,\frac{\sigma^2_k}{\sigma^2_k+\lambda_k})  V_k^T
}
\end{split}
\end{equation}

Now, if $\lambda_i=\lambda$, we have: 
\begin{equation}\label{eq:regressionpca}
\boxed{
\widehat{W} = V_k diag(\frac{\sigma^2_1}{\sigma^2_1+\lambda}, \cdots, \frac{\sigma^2_k}{\sigma^2_k+\lambda}) V_k^T 
}
\end{equation}
Note that this special case $\Gamma^T\Gamma=\lambda I$ has indeed been used in ~\cite{LFlow16} for implicit recommendation. However, the authors do not realize that it actually has a closed-form solution. 

We also note that using the matrix factorization perspective, we obtain the user ($P$) and item ($Q$) matrices as: 

\begin{small}
\begin{equation}
\label{eq:regressionmf}
\boxed{
\begin{split}
P=& X V_k diag(\frac{\sigma^2_1}{\sigma^2_1+\lambda},\dots,\frac{\sigma^2_k}{\sigma^2_k+\lambda})  = U_k diag(\frac{\sigma_1}{1+\lambda/\sigma^2_1},\dots,\frac{\sigma_k}{1+\lambda/\sigma^2_k}) \\
Q=& V_k 
\end{split}
}
\end{equation}
\end{small}

\subsection{Matrix Factorization}
To study the relationship between the regression approach and matrix factorization, we consider the basic regularized SVD ~\cite{zheng2018regularized}: 

\begin{equation}
\label{eq:matrixfactorization}
     \arg \min_{P,Q} || X- PQ^T||_F^2 + \lambda^\prime (||P||_F^2 + ||Q||_F^2), 
\end{equation}
The solution for this type problem is typically based on Alternating Least Square, but authors in ~\cite{zheng2018regularized} have found  a closed-form solution. 

Let $X=U \Sigma V^T$, and then let 
\begin{small}
\begin{equation}
\label{eq:mf}
\boxed{
\begin{split}
    P&=U_k diag(\sigma_1-\lambda^\prime, \cdots, \sigma_k-\lambda^\prime)   \\
    &=U_k diag(\sigma_1 (1 -\lambda^\prime/\sigma_1), \cdots, \sigma_k(1-\lambda^\prime/\sigma_k)) \\
    Q&=V_k 
\end{split}
}
\end{equation}
\end{small}

Before we further analyze the relationship between them, we ask the following interesting question: Can matrix factorization be represented as a linear encoder? In other words, we seek if there is an $W$ such that $XW=PQ$ (defined by matrix factorization). Let the Moore-Penrose inverse $X^+=V \Sigma^{-1} U^T$, then we have $\widehat{W} =X^+PQ$, 
\begin{equation}
\label{eq:mfpca}
\boxed{
\begin{split}
    \widehat{W}=V_k diag(1-\lambda^\prime/\sigma_1, 
               \cdots, 1-\lambda^\prime/\sigma_k) V_k^T 
\end{split}
}
\end{equation}

\subsection{Model Comparison and Analysis}
When $\lambda=\lambda^\prime=0$ (no regularization), then both approaches (using the matrix factorization) correspond to the standard SVD decomposition, where $P=U_k \Sigma_k$ and $Q=V_k^T$. 
Further, both approaches will  also have  $\widehat{W}=V_k V_k^T$. Then, let us consider $X\widehat{W}=X V_k V_k^T$, which is indeed our standard principal component analysis (PCA), where $V_k$ serves as a linear map which transforms each row vector $x_i^T$ to the new coordinate under the principal component basis. Clearly, when $\lambda=\lambda^\prime \neq 0$, both models start to diverge and behave differently, and results in their difference in terms of regularization penalties, model complexities and eventually model accuracy. 

\noindent{\bf The Geometric Transformation}
From the matrix factorization perspective, the Formulas ~\ref{eq:regressionmf} and ~\ref{eq:mf} essentially tells us that these two different approaches both scale-down the singular values of the user-item interaction matrix (binary) with slightly different manner: $\frac{1}{1+\lambda/\sigma_i^2}$ for low-rank regression and $1-\lambda/\sigma_i$ for matrix factorization. Figure~\ref{fig:ml-20m-sigma} illustrates the compressed singular value for {\em ML-20M} dataset with LRR corresponds to the low-rank regression and MF corresponds to the matrix factorization. Figure~\ref{fig:ml-20m-pca} illustrates the compression ratios, which also has a direct geometric explanation: if we consider both approaches as the linear auto-encoder (or regression), as being described in Formulas~\ref{eq:regressionpca} and ~\ref{eq:mfpca}, then the comprehension ratios directly scale down the coordinates ($XV_k$) for the principal component basis in  $V_k^T$. 


\begin{figure}%
    \centering
    \subfloat[\centering scaling-down singular values \label{fig:ml-20m-sigma} ]{{\includegraphics[width=.615\linewidth]{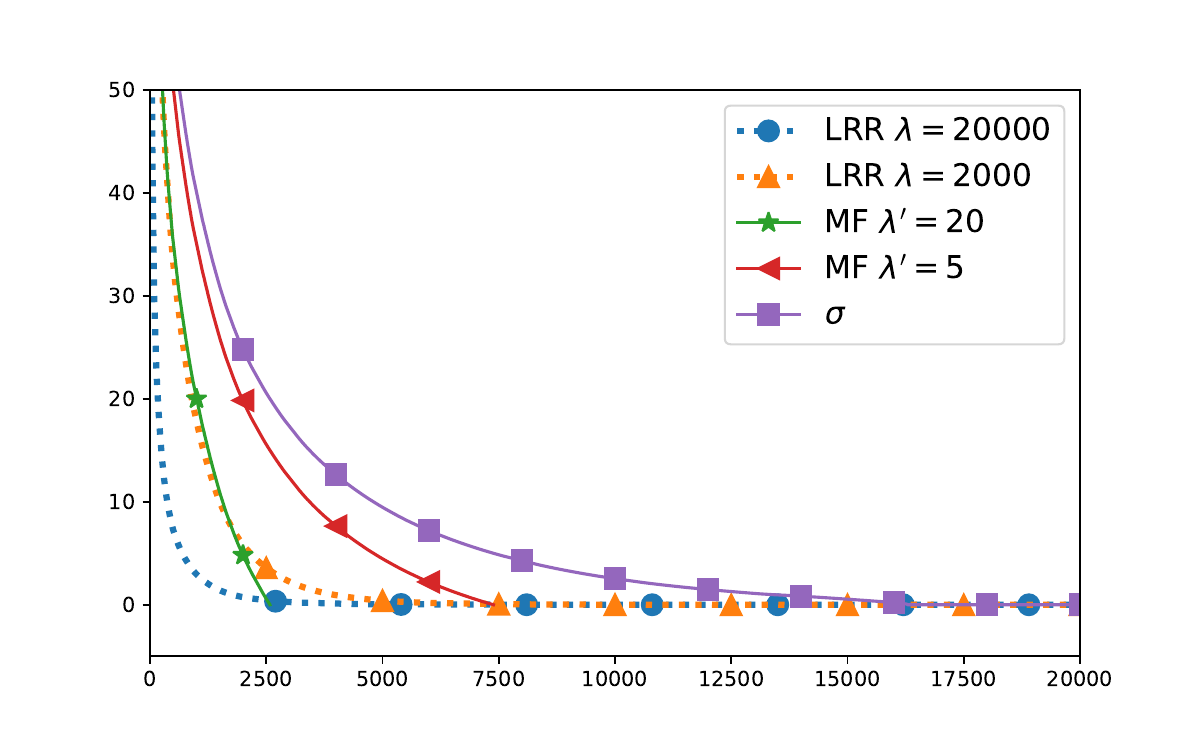} }}%
    \subfloat[\centering compression ratios \label{fig:ml-20m-pca}]{{\includegraphics[width=.385\linewidth]{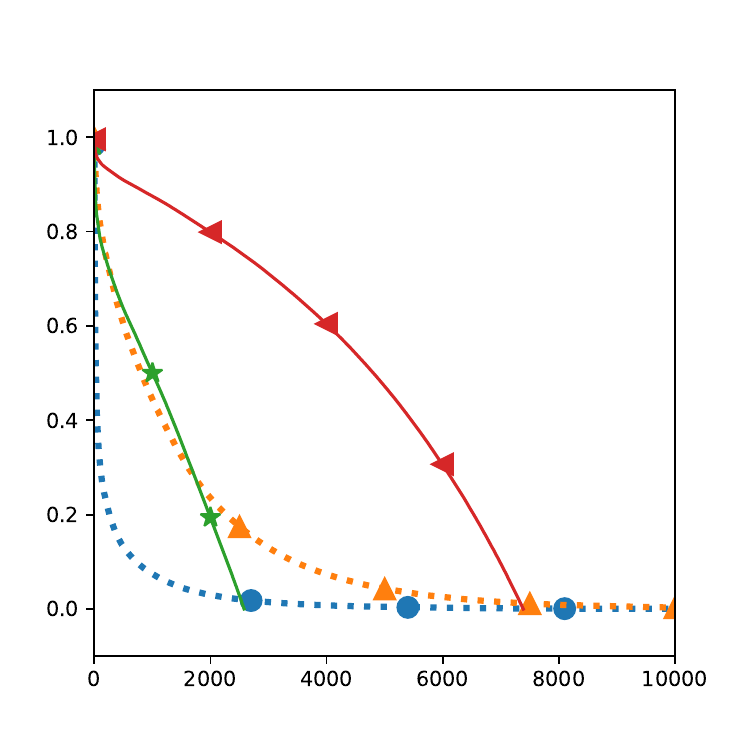} }}%
    \caption{Geometric Transformation for ML-20M.}%
    \label{fig:example}%
\end{figure}

\noindent{\bf The Effective Range of $\lambda$ and $\lambda^\prime$ and Latent Dimension $k$}
For the regression based approach, we note that when $\lambda<\sigma_i^2$, then there is relatively small effect to scale down $\sigma_i$ (or any other singular values larger than $\sigma_i$). 
Given this, $\lambda$ tends to be quite big, and it has close to binary effect: let $\sigma_i^2 \geq \lambda_i > \sigma_{i+1}^2$, then any $\sigma_j<\sigma_i$ has small effect, and  for any $\sigma_j>\sigma_i$, then, the reduction will become more significant (proportionally). Typically, $\lambda$ is within the range of the first few singular values (in other words, $i$ is very small, and $\lambda$ is quite large). 

For the matrix factorization, assuming we consider the $k$ dimensional latent factors, then we note that $\sigma_k> \lambda^\prime$, which effectively limit the range of $\lambda^\prime$. 
Furthermore, we notice that since each singular value is reduced by the same amount $\lambda$, which makes the latent dimensions with smaller singular values are even less relevant than their original value. Thus, this leads to the matrix factorization typically has a relatively small number of dimensions to use (typically less than $1000$). 

For the low-rank regression, as the singular value reduced, its proportion $(\frac{1}{1+\lambda/\sigma_i^2})$ will also be reduced down, the same as the matrix factorization. However, unlike the regularized matrix factorization approach whose absolute reduction may vary: 
\begin{equation}
\Delta_i=\sigma_i-\frac{\sigma_i}{1+\lambda/\sigma_i^2} = \frac{\lambda}{\sigma_i+\lambda/\sigma_i}
\end{equation}
Interesting, when $\sigma_i$ is both very large or very small, its absolute reduction are fairly small, but it may reduce those in between more. Thus, this effectively enables the low-rank approaches to use more latent dimensions, thus larger $k$ (typically larger than $1000$). 

Now, an interesting {\em conjecture} of an optimal set of $\lambda_i$ in Equation~\ref{eq:regressionlambda}, is that they should help (relatively) scale-down those large principal components and help (relatively) scale-up those smaller principal components for better performance. However, how can we search the optimal set of $\lambda_i$ for a large number of $k$? We will introduce a learning algorithm for this purpose in Subsection~\ref{learning}, and then utilize that to study the conjecture (through an experimental study in Section~\ref{experiments}). This help provide a better understanding of the adjustment of singular values.  

\noindent{\bf Model Complexity}
For both types of models, we observe that the gram matrix $X^TX$ serves as the sufficient statistics. This indeed suggest that the complexities of both models (number of effective parameters~\cite{hastie@statisticallearning}) will have no higher than $X^T X$. In the past, we typically consider $P$ and $Q$ (together) are defined as the model parameters for matrix factorization. Thus, the common assumption is that MF has model complexities ($O(mk+nk)$). However, the above analysis based on linear autoencoder/regression  perspective, shows that both models essentially only have $V_k$ (together with the scaled principal components). (See Equations~\ref{eq:regressionpca} and ~\ref{eq:mfpca}) for $W$ estimation. Thus, their model complexity are both $O(nk)$ (but with different $k$ for different models). 

Now relating to the aforementioned discussion on the latent dimension factor, we can immediately observe the model complexity of the basic low-rank regression actually have higher complexity than its corresponding matrix factorization model (as the former can allow larger $k$ than the latter). 

Note that for the real-world recommendation models, the matrix factorization will utilize the weight matrix $C$ (equation \ref{eq:WMFALS}) to increase model complexity (which does not have closed-form solution~\cite{hu2008collaborative}). However, due to the alternating least square solution, its number of effective parameters will remain at $O(nk)$. Thus, it is more restricted and less flexible than the regression based approaches.

\section{Parameter Search and Nearby Models}
\label{learning}

In this section, we first aim to identify a set of {\em optimal} (hyper)parameters $\{\lambda_i:1 \leq i \leq k\}$ for the closed-form solution~\ref{eq:regressionlambda}. Clearly, when the parameter number is small, we can deploy a grid search algorithm as illustrated in Algorithm~\ref{alg:pca} for search $k$ and $\lambda$ for the closed-form in Equation~\ref{eq:regressionpca}. However, for a large number of (hyper)parameters, we have to resort to new approaches (Subsection~\ref{parametersearch}). 
Furthermore, once the parameter searching algorithm is available, we consider to utilize it for searching {\em nearby models} for an existing model (Subsection~\ref{boosting}). As we mentioned before, this can be considered as a post-model fitting exploration in statistical learning~\cite{guan2018post}. 


\begin{algorithm}[t]
\caption{
Hyperparameter Search for Formula~\ref{eq:regressionpca} }\label{alg:pca}
\begin{flushleft}
        \textbf{INPUT:} 
        Hyperparameter candidate lists: $\lambda_l$, $k_l$,
        user-item binary matrix $X$.\\
        \textbf{OUTPUT:} Model performance for all hyperparameter $\lambda_l$, $k_l$ combinations.
\end{flushleft}
\begin{algorithmic}[1]
\STATE $X^TX= V \Sigma^T \Sigma V^T $ \text{ (Eigen Decomposition) }
\FORALL{$\lambda \in$ $\lambda$ list} 

    \STATE\parbox[t]{200pt}{$\Delta \coloneqq (\Sigma^T\Sigma + \lambda I)^{-1}\Sigma^T\Sigma = diag(\frac{d^2_1}{d^2_1+\lambda},\dots,\frac{d^2_n}{d^2_n+\lambda})$}
    \FORALL{$k \in$ $k$ list} 
    \STATE ${}$\hspace{2em} $\Delta_k\leftarrow$ first $k$ columns and rows of $\Delta$
    \STATE ${}$\hspace{2em} $V_k\leftarrow$ first $k$ columns of $V$
    \STATE ${}$\hspace{2em} $W_k\leftarrow$   $V_k\Delta_k V^T_k$
    \STATE ${}$\hspace{2em}evaluate($W_k$) based on nDCG and/or Recall@K
    \ENDFOR
\ENDFOR
\end{algorithmic}
\end{algorithm}

\subsection{Parameter Search}
\label{parametersearch}
In this subsection, we assume the closed-form solution in Equation~\ref{eq:regressionpca} ($
\widehat{W} = V_k diag(\frac{\sigma^2_1}{\sigma^2_1+\lambda},\dots,\frac{\sigma^2_n}{\sigma^2_k+\lambda}) V_k^T$) has optimized hyperparameters $\lambda$ and $k$ through the grid search algorithm Algorithm~\ref{alg:pca}. Our question is how to identify optimized parameter $\lambda_1, \cdots \lambda_k$ in Equation~\ref{eq:regressionlambda}: $
\widehat{W} = V_k  diag(\frac{1}{1+\lambda_1/\sigma^2_1},\dots,\frac{1}{1+\lambda_n/\sigma^2_n})  V_k^T$. 

The challenge is that the dimension (rank) $k$ is fairly large and the typical (hyper)parameter search cannot work in such high dimensional space~\cite{randomhyper12}. Also, as we have the closed-form, it does not make sense to utilize the (original) criterion such as Equation~\ref{eq:lowrank} for optimization. Ideally, we would like to  evaluate the accuracy of any parameter setting (such as in Algorithm~\ref{alg:pca}) based on nDCG or AUC~\cite{charubook}. 
Clearly, for this high dimensional continuous space, this is too expensive.  To deal with this problem, we consider to utilize the {\em BPR} loss function which can be considered as a continuous analogy of AUC~\cite{bpr09}, and parameterize $\lambda_i$ with a search space centered around the optimal $\lambda$ discovered by Algorithm~\ref{alg:pca}:
\begin{equation*}
    \lambda_i(\alpha_i)=\lambda+ c \times tanh(\alpha_i), 
\end{equation*}
where $c$ is the search range, which is typically a fraction of $\lambda$ and $\alpha_i$ is the parameter to be tuned in order to find the right $\lambda_i$. Note that this method effectively provides a bounded search space $(\lambda-c, \lambda+c)$ for each $\lambda_i$. 

Given this, the new objective function based on BPR is:  
\begin{equation*}
    \mathcal{L}=\sum_{u,i\in X_u^+,j \in X_u^-} -\log( \delta(t x_u  (W(\alpha_1,\cdots,\alpha_k)_{*i}-W(\alpha_1,\cdots,\alpha_k)_{*j}))  
\end{equation*}
where $W(\alpha_1,\cdots,\alpha_k)=V_k  diag(\frac{\sigma^2_1}{\sigma^2_1+\lambda_1(\alpha_i)},\dots,\frac{\sigma^2_n}{\sigma^2_n+\lambda_n(\alpha_n)})  V_k^T$, 
and $W(\alpha_1,\cdots,\alpha_k)_{*i}$ is the $i$-th column of $W$ matrix, $x_u$ is the $u-th$ row of matrix $X$ and $t$ is a scaling constant. 
Here $t$ and $c$ are hyper-parameters for this learning procedure.

Note that this is a non-linear loss function for a linear model and the entire loss function can be directly implemented as a simple neural network, and ADAM (or other gradient descent) optimization procedure can be utilized for optimization. We can also add other optimization such as dropout and explicitly enforcing the zero diagonal of $W$~\cite{Steck_2019}. 

\subsection{Nearby Linear Models}
\label{boosting}

In this subsection, we consider how to further leverage the new learning procedure for other linear models   to help identify the (hyper)parameters. Inspired by the recent efforts of post-model fitting exploration~\cite{guan2018post},  we consider to augment the existing learned $W$ from any existing models (or adding on top of the aforementioned closed-form solution) with two types of parameters:
\begin{equation}
\label{eq:augment}
\begin{split}
    W_{HT} &= diagM(H) \cdot W \cdot diagM(T) \\
    W_{S} &= S \odot \widehat{W} \odot (\widehat{W}\geq t)
\end{split}
\end{equation}
where $H=(\delta(h_1), \cdots \delta(h_n))$ and $T=(\delta(t_1),\cdots,\delta(t_n))$ are the {\em head} and {\em tail} vectors with values between $0$ and $1$ (implemented through sigmoid function). We also refer to the diagonal matrices $diagM(H)$ and $diagM(T)$ as the head and tail matrices. Basically, these diagonal matrices $diagM(H)$ and $diagM(T)$ help re-scale the row and column vectors in $W$. 
Furthermore, $S=(\delta(s_{ij}))$ is a matrix with values between $0$ and $1$ (implemented through sigmoid function). Finally, $\widehat{W}\geq t$ is a boolean matrix for sparsification: when $\widehat{W}_{ij}>t$, its element is $1$, otherwise, it it zero.  Thus, this augmented model basically consider to sparsify the learned similar matrix $W$ and re-scale its remaining weights. 
Note that both $W_{HT}$ and $W_S$ can be considered as the {\em nearby} models for the existing models with learned $\widehat{W}$. 
Note that studying these models can also help us understand how close these available learner models are with respect to their limit for recommendation tasks. Since the optimization is more close to the ``true'' objective function, it helps us to squeeze out any potential better models near these existing models. 
In Section~\ref{experiments}, we will experimentally validate if there is any space for improvement based on those simple augmented learning models.

\section{Experimental Results}
\label{experiments}


In this section, we experimentally study the basic linear models as well as the (hyper)parameter search algorithms and its applications to the nearby models. Note that our goal here is not to demonstrate the superiority of these basic/closed-form solutions, but to show they can fare well against the state-of-the-art linear models. This can thus help validate using these basic models to study these advanced linear models~\cite{Steck_2019,DBLP:conf/nips/Steck20,hu2008collaborative}. 
Specifically, we aim to answer: 


\begin{itemize}
    \item (Question 1) How do the basic regression and matrix factorization based models (and their closed-form solutions) compare against the state-of-the-art linear models? Also we hope to compare the two basic models (using their closed-form solutions) to help provide evidence if the matrix factorization approaches have inherently disadvantage for the implicit recommendation task.  
    
    \item (Question 2) How can the learning algorithm to help search the optimal parameter for the closed-form solution of Equation ~\ref{eq:regressionlambda} as well as its augmented models (adding both head and tail matrices)? How does the (augmented) closed-form solution perform against the state-of-the-art methods? We are also interested in understanding how the learned $\{\lambda_i\}$ parameters look like with respect to the constant $\lambda$. 
    
   \item (Question 3)  How does the nearby models based on the head and tail matrices $W_{HT}$ and sparsification $W_S$ introduced in Subsection~\ref{boosting} perform? Can any existing state-of-the-art linear models  be boosted by searching through the augmented nearby models?  
\end{itemize}


\noindent{\bf Experimental Setup:}
We use  three  commonly used datasets for recommendation studies: MovieLens 20 Million (ML-20M) \cite{ml20m}, Netflix Prize (Netflix) \cite{netflix}, and the Million Song Data (MSD)\cite{msddataset}.  The characteristics of first two datasets are in the bottom of Table~\ref{table: maintable}. The characteristics of the  third dataset and its results is in Appendix. 

For the state-of-the-art recommendation algorithms, we consider the following: ALS~\cite{hu2008collaborative} for matrix factorization approaches, SLIM~\cite{slim01}, EASE~\cite{Steck_2019}, and EDLAE~\cite{DBLP:conf/nips/Steck20} for regression models, CDAE ~\cite{cdae16} and MultiVAE~\cite{liang2018variational} for deep learning models.  For most of the experiment settings, we follow ~\cite{liang2018variational,Steck_2019,DBLP:conf/nips/Steck20} for the {\em strong generalization} by splitting the users into training, validation and tests group. 
Also following ~\cite{liang2018variational,Steck_2019,DBLP:conf/nips/Steck20}, we report the results using metrics $Recall@20$, $Recall@50$ and $nDCG@100$. 

Finally, note that our code are openly available (see Appendix).

\begin{table}[]
\resizebox{\linewidth}{!}{%
\begin{tabular}{|c|c|c|c|c|}
\hline
ML-20M &
  \begin{tabular}[c]{@{}c@{}}EASE\\ $\lambda$=400\end{tabular} &
  \begin{tabular}[c]{@{}c@{}}LRR\\ k =2K, $\lambda$ = 10K\end{tabular} &
  \begin{tabular}[c]{@{}c@{}}MF\\ k = 1K, $\lambda$ = 50\end{tabular} &
  \begin{tabular}[c]{@{}c@{}}WMF(ALS)\\ k = 100, C = 10, $\lambda$ = 1e2\end{tabular} \\ \hline
Recall@20 & \textbf{0.39111}  & 0.37635           & 0.36358 & 0.36327  \\ \hline
Recall@50 & \textbf{0.52083} & {0.51144} & 0.50069  & 0.50232\\ \hline
nDCG@100  & \textbf{0.42006} & 0.40760          & 0.39187 & 0.39314 \\ \hline
\end{tabular}%
}
\caption{ML-20M: Basic Model Evaluation}
\label{tab:ml-20mbasic}
\end{table}

\begin{table}[]
\begin{small}
\resizebox{\linewidth}{!}{%
\begin{tabular}{|c|c|c|c|c|}
\hline
Netflix &
  \begin{tabular}[c]{@{}c@{}}EASE\\ $\lambda$= 1000\end{tabular} &
  \begin{tabular}[c]{@{}c@{}}LRR\\ k =3K, $\lambda$ = 40K\end{tabular} &
  \begin{tabular}[c]{@{}c@{}}MF\\ k=1K, $\lambda$=100\end{tabular} &
  \begin{tabular}[c]{@{}c@{}}WMF(ALS)\\ k=100, C = 5, $\lambda$=1e2\end{tabular} \\ \hline
Recall@20 & \textbf{0.36064}  & 0.3478 & 0.33117  & 0.3213 \\ \hline
Recall@50 & \textbf{0.44419}  & 0.4314   & 0.41719   & 0.40629 \\ \hline
nDCG@100  & \textbf{0.39225} & 0.38018 & 0.36462 & 0.35548  \\ \hline
\end{tabular}%
}
\end{small}
\caption{Netflix: Basic Model Evaluation}
\label{tab:netflixbasic}
\end{table}

\noindent{\bf Basic Model Evaluation:} In this experiment, we aim to evaluate the the closed-form (Formulas~\ref{eq:regressionpca}, referred to as $LRR$ and ~\ref{eq:mfpca}, referred to as $MF$)  of the two basic models (Equations~\ref{eq:lowrank} and  ~\ref{eq:matrixfactorization}). We compare them against the state-of-the-art regression model EASE~\cite{Steck_2019} and ALS~\cite{hu2008collaborative}. Since this is mainly for evaluating their prediction capacity (not on how they perform on the real world environment), here we utilize the leave-one-out method to evaluate these models. Note that this actually provides an advantage to the matrix factorization approaches as they prefer to learn the embeddings (latent factors) before its prediction. 

Tables~\ref{tab:ml-20mbasic} and ~\ref{tab:netflixbasic} show the results for these four linear models on the ml-20m and netflix datasets, respectively. We perform a grid search for each of these models (the grid search results are reported in Appendix), and report their better settings (and results) in these tables. From these results, we observe:
(1) Both basic models $LRR$ and $MF$ have very comparable performances against their advanced version. Note that $LRR$ does not have the zero diagonal constraint and use reduced rank regression compared with EASE; and $MF$ does not have the weighted matrix in ALS~\cite{hu2008collaborative}.  This helps confirm the base models can indeed capture the essence of the advanced models and thus our theoretical analysis on these models can help (partially) reflect the behaviors from advanced models.  
(2) Both regression models are consistently and significantly better than the matrix factorization based approaches. This helps further consolidate the observations from other studies~\cite{DacremaBJ21} that the regression methods have the  advantage over the matrix factorization methods. 

\begin{figure}%
    \centering
    \subfloat[\centering optimal adjusted singular value ]{{\includegraphics[width=.615\linewidth]{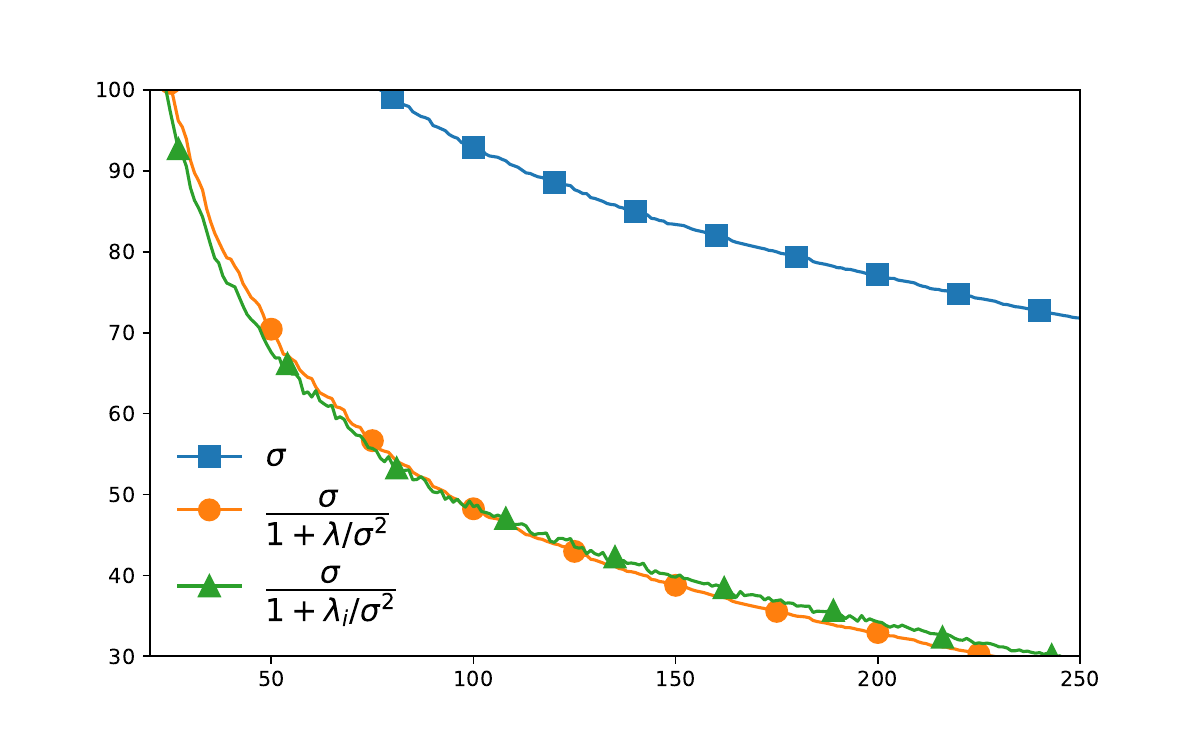} }}%
    \subfloat[\centering optimized (hyper)parameters $\lambda_i$ ]{{\includegraphics[width=.385\linewidth]{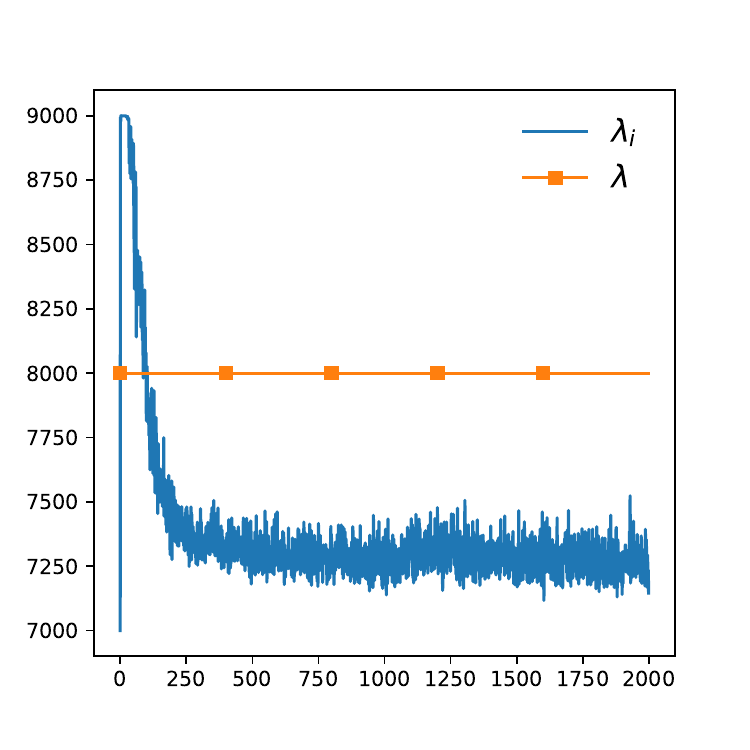} }}%
    \caption{Optimization Results for Formula~\ref{eq:regressionlambda}}%
    \label{fig:parameterresults}%
\end{figure}

\begin{table*}[th]
\resizebox{\textwidth}{!}{%
\begin{scriptsize}
\begin{tabular}{|c|c|c|c|c|c|c|c|}
\hline
\multicolumn{2}{|c|}{\multirow{2}{*}{Model}}   & \multicolumn{3}{c|}{ML-20M} & \multicolumn{3}{c|}{Netflix} \\ \cline{3-8} 
\multicolumn{2}{|c|}{}                         & Recall@20 & Recall@50 & nDCG@100 & Recall@20 & Recall@50 & nDCG@100 \\ \hline
\multicolumn{2}{|c|}{LRR (closed-form) } & 0.376     & 0.513     & 0.406    & 0.347     & 0.432     & 0.380    \\ \cline{3-8} 
\multicolumn{2}{|c|}{LRR + $\lambda_i$}             & 0.380     & 0.515     & 0.410    & 0.348     & 0.433     & 0.381    \\ \cline{3-8} 
\multicolumn{2}{|c|}{LRR + $\lambda_i$ + HT}              & 0.386     & 0.520     & 0.418    & 0.351     & 0.435     & 0.384    \\ \cline{3-8} 
\multicolumn{2}{|c|}{LRR + $\lambda_i$ + HT + RMD}          & 0.386     & 0.522     & 0.418    & 0.351     & 0.435     & 0.384    \\ \hline
\multicolumn{2}{|c|}{EDLAE}          & 0.389     & 0.521     & 0.422    & 0.362     & 0.446     & 0.393    \\ \cline{3-8} 
\multicolumn{2}{|c|}{EDLAE + $HT$}              & 0.394     & 0.527     & 0.424    & 0.361     & 0.446     & 0.393    \\ \cline{3-8} 
\multicolumn{2}{|c|}{EDLAE Full Rank}                & 0.393     & 0.523     & 0.424    & 0.364     & 0.449     & 0.397    \\ \cline{3-8} 
\multicolumn{2}{|c|}{EDLAE Full Rank + $HT$}          & 0.395     & 0.527     & 0.426    & 0.364     & 0.449     & 0.396    \\ \cline{3-8} 
\multicolumn{2}{|c|}{EDLAE Full Rank + Sparsification} & 0.394     & 0.526     & 0.423    & 0.365     & 0.450     & 0.397    \\ \hline
\multicolumn{2}{|c|}{SLIM}                     & 0.370     & 0.495     & 0.401    & 0.347     & 0.428     & 0.379    \\ \cline{3-8} 
\multicolumn{2}{|c|}{ALS/WMF}                      & 0.363    & 0.502     & 0.393    & 0.321     & 0.406     & 0.355    \\ \cline{3-8} 
\multicolumn{2}{|c|}{EASE}                     & 0.391     & 0.521     & 0.420    & 0.360     & 0.444     & 0.392    \\ \cline{3-8} 
\multicolumn{2}{|c|}{CDAE}                     & 0.391     & 0.523    & 0.418   & 0.343     & 0.428    & 0.376    \\ \cline{3-8} 
\multicolumn{2}{|c|}{MULT-DAE}                     & 0.387     & 0.524     & 0.419   & 0.344     & 0.438     & 0.380   \\ \cline{3-8} 
\multicolumn{2}{|c|}{MULT-VAE}                 & 0.395     & 0.537     & 0.426    & 0.351     & 0.444     & 0.386    \\ \hline
\multirow{3}{*}{dataset statistics} & \# items & \multicolumn{3}{c|}{20108}  & \multicolumn{3}{c|}{17769}   \\
               & \# users                      & \multicolumn{3}{c|}{136677}      & \multicolumn{3}{c|}{463435}      \\
               & \# interactions               & \multicolumn{3}{c|}{10 millions}       & \multicolumn{3}{c|}{57 millions}       \\ \hline
\end{tabular}%
\end{scriptsize}
}
\caption{The performance comparison between different models. $\lambda_i$:learned (hyper)parameters; HT: augmented models with head and tail parameter matrix; RMD: with removing the diagonal matrix (enforcing zero diagonal). For more details of the experimental set-up and model, please refer to the appendix.}
\label{table: maintable}
\end{table*}

\noindent{\bf Optimizing Closed-Form Solutions:}
In this and next experiment, we will follow the strong generalization setting by splitting the users into training, validation and testing groups. The top section of Table~\ref{table: maintable} shows the experimental results of using the closed-form solution (Formula~\ref{eq:regressionlambda}). Here, (1) $LRR(closed-form)$ is the starting point for $\lambda$ being constant; (2) $LRR+\lambda_i$ utilizes the BPR learning algorithm in Subsection~\ref{learning} to search the hyperparameter space; (3) $LRR+\lambda_i+HT$ uses $diagM(H)\cdot W \cdot diagM(T)$ (as the targeted similarity matrix), where $W$ is defined in Formula~\ref{eq:regressionlambda} (here the optimization will simultaneously search hyperparameters $\{\lambda_i\}$ and head ($H$), tail ($T$) vectors; (4) finally, $LRR+\lambda_i+HT+RMD$ further enforces the zero diagonal constraints. We also add dropout (with dropout rate $0.5$) for the model training for models ($2-4$). 

We observe the variants of the closed-form solutions are comparable against the state-of-the-art linear models and deep learning models. For instance, on ML-20M, $LRR+\lambda_i+HT+RMD$ reports $0.522$ $Recall@50$, is among the best for the existing linear models (without additional boosting from the augmented nearby models). 

Finally, Figure~\ref{fig:parameterresults} illustrates the parameter search results for  Formula~\ref{eq:regressionlambda} from the learning algorithm.  Specifically, Figure~\ref{fig:parameterresults} (a) shows how the singular value are adjusted vs the compressed singular value for a constant $\lambda=8000$ (Formula~\ref{eq:regressionpca}). We provide $c = 1000$ to allow each individual $\lambda_i$ search between $7000$ to $9000$. Figure~\ref{fig:parameterresults} (b) shows the search results for the parameters $\lambda_i$. 
As we conjectured, we can see that the initial $\lambda_i$ is quite large which leads to smaller singular values compared with adjusted singular value from Formula~\ref{eq:regressionpca}. Then the parameters $\lambda_i$ reduces which make the smaller singular values reduced less. This can help more (smaller) singular values to have better presence in the final prediction.  

\noindent{\bf Nearby Models}
In this experiment, we augment the latest regression models  EDLAE (full rank and reduced rank) ~\cite{DBLP:conf/nips/Steck20} with additional parameters and apply the parameter learning algorithm to optimize the parameters: 
(1) $EDLAE$ is the original reduced rank regression with rank $k=1000$; (2) $EDLAE+HT$ corresponds to the augmented model with head and tail matrices, $W_{HT}$ from Formula~\ref{eq:augment}; (3) $EDLAE\ Full\ Rank$ is the original full rank regression; (4) $EDLAE\ Full\ Rank+HT$ applies the head and tail matrices on the learned similarity matrix from $EDLAE\ Full\ Rank$; (5) $EDLAE\ Full\ Rank+Sparsification$ applies the $W_S$ from Formula~\ref{eq:augment}, which sparsifies the similarity matrix of $EDLAE\ Full\ Rank$ with additional parameters in matrix $S$ to further adjust those remaining entries in the similarity matrix. 

The experimental results on ML-20M and Netflix of these augmented (nearby) models are listed in the middle section in Table~\ref{table: maintable}. We can see that on the ML-20M dataset, the $Recall@50$ has close to $1\%$ boost while other metrics has small improvement. This indeed demonstrates the nearby models may provide non-trivial improvement over the existing models. On the Netflix dataset, the nearby models only have minor changes and indicates the originally learned model may already achieve the local optimum.

\section{Conclusion and Discussion}

In this work, we provide a thorough investigation into the relationship between arguably two of the most important recommendation approaches the neighborhood regression approach and the matrix factorization approach. We show how they inherently connect with each other as well as how they differ from one another. However, our study mainly focuses on the implicit setting: here the goal is not to recover the original ratings (like in the explicit setting), but to recover a ''likelihood'' (or a preference) of the interaction. Thus, the absolute value/rating is not of interest. In fact, for most of the linear regression models, the predicted value can be very small (close to zero than one). What matters here is the relative rank of the predicted scores. Thus it helps to use more latent factors to express the richness of user-item interactions. This can be rather different from the rating recovery, which requires the original singular values to be preserved. 
Especially, the current approaches of explicit matrix factorization which often consider only the positive values, and thus the methodology developed in this work cannot be immediately applied in this setting. Indeed, Koren and Bell in ~\cite{advancedCF} have analyzed the relationship between neighborhood and factorization models under explicit settings. It remains to be seen whether the insights gained here can be applied to the explicit setting.  
 
Also, we would like to point out that this is the first work to investigate the nearby linear models. We consider two basic models which utilize limited additional parameters to help explore the additional models. An interesting question is whether we can explore more nearby models.

Finally, we note that the theoretical models need eigendecomposition which makes them infeasible for the real-world datasets with millions of items. But our purpose here is to leverage such models to help understand the tradeoffs and limitations of linear models, not to replace them.  We hope what being revealed in this work can help design better linear and nonlinear models for recommendation.

\chapter{Revisiting Recommendation Loss Functions through
Contrastive Learning}\label{chpt:ccl}

\section{Introduction}

 Recommendation and personalization play a pivotal role in dealing with information exploration, started from the Internet, User Generated Content, now fused with upcoming AI-generated content.  As recommendation algorithms continue to boom \cite{charubook,zhang2019deep}, largely thanks to deep learning models, many models' real effectiveness and claimed superiority has been called into question, such as various deep learning models vs. simple linear and heuristic methods ~\cite{RecSys19Evaluation}, and the potential  ``inconsistent'' issue of  using item-sampling based evaluation \citep{KricheneR22,li2020sampling}, among others. In fact, a major problem in existing recommendation research is their ``hyper-focus'' on evaluation metrics with often weak or not-fully-tuned baselines to demonstrate the progress of their newly proposed model ~\cite{RecSys19Evaluation}. However, recent research has shown that properly evaluating (simple) baselines is indeed not an easy task~\cite{difficulty@rendle}. 

 To deal with the aforementioned challenges, aside from standardized benchmarks and evaluation criteria, research communities have devoted efforts to developing and optimizing simple and stronger baseline models, such as the latest effort on fine-tuning iALS~\cite{ials_revisiting,ials++} and SimpleX (CCL)~\cite{simplex}. In the meantime, researchers ~\cite{jin@linear} have started to theoretically analyze and compare different models, such as matrix factorization (iALS) vs. low-rank regression (EASE). By simplifying and unifying different models under a similar closed form, they observe how different approaches ``shrink'' singular values differently and thus provide a deep understanding and explain how and why the models behave differently, and the exact factor leads to a performance gap. However, the above work is limited only to linear models, and their regularization landscapes ~\cite{li2022on}. 

Taking the cue from the aforementioned research, this chapter aims to  rigorously analyze and compare the loss functions,  which play a key role in designing recommendation models~\cite{simplex}. This chapter is also particularly driven by the emergence of {\em contrastive learning} loss~\cite{SimCLR,infonce,debiased} has shown great success in various learning tasks, such as computer vision and NLP, etc.  Even though there are a few research works ~\cite{cl-rec,alignment-uniformity} in adopting contrastive learning paradigms into individual recommendation models, a systematic understanding of the recommendation loss function from the lens of contrastive learning is still lacking. For instance, how recommendation models can take advantage of (debiased) contrastive loss~\cite{debiased} (including mutual information-based losses \citep{MINE}? How pairwise BPR~\cite{bpr} (Bayesian Pairwise Ranking) loss relate to these recent contrastive losses~\cite{debiased,infonce,MINE}?  How can we debias the $L2$ (mean-squared-error, MSE) or $L1$ (mean-absolute-error, MAE) loss function with respect to the contrastive learning loss? Following this, does the well-known linear models, such as iALS \citep{ials_revisiting,hu2008collaborative,ials++} and EASE \citep{ease}, needs to be debiased with respect to contrastive learning losses?

To answer these questions, we made the following contribution to this chapter: 
\begin{itemize}
\item (\cref{sec:ccl-softmax}) We revisit the softmax loss (and BPR loss) using the latest debiased contrastive loss~\cite{debiased}, which provides a natural debias approach instead of typically rather sophisticated debias methods through propensity score ~\cite{unbiased-propensity}. We also introduce a mutual information neural estimator (MINE) loss~\cite{MINE} to the recommendation setting. To the best of our knowledge, this is the first study to study MINE for recommendation models.  Experimental results demonstrate its (surprising) performance compared with other known losses.  Finally, through the lower bound analysis, we are able to relate the contrastive learning losses and BPR loss. 
\item (\cref{sec:ccl-linear}) We generalize the debiased contrastive loss ~\cite{debiased} to  the pointwise loss (MSE ~\cite{hu2008collaborative} and CCL~\cite{simplex}) in recommendation models and design the new debiased pointwise loss. We then examine the well-known linear models with pointwise loss, including iALS and EASE; rather surprisingly, our theoretical analysis reveals that both are inherently debiased under the reasonable assumption with respect to the contrastive learning losses.    
\item (\cref{sec:ccl-experiment}) We experimentally validate the debiased InfoNCE and point-wise losses indeed perform better than the biased (commonly used) alternatives. We also show the surprising performance of the newly introduced mutual information-based recommendation loss and its refinement (MINE and MINE+). 
 \end{itemize}
 Finally, Section ~\ref{sec:ccl-related} reviews the related work, and Section ~\ref{sec:ccl-conclusion} concludes. This chapter towards to build a loss function theory for a recommendation based on contrastive learning.

\section{Contrastive Recommendation Loss} 
\label{sec:ccl-softmax}
 In this section, we introduce Debiased InfoNCE~\cite{infonce} and Mutual-Information based Loss (MINE) ~\cite{MINE} to the recommender system settings. Note that, the study in this section, is not to claim those losses are new (they stem from the latest contrastive learning research~\cite{unbiased,debiased}), but to bring them into the recommendation community. Indeed, to our surprise, those losses from contrastive learning have not been fully explored and leveraged by the recommendation models. To reemphasize, our goals are two folds in this section: 1) demonstrating  the benefits of these contrastive learning inspired losses are useful for recommendation models; 2) utilizing the lens of contrastive learning to carefully reexamine the losses recommendation community have designed before towards better understanding and unifying them. 

\noindent{\bf Notation:} 
We use $\mathcal{U}$ and $\mathcal{I}$ to denote the user and item sets, respectively. For an individual user $u\in \mathcal{U}$, we use $\mathcal{I}^+_u$ to denote the set of items that have been interacted (such as Clicked/Viewed or Purchased) by the user $u$, and $\mathcal{I}\backslash \MI^+_u$ to represent the remaining set of items that have not been interacted by user $u$. Similarly, $\MU^+_i$ consists of users interacting with item $i$, and $\MU\backslash \MU_i^+$ includes remaining users.
Often, we denoted $r_{ui}=1$ if item $u$ is known to interact with of item $i$ and $r_{ui}=0$ if such interaction is unknown.

Given this, most of the recommendation systems utilize either matrix factorization (shallow encoder) or deep neural architecture to produce a latent user vector $v_u$ and a latent item vector $v_i$, then use their inner product ($<v_u, v_i>$) or cosine ($<v_u/||v_u||, v_i/||v_i||>$) to produce a similarity measure $\hat{y}_{ui}$. The loss function is then used to produce a (differentiable) measure to quantify and encapsulate all such $\hat{y}_{ui}$. 

\subsection{Basic (Sampled) Softmax and InfoNCE loss}
\noindent{\bf Softmax Loss:} 
In recommendation models, a softmax function is often utilized to transform the similarity measure into a probability measure~\cite{youtube,cl-rec}: $
p(i|u) = \frac{\exp(\hat{y}_{ui})}{\sum_{j \in \mathcal{I}} \exp(\hat{y}_{uj})}$. Given this, the maximal likelihood estimation (MLE) can be utilized to model the fit of the data through the likelihood 
($\Pi_{u\in \mathcal{U}, i \in \MI^+_u} p(i|u)$), and the negative log-likelihood serves as the loss function: 
\begin{equation} 
\mathcal{L}_{soft} = - \E_u  \log \sum_{i \in \mathcal \MI^+_u } \frac{exp(\hat{y}_{ui})}{\sum_{j \in \mathcal{I}} exp(\hat{y}_{uj})}
\end{equation} 

Note that in the earlier research, the sigmoid function has also adopted for transforming similarity measure to probability: $p(i|u) =\sigma(\hat{y}_{ui})$. However, it has shown to be less effective ~\cite{ncfmf}. Furthermore, the loss function is a binary cross-entropy where the negative sample part is eliminated as being treated as $0$. Finally, the main challenge here is that the denominator in the loss sums over all possible items in $\mathcal{I}$, which is often impossible in practice and thus requires approximation.  In the past, the sampled softmax approaches typically utilized the important sampling for estimation~\cite{sample-softmax}: 
\begin{equation}
    \mathcal{L}_{ssoft}= - \E_u \log \sum_{i \in \mathcal \MI^+_u } \frac{exp(\hat{y}_{ui})}{exp(\hat{y}_{ui}) + \sum_{j=1; j \sim p^-_u}^N  exp(\hat{y}_{uj})/p_u^-(j)}
\end{equation}
where, $p_u^-$ is a predefined negative sampling distribution and often is implemented as $p$, which is proportional to the item popularity. It has been shown that sampled softmax performs better than NCE ~\cite{nce} and negative sampling ones ~\cite{bpr,sampling_strategy} when the item number is large~\cite{cl-rec}. 

\noindent{\bf Contrastive Learning Loss (InfoNCE):} 
The key idea of contrastive learning is to contrast semantically similar (positive) and dissimilar (negative) pairs of data points, thus pulling the similar points closer and pushing the dissimilar points apart. There are several contrastive learning losses have been proposed~\cite{SimCLR,infonce,unbiased}, and among them, InfoNCE loss is one of the most well-known and has been adopted in the recommendation as follows: 

\begin{equation}\label{eq:infonce_0}
   \mathcal{L}_{info}= - \E_u \E_{i \sim p_u^+} \log \frac{exp(\hat{y}_{ui})}{exp(\hat{y}_{ui}) + \sum_{j=1; j \sim p_u^{-}}^N  exp(\hat{y}_{uj})}
\end{equation}

where $p_u^+$ ($p_u^-$) is the positive (negative) item distribution for user $u$. 
Note that unlike the sampled softmax, the correction from the sampling distribution is not present in the InfoNCE loss. In fact, this loss does not actually approximate the softmax probability measure $p(i|u)$ while it corresponds to the probability that {\em item $i$ is the positive item, and the remaining points are noise}, and $exp(\hat{y}_{ui})$ measures the density ration ($exp(\hat{y}_{ui}) \varpropto p(i|u)/p(i)$) ~\cite{infonce}. 
Again, the overall InfoNCE loss is the cross-entropy of identifying the positive items correctly, and it also shows maximize the mutual information between user $u$ and item $i$ ($I(i,u)$) and minimize $I(j,u)$ for item $j$ and $u$ being unrelated \citep{infonce}. 

In practice, as the true negative labels (so does all true positive labels) are not available,  negative items $j$ are typically drawn uniformly (or proportional to their popularity) from the entire item set $\mathcal{I}$ or $\mathcal{I}\backslash \MI^+_u$, i.e, $j \sim p_u$: 

\begin{equation}
   \widetilde{\mathcal{L}}_{info}= - \E_u \E_{i \sim p_u^+} \log \frac{exp(\hat{y}_{ui})}{exp(\hat{y}_{ui}) + \sum_{j=1; j \sim p_u}^N  exp(\hat{y}_{uj})}
\end{equation}

Note that this loss $\widetilde{\mathcal{L}}_{info}$ in practice  is sometimes treated as a softmax estimation, despite is not a proper estimator for the softmax loss~\cite{infonce,youtube}. In addition, the $\widetilde{\mathcal{L}}_{info}$ approximation easily leads to {\em sampling bias}: when sampling a large number of supposed negative items, they would include some positive items i.e., those top $K$ items in recommendation setting. It has been known such sampling bias can empirically lead to a significant performance drop for machine learning tasks~\cite{bias-debias,sampling_strategy}.

\vspace*{-2.0ex}
\subsection{Leveraging Debiased Contrastive InfoNCE}\label{subsec:debiased_infonce}
Can we reduce the gap between the approximation $\widetilde{\mathcal{L}}_{info}$ and ideal InfoNCE loss $\mathcal{L}_{info}$ and mitigate sampling bias?
In ~\cite{debiased}, a novel debiased estimator is proposed while only accessing positive items and unlabeled training data. Below, we demonstrate how this approach can be adopted in the recommendation setting. 

Given user $u$, let us consider the implicit (binary) preference probability (class distribution) of $u$: $\rho_u(c)$, where $\rho_u(c=+)=\tau^+_u$ represents the probability of user $u$ likes an item, and $\rho_u(c=-)=\tau^-_u=1-\tau^+_u$ corresponds to being not interested. In  the recommendation setting, we may consider $\rho_u(+)$ as the fraction of positive items, i.e., those top $K$ items, and  $\rho_u(-)$ are the fraction of true negative items.  Let the item class joint distribution be $p_u (i,c)=p_u(i|c)\rho_u(c)$. 
Then $p^+_u(i)=p_u(i|c=+)$ is the probability of observing $i$ as a positive item for user $u$ and $p_u^-(j)=p_u(j|c=-)$ the probability of a negative example. Given this, the distribution of an individual item of user $u$ is:
\begin{equation}
    \begin{split}
        p_u(i) = \tau^+_u \cdot p^+_u(i) + \tau^-_u \cdot p_u^-(i)
    \end{split}
\end{equation}

For implicit recommendation, we can assume the known positive items $\mathcal{I}_u^+$ are being uniformly sampled according to $p^+_u$. For $\tau^+_u$, we can count all the existing positive items ($r_{ui}=1$), and the unknown top $K$ items as being positive, thus, $\tau^+_u=(|\MI_{u}^+|+K)/|\mathcal{I}|$. Alternatively, we can assume the number of total positive items for $u$ is proportional to the number of existing (known) positive items, i.e, $\tau^+_u=(1+\alpha) |\MI_{u}^+|/|\mathcal{I}|$, where $\alpha\geq 0$. As mentioned earlier, the typical assumption of $p_u$ (the probability distribution of selection an item is also uniform). 
Then, we can represent $p^-_{u}$ using the total and positive sampling as: $p^-_{u}(i)=1/\tau^-_u (p_u(i)-\tau^+_u \cdot p^+_u(i))$.

Given this, let us consider the following ideal (debiased) InfoNCE loss: 
\begin{small}
\begin{equation}
\begin{split}
    &\overline{\mathcal{L}}_{Info}=- \E_{u} \E_{i \sim p_u^+}\Bigg[\log\frac{\exp(\hat{y}_{ui})}{\exp(\hat{y}_{ui}) + \lambda \E_{j\sim p_u^-} \exp({\hat{y}_{uj}}) }
    \Bigg] = \\
    &- \E_{u} \E_{i\sim p_u^+}\Bigg[\log\frac{\exp(\hat{y}_{ui})}{\exp(\hat{y}_{ui}) + \lambda \frac{1}{\tau^-}\Big( 
    \E_{t\sim p_u}[\exp(\hat{y}_{ut})] - \tau^+ \E_{j\sim p^+_u}[\exp(\hat{y}_{uj})]\Big) }
    \Bigg]
\end{split}
\end{equation}
\end{small}
Note that when $\lambda=N$, the typical InfoNCE loss $\mathcal{L}_{Info}$ (\cref{eq:infonce_0}) converges into $\overline{\mathcal{L}}_{Info}$ when $N \rightarrow \infty$.
Now, we can use $N$ samples from $p_u$ (not negative samples), an $M$ positive samples from $p_u^+$ to empirically estimate $\overline{\mathcal{L}}_{Info}$ following ~\cite{debiased}: 
\begin{small}
\begin{equation}
\begin{split}
    &\mathcal{L}^{Debiased}_{Info}=-\E_{u} \E_{i \sim p_u^+; j \sim p_u; k \sim p_u^+} \Bigg[\log\frac{\exp(\hat{y}_{ui})}{\exp(\hat{y}_{ui}) + \lambda f(\{j\}_{1}^N, \{k\}_{1}^M})
    \Bigg], where \\
    & f(\{ j\}_{1}^N, \{k\}_{1}^N)= 
    \max \Bigl\{ \frac{1}{\tau_u^-} \Bigl( \frac{1}{N} \sum_{j=1;j\sim p_u}^N exp(f_{uj})- \tau_u^+\frac{1}{M} \sum_{k=1;k\sim p_u^+}^M exp(f_{uk}) \Bigr), e^{1/t}  \Bigr\}
\end{split}
\end{equation}
\end{small}
Here $f$ is constrained to be greater than its theoretical minimum $e^{-1/t}$ to prevent the negative number inside $\log$.
We will experimentally study this debiased InfoNCE in \cref{sec:ccl-experiment}. 

\subsection{Connection with MINE and BPR}
In the following, we first introduce Mutual Information Neural Estimator (MINE) loss and then discuss some common lower bounds for  InfoNCE and MINE, BPR as well as their relationship. Later in the \cref{sec:ccl-experiment}, we evaluate MINE against existing losses for recommendation and show it  performs surprisingly well. To the best of our knowledge,  this work is the first to study and apply MINE \citep{MINE} to the recommendation setting. 




 The Mutual Information Neural Estimator (MINE) loss is introduced in ~\cite{MINE}, which can directly measure the mutual information between user and item: 
 \begin{equation}
    \widehat{I(u,i)}=sup_{(v_u;v_i)} \E_{p_{u,i}}[\hat{y}_{ui}] - \log \E_{p_u \otimes p_i}[e^{\hat{y}_{ui}}]
\end{equation}
Here, the positive $\hat{y}_{ui}$ items sampled $n$ times according to the joint user-item distribution $p_{u,i}$. The negative items are sampled according to the empirical marginal  item distribution of $p_i$, which are proportional to their popularity (or in practice, sampled uniformly).
A simple adaption of MINE to the recommendation setting can simply exclude the positive item $i$ from the denominator of InfoNCE~\cite{infonce}: 

\begin{equation}\label{eq:mine}
\begin{split}
    \mathcal{L}_{mine}& =- \E_u \E_{i\sim p_u^+} \Bigg[\hat{y}_{ui}- \log \E_{j \sim p_i} [\exp \hat{y}_{uj}]\Bigg] 
\end{split}
\end{equation}

Its empirical loss can be written as (the factor of $1/N$ will not affect the optimization): 

\begin{equation}\label{eq:mine}
\begin{split}
    \widetilde{\mathcal{L}}_{mine}& =- \E_u \E_{i\sim p_u^+} \Bigg[\hat{y}_{ui}- \log \sum_{j=1; j \sim p_i}^N [\exp \hat{y}_{uj}]\Bigg] 
\end{split}
\end{equation}

\noindent{\bf Lower Bounds:} 
Given this, we can easily observe (based on Jensen's inequality): 
\begin{equation} 
\begin{split}
    &\widetilde{\mathcal{L}}_{info} = - \E_u \E_{i \sim p_u^+} \log \frac{exp(\hat{y}_{ui})}{exp(\hat{y}_{ui}) + \sum_{j=1; j \sim p_u}^N  exp(\hat{y}_{uj})} \geq \\
    & \widetilde{\mathcal{L}}_{mine} =  
    \E_u \E_{i\sim p_u^+} \log (\sum_{j=1; j\sim p_i}^N exp(\hat{y}_{uj} - \hat{y}_{ui})) \approx \\  
    & \E_u \E_{i\sim p_u^+} \log N \E_{j\sim p_i} exp(\hat{y}_{uj} - \hat{y}_{ui})) \geq  \\
    & \E_u \E_{i\sim p_u^+} \E_{j\sim p_i} \bigl( \hat{y}_{uj} - \hat{y}_{ui}  \bigr) + \log N 
    \geq \E_u \E_{i\sim p_u^+} \E_{j\sim p_i} \bigl( \hat{y}_{uj} - \hat{y}_{ui}  \bigr)
\end{split}
\end{equation}
Alternative, by using LogSumExp lower-bound~\cite{unbiased}, we also observe:  
\begin{equation} 
\begin{split}
    & \E_u \E_{i\sim p_u^+} \max_{j=1; j\sim p_i} ^N \bigl( 0, \hat{y}_{uj} - \hat{y}_{ui}  \bigr)  + \log (N+1) \geq \\
     & \widetilde{\mathcal{L}}_{info} \geq 
     \widetilde{\mathcal{L}}_{mine} =  
    \E_u \E_{i\sim p_u^+} \log (\sum_{j=1; j\sim p_i}^N exp(\hat{y}_{uj} - \hat{y}_{ui})) \geq \\ 
    &\E_u \E_{i\sim p_u^+} \max_{j=1; j\sim p_i} ^N \bigl( \hat{y}_{uj} - \hat{y}_{ui}  \bigr)
\end{split}
\end{equation}
Note that the second lower bound is tighter than the first one. We also observe the looser bound is also shared by the well-known pairwise BPR loss, which can be written as:

\begin{equation}
\begin{split}
&\widetilde{\mathcal{L}}_{bpr} = \E_u \E_{i\sim p_u^+} \sum_{j=1;j\sim p_i}^N -\log \sigma (\hat{y}_{ui} - \hat{y}_{uj})  \\
& =  \E_u \E_{i\sim p_u^+} \sum_{j=1; j\sim p_i}^N \log (1 + exp(\hat{y}_{uj} - \hat{y}_{ui})) 
\end{split}
\end{equation}

Following above analysis, we observe: 
\begin{equation}\label{eq:bound1}
\begin{split}
    &\widetilde{\mathcal{L}}_{info} = \E_u \E_{i \sim p_u^+} \log (1 + \sum_{j=1; j \sim p_i}^N  exp(\hat{y}_{uj}-\hat{y}_{ui})) \approx \\
    & \E_u \E_{i \sim p_u^+} \log (1 + N \E_{j \sim p_i}  exp(\hat{y}_{uj}-\hat{y}_{ui}))
    \geq \\
    & \E_u \E_{i \sim p_u^+} \E_{j \sim p_i} \log (1 + N exp(\hat{y}_{uj}-\hat{y}_{ui})) 
     \geq \\ 
    & \E_u \E_{i \sim p_u^+} \E_{j \sim p_i} \log (1/N + exp(\hat{y}_{uj}-\hat{y}_{ui})) + \log N \geq \\ 
    & \E_u \E_{i \sim p_u^+} \E_{j \sim p_i} \bigl(\hat{y}_{uj}-\hat{y}_{ui}\bigr)) + \log N  \geq \\
    &  \E_u \E_{i \sim p_u^+} \E_{j \sim p_i} \bigl(\hat{y}_{uj}-\hat{y}_{ui}\bigr))
\end{split}
\end{equation}
Thus, the BPR loss shares the same lower bound $E_u E_{i \sim p_u^+} \E_{j \sim p_i} \bigl(\hat{y}_{uj}-\hat{y}_{ui}\bigr))$ as the InfoNCE and MINE, which aims to maximize the average difference between a positive item score ($\hat{y}_{ui}$) and a negative item score ($\hat{y}_{uj}$). 

However, the tighter lower bound of BPR diverges from the InfoNCE and MINE loss using the LogSumExp lower bound: 
\begin{equation}\label{eq:bound2}
\begin{split}
&\widetilde{\mathcal{L}}_{bpr} =  \E_u \E_{i\sim p_u^+} \sum_{j=1; j\sim p_i}^N \log (1 + exp(\hat{y}_{uj} - \hat{y}_{ui})) \geq \\
& \E_u \E_{i\sim p_u^+} \sum_{j=1; j\sim p_i}^N 
max(0, \hat{y}_{uj} - \hat{y}_{ui}) 
\varpropto \\ & \E_u \E_{i\sim p_u^+} \E_{j\sim p_i}
max(0, \hat{y}_{uj} - \hat{y}_{ui})
\end{split}
\end{equation}
Note that the bound in \cref{eq:bound2} only sightly improves over the one in \cref{eq:bound1} and cannot reflect the bounds shared by InfoNCE and MINE, which aims to minimize the largest gain of a negative item (random) over a positive item ($\max_{j=1; j\sim p_i} ^N \bigl( \hat{y}_{uj} - \hat{y}_{ui}  \bigr)$). We also experimentally validate using the first bound $\E_u \E_{i \sim p_u^+} \E_{j \sim p_i} \bigl(\hat{y}_{uj}-\hat{y}_{ui}\bigr)$ leads to poor performance, indicating this can be a potentially underlying factor the relatively inferior performance of BRP comparing with other losses, such as softmax, InfoNCE, and MINE (\cref{sec:ccl-experiment}). 

\noindent{\textbf{MINE+ loss for Recommendation:}}
Utimately, we consider the following refined MINE loss format for recommendation which as we will show in \cref{sec:ccl-experiment}, can provide an additional boost comparing with the Vanilla formula (\cref{eq:mine}): 

\begin{equation}\label{eq:mine+}
\begin{split}
    \widetilde{\mathcal{L}}_{mine+}& =- \E_u \E_{i\sim p_u^+} \Bigg[\hat{y}_{ui}/t- \lambda \log \sum_{j=1; j \sim p_i}^N [\exp  (\hat{y}_{uj}/t]\Bigg] 
\end{split}
\end{equation}
Here, $\hat{y}_{ui}$ uses the cosine similarity between user and item latent vectors, $t$ serves as the temperature  (similar to other contrastive learning loss), and finally, $\lambda$ helps balance the weight between positive scores and negative scores.

\section{Debiased Pointwise Loss}\label{sec:ccl-linear}
Besides the typical listwise (softmax, InfoNCE) and pairwise (BPR) losses, another type of recommendation loss is the pointwise loss. Most of the linear (non-deep) models are based on pointwise losses, including the well-known iALS~\cite{hu2008collaborative,ials_revisiting}, EASE~\cite{ease}, and latest CCL~\cite{simplex}. Here, these loss functions aim to pull the estimated score $\hat{y}_{ui}$ closer to its default score $r_{ui}$. In all the existing losses, when the interaction is known, $i \in \MI^+_u$, $r_{ui}=1$; otherwise, the unlabeled (interaction unknown) items are treated as $r_{ui}=0$. 

Following the analysis of debiased contrastive learning, clearly, some of the positive items are unlabeled, and making them closer to $0$ is not the ideal option. In fact, if the optimization indeed reaches the targeted score, we actually do not learn any useful recommendations. A natural question is whether we can improve  the pointwise losses following the debiased contrastive learning paradigm ~\cite{debiased}. In addition, we will investigate if the popular linear methods, such as iALS and EASE, are debiased or not. 

\subsection{Debiased MSE and CCL}

Let us denote the generic single pointwise loss function $l^+(\hat{y}_{ui})$ and $l^-(\hat{y}_{ui})$, which measure how close the individual positive (negative) item score to their ideal target value. 

\noindent{\bf MSE single pointwise loss:}
Mean-Squared-Error (MSE) is one of the most widely used recommendation loss functions.  Indeed, all the earlier collaborative filtering (Matrix factorization) models and pointwise losses were  based on MSE (with different regularization). Some recent linear models, such as SLIM~\cite{slim} and EASE~\cite{ease}, are also based on MSE. 
Its single pointwise loss function is denoted as: 
\begin{equation}
\begin{cases}
    & l^+_{mse}(\hat{y}_{ui})=(1-\hat{y}_{ui})^2 \\
    & l^-_{mse}(\hat{y}_{ui})=(\hat{y}_{ui})^2
\end{cases}
\end{equation}

\noindent{\bf CCL single pointwise loss:}
The cosine contrastive loss (CCL) is recently proposed loss and has shown to be very effective for recommendation.
The original Contrastive Cosine Loss (CCL) \citep{simplex} can be written as:
\begin{equation}
    \begin{split}
        \mathcal{L}_{CCL}=\E_{u}
        \Bigg(\sum\limits_{i\in \mathcal{I}^+_u} (1-\hat{y}_{ui}) + \frac{w}{N}\sum\limits_{j\in\mathcal{I}}^N ReLU(\hat{y}_{uj}-\epsilon)\Bigg)
    \end{split}
\end{equation}
where $\hat{y}_{uj}$ is cosine similarity, $ReLU(x) = max(0, x)$ is the activation function, $w$ is the negative weight, $N$ is the number of negative samples and $\epsilon$ is margin.
Its single pointwise loss thus can be written as: 
\begin{equation}
\begin{cases}
    & l^+_{ccl}(\hat{y}_{ui})=1-\hat{y}_{ui} \\
    & l^-_{ccl}(\hat{y}_{ui})=ReLU(\hat{y}_{ui}-\epsilon)
\end{cases}
\end{equation}

Given this, we introduce the ideal pointwise loss for recommendation models (following the basic idea of (debiased) contrastive learning) : 

\begin{definition} {\bf (Ideal Pointwise Loss)}
The ideal pointwise loss function for recommendation models is the expected single pointwise loss for all positive items (sampled from $\rho_u(+)=
\tau^+$) and for negative items (sampled from $\rho_u(-)=
\tau^-$):  
\begin{small}
\begin{equation}
    \begin{split}
    \overline{\mathcal{L}}_{point} = \E_{u} \Bigg(\tau_u^+\E_{i\sim p_u^+} l^+(\hat{y}_{ui}) +
    \lambda\tau_u^-\E_{j\sim p_u^-} l^-(\hat{y}_{uj}) \Bigg)
    \end{split}
\end{equation}
\end{small}
\end{definition}

Why this is ideal? When the optimization reaches the potential minimal, we have for all positive items $i$ and all the negative items $j$ reaching all the minimum of individual loss ($l^+$ and $l^-$).  
Here, $\lambda$ help adjusts the balance between the positive and negative losses, similar to iALS ~\cite{hu2008collaborative,ials_revisiting} and CCL~\cite{simplex}. 
However, since $p_u^+$ and $p_u^-$ is unknown, we utilize the same debiasing approach as InfoNCE ~\cite{infonce}, also in \cref{subsec:debiased_infonce}, for the debiased MSE: 

\begin{definition} {\bf (Debiased Ideal pointwise loss)}
\begin{equation}
    \begin{split}
    &\mathcal{L}^{Debiased}_{point}= \E_{u} \Bigl( \tau_u^+\E_{i\sim p_u^+} l^+(\hat{y}_{ui}) + \tau_u^-\lambda \cdot \E_{j\sim p_u^-} l^-(\hat{y}_{ui}) \Bigr)\\
    & = \E_{u} \Big(\tau_u^+\E_{i\sim p_u^+}l^+(\hat{y}_{ui}) + \lambda\Big( 
    \E_{j\sim p_u}l^-(\hat{y}_{ui})  - \tau_u^+ \E_{k\sim k^+_u}l^-(\hat{y}_{uj}) \Big)\Big)
    \end{split}
\end{equation}
\end{definition}

In practice, assuming for each positive item, we sample $N$ positive items and $M$ negative items, then the empirical loss can be written as: 

\begin{equation}
    \begin{split}
    &\widetilde{\mathcal{L}}^{Debiased}_{point}= \E_{u} E_{i \sim p_u^+} \Bigl( \tau_u^+l^+(\hat{y}_{ui}) + \\ 
    &\lambda \Bigl( \frac{1}{N} \sum_{j=1;j\sim p_u}^N l^-(\hat{y}_{uj})- \tau_u^+\frac{1}{M} \sum_{k=1;k\sim p_u^+}^M l^-(\hat{y}_{uk}) \Bigr) \Bigr)
    \end{split}
\end{equation} 

As an example, we have the ideal $\overline{\mathcal{L}}_{CCL}$, debiased $\mathcal{L}^{debiased}_{CCL}$ , empirical debiased CCL $\widetilde{\mathcal{L}}^{Debiased}_{ccl}$ losses as follows: 

\begin{equation*}
    \begin{cases}
    & \overline{\mathcal{L}}_{CCL} = \E_{u} \Bigg(\tau_u^+\E_{i\sim p_u^+}(1-\hat{y}_{ui}) +
    \lambda\tau_u^-\E_{j\sim p_u^-}[ReLU(\hat{y}_{uj}-\epsilon)]\Bigg) \\
    &\mathcal{L}^{debiased}_{CCL} = \E_{u} \Bigg(\tau_u^+\E_{i\sim p_u^+}(1-\hat{y}_{ui}) \\  
    &+  \lambda\Big( 
    \E_{j\sim p_u}[ReLU(\hat{y}_{uj}-\epsilon)] - \tau_u^+ \E_{k\sim p^+_u}[ReLU(\hat{y}_{uk}-\epsilon)]\Big)\Bigg) \\
    &\widetilde{\mathcal{L}}^{Debiased}_{ccl}= \E_{u} E_{i \sim p_u^+} \Bigl( \tau_u^+(1-\hat{y}_{ui}) + \\ 
    &\lambda \Bigl( \frac{1}{N} \sum_{j=1;j\sim p_u}^N ReLU(\hat{f}_{uj}-\epsilon)- \tau_u^+\frac{1}{M} \sum_{k=1;k\sim p_u^+}^M ReLU(\hat{f}_{uj}-\epsilon) \Bigr) \Bigr)
    \end{cases}
\end{equation*}
We will study how (empirical) debiased MSE and CCL in the experimental result section. 

\subsection{iALS, EASE and their Debiasness}  
Here, we investigate how the debiased MSE loss will impact the solution of two (arguably) most popular linear recommendation models, iALS ~\cite{hu2008collaborative,ials_revisiting} and EASE~\cite{ease}. 
Rather surprisingly, we found the solvers of both models can absorb the debiased loss under their existing framework with reasonable conditions (details below). In other words, both iALS and EASE can be considered to be inherently debiased. 

To obtain the aforementioned insights, we first transform the debiased loss into the summation form being used in iALS and EASE. 
\begin{small}
\begin{equation*}
\begin{split}
&\mathcal{L}^{Debiased}_{mse} \approx 
\E_{u} \Bigg[ \frac{\tau^+_u} {|\mathcal{I}^+_u|}\sum\limits_{i\in \mathcal{I}_u^+}(\hat{y}_{ui}-1)^2 + \lambda\Bigg(\frac{1}{|\mathcal{I}|} \sum\limits_{t\in\mathcal{I}} \hat{y}_{ut}^2 -\frac{\tau^+_u}{|\mathcal{I}_u^+|}\sum\limits_{q\in \mathcal{I}_u^+}\hat{y}^2_{uq} \Bigg)\Bigg]\\
&=\sum\limits_{u} \Bigg[ \frac{1}{|\mathcal{I}|} c_u \sum\limits_{i\in \mathcal{I}_u^+}(\hat{y}_{ui}-1)^2 \  \ \ \ \ \ \ \ \ \ \ \ \ \ \text{where, }  c_u=\frac{|\mathcal{I}|}{|\mathcal{I}_u^+|} \tau_u^+
\\
&+ \lambda\Bigg(\frac{1}{|\mathcal{I}|} \sum\limits_{t\in\mathcal{I}} \hat{y}_{ut}^2 -\frac{1}{|\mathcal{I}|} c_u \sum\limits_{q\in \mathcal{I}_u^+}\hat{y}^2_{uq} \Bigg)\Bigg]\\
&\propto \sum\limits_{u} \Bigg[ c_u \sum\limits_{i\in \mathcal{I}_u^+}(\hat{y}_{ui}-1)^2 + \lambda\Bigg(\sum\limits_{t\in\mathcal{I}} \hat{y}_{ut}^2 - c_u\sum\limits_{q\in \mathcal{I}_u^+}\hat{y}^2_{uq} \Bigg)\Bigg]\\
&=\sum\limits_{u} \Bigg[ \sum\limits_{i\in \mathcal{I}_u^+}[c_u(\hat{y}_{ui}-1)^2 -c_u\lambda \hat{y}^2_{uq}]+ \lambda\sum\limits_{t\in\mathcal{I}} \hat{y}_{ut}^2 \Bigg]
\end{split}
\end{equation*}
\end{small}

\noindent{\bf Debiased iALS:}
Following the original iALS paper \citep{ials@hu2008}, and the latest revised iALS \citep{ials_revisiting,ials++}, the objective of iALS is given by \citep{ials_revisiting}:

\begin{equation}
    \begin{split}
    \mathcal{L}_{iALS}&=\sum\limits_{(u,i)\in S} (\hat{y}(u,i)-1)^2 + \alpha_0\sum\limits_{u\in U}\sum\limits_{i\in I}\hat{y}(u,i)^2\\
    &+\lambda\Bigg( \sum\limits_{u\in U}(|\mathcal{I}_u^+|+\alpha_0|\mathcal{I}|)^{\nu}||\mathbf{w}_u||^2 + \sum\limits_{i\in I}(|U_i^+|+\alpha_0|\mathcal{U}|)^{\nu}||\mathbf{h}_i||^2\Bigg)
    \end{split}
\end{equation}
where $\lambda$ is global regularization weight (hyperparameter) and $\alpha_0$ is unobserved weight (hyperparameter). $\mu$ is generally set to be $1$. 
Also, $\mathbf{w}_u$ and $\mathbf{h}_i$ are user and item vectors for user $u$ and item $i$, respectively. 

Now, we consider applying the debiased MSE loss to replace the original MSE loss (the first line in $\mathcal{L}_{iALS}$) with the second line unchanged: 

\begin{small}
\begin{equation*}
\begin{split}
&\mathcal{L}^{Debiased}_{iALS}
=\sum\limits_{u\in U} \Bigg[ \sum\limits_{i\in \mathcal{I}_u^+}[c_u(\hat{y}_{ui}-1)^2 -c_u\cdot \alpha_0 \cdot \hat{y}^2_{ui}]+ \alpha_0\sum\limits_{t\in I} \hat{y}_{ut}^2 \Bigg]\\
 &+\lambda\Bigg( \sum\limits_{u\in U}(|\mathcal{I}_u^+|+\alpha_0|\mathcal{I}|)^{\nu}||\mathbf{w}_u||^2 + \sum\limits_{i\in I}(|\mathcal{U}_i^+|+\alpha_0|\mathcal{U}|)^{\nu}||\mathbf{h}_i||^2\Bigg)
\end{split}
\end{equation*}
\end{small}

Then we have the following conclusion:
\begin{theorem}\label{th1}
For any debiased iALS loss $\mathcal{L}^{Debiased}_{iALS}$ with parameters $\alpha_0$ and $\lambda$ with constant $c_u$ for all users, there are original iALS loss with parameters  $\alpha_0^\prime$ and $\lambda^\prime$, which have the same closed form solutions (up to a constant factor) for fixing item vectors and user vectors, respectively. (proof in \cref{app:proof})
\end{theorem}

\noindent{\bf Debiased EASE:}
EASE~\cite{ease} has shown to be a simple yet effective recommendation model and evaluation bases (for a small number of items), and it aims to minimize:
\begin{equation*}
    \begin{split}
        &\mathcal{L}_{ease}=||X-XW||^2_F + \lambda||W||^2_F\\
        & s.t. \quad diag(W)=0
    \end{split}
\end{equation*}
It has a closed-form solution~\cite{ease}:
\begin{equation}
\begin{split}
&P = (X^TX+\lambda I)^{-1}\\
    &W^* = I-P\cdot dMat(diag(1\oslash P))
\end{split}
\end{equation}
where $\oslash$ denotes the elementwise division, and $diag$ vectorize the diagonal and $dMat$ transforms a vector to a diagonal matrix. 

To apply the debiased MSE loss into the EASE objective, let us first further transform $\mathcal{L}^{Debiased}_{mse}$ into parts of known positive items and parts of unknown items: 
\begin{equation*}
    \begin{split}
&\mathcal{L}^{Debiased}_{mse}=\sum\limits_{u} \Bigg[ \sum\limits_{i\in \mathcal{I}_u^+}[c_u(\hat{y}_{ui}-1)^2 -c_u\lambda \hat{y}^2_{uq}]+ \lambda\sum\limits_{t\in\mathcal{I}} \hat{y}_{ut}^2 \Bigg]\\
        &=\sum\limits_{u} \Bigg[ \sum\limits_{i\in \mathcal{I}_u^+}[c_u(\hat{y}_{ui}-1)^2 -c_u\lambda \hat{y}^2_{ui}]+ \lambda\sum\limits_{t\in\mathcal{I}_u^+} \hat{y}_{ut}^2 + \lambda\sum\limits_{p\in\mathcal{I}\backslash \mathcal{I}_u^+} \hat{y}_{ut}^2 \Bigg]\\
&=\sum\limits_{u} \Bigg[\sum\limits_{i\in \mathcal{I}_u^+}[c_u(\hat{y}_{ui}-1)^2+c_u \sum\limits_{p\in\mathcal{I}\backslash \mathcal{I}_u^+} \hat{y}_{ut}^2 \Bigg]\\
&+ \lambda (1-c_u)\sum\limits_{u}\sum\limits_{i\in \mathcal{I}_u^+}\hat{y}^2_{ui} + (\lambda-c_u) \sum\limits_{p\in\mathcal{I}\backslash \mathcal{I}_u^+} \hat{y}_{ut}^2 \\
 &=||\sqrt{C_u}(X-XW)||^2_F \\
        &-\lambda||X\odot \sqrt{C_u-I}XW||^2_F -||(1-X)\odot \sqrt{C_u-\lambda I}XW||^2_F
    \end{split}
\end{equation*}
Note that 
\begin{equation}
\sum\limits_{u}\sum\limits_{i\in \mathcal{I}}\hat{y}^2_{ui}=||XW||^2_F
\end{equation}

To find the closed-form solution, we further restrict $\lambda=1$ and consider $c_u$ as a constant (always $>1$), thus, we have the following simplified format: 


\begin{equation}
    \begin{split}
    \mathcal{L}^{Debiased}_{mse}=&||\sqrt{C_u}(X-XW)||^2_F
        -||\sqrt{C_u-I}XW||^2_F\\
        &=||X-XW||^2_F
        -\alpha ||XW||^2_F 
    \end{split}
\end{equation}
where, $\alpha=c_u-1$. 
Note that if we only minimize this objective function $\mathcal{L}^{Debiased}_{mse}$, we have the closed form: 
$$W = ((1-\lambda )X^TX)^{-1}X^TX.$$ 

Now, considering the debiased version of EASE:  
\begin{equation*}
    \begin{split}
        &\mathcal{L}^{Debiased}_{ease}=||X-XW||^2_F-\alpha||XW||^2_F + \lambda||W||^2_F\\
        & s.t. \quad diag(W)=0
    \end{split}
\end{equation*}
It has the following closed-form solution (by similar inference as EASE~\cite{ease}):


\begin{equation}
\begin{split}
\widehat{W}&=\frac{1}{1-\alpha} (I -\hat{P}\cdot dMat( \vec{1}\oslash diag(\hat{P}) ), \\ 
where \ \  & \hat{P}= (X^TX+\frac{\lambda}{1-\alpha} I)^{-1}
\end{split}
\end{equation}
Now, we can make the following observation: 
\begin{theorem}
For any debiased EASE loss $\mathcal{L}^{Debiased}_{ease}$ with parameters $\alpha$ and $\lambda$ with constant $c_u>1$ for all users, there are original EASE loss with parameter,  $\lambda^\prime$, which have the same closed form solutions EASE (up to constant factor). 
\end{theorem}
\begin{proof}
Following above analysis and let 
$\lambda^\prime=\frac{\lambda}{(1-\alpha)c_u}$.
\end{proof}

These results also indicate the sampling based approach to optimize the debiased $MSE$ loss may also be rather limited. In the experimental results, we will further validate this conjecture on debiased MSE loss. 

\section{Experiment}\label{sec:ccl-experiment}

\begin{table}
\caption{Statistics of the datasets.}
\label{tab:dataset}
\centering
\begin{tabular}{c|c|c|c}
\hline { Dataset } & User \# & Item \# & Interaction \#  \\
\hline 
\hline Yelp2018 & 31,668 & 38,048 & $1,561,406$  \\
\hline Gowalla & 29,858 & 40,981 & $1,027,370$  \\
\hline Amazon-Books & 52,643 & 91,599 & $2,984,108$  \\
\hline
\end{tabular}
\end{table}

\begin{table*}[]
\caption{The performance of different loss functions w.r.t its debiased one on MF (Matrix Factorization). The unit of the metric values is $\%$. We also present and highlight the relative improvement (RI) of the debiased loss function with its vanilla (biased) counterpart.}
\label{tab:debias}
\resizebox{1.0\textwidth}{!}{%
\begin{scriptsize}
\begin{tabular}{|cc|cc|cc|cc|c|}
\hline
\multicolumn{2}{|c|}{\multirow{2}{*}{Loss}}                                                                                                                           & \multicolumn{2}{c|}{Yelp}                              & \multicolumn{2}{c|}{Gowalla}                             & \multicolumn{2}{c|}{Amazon-Books}                        & \multirow{2}{*}{Average RI\%} \\ \cline{3-8}
\multicolumn{2}{|c|}{}                                                                                                                                                & \multicolumn{1}{c|}{Recall@20}       & NDCG@20         & \multicolumn{1}{c|}{Recall@20}        & NDCG@20          & \multicolumn{1}{c|}{Recall@20}        & NDCG@20          &                               \\ \hline
\multicolumn{1}{|c|}{\multirow{3}{*}{CCL}}                                                   & biased                                                                 & \multicolumn{1}{c|}{6.91}            & 5.67            & \multicolumn{1}{c|}{18.17}            & 14.61            & \multicolumn{1}{c|}{5.27}             & 4.22             & -                             \\ \cline{2-9} 
\multicolumn{1}{|c|}{}                                                                       & debiased                                                               & \multicolumn{1}{c|}{6.98}            & 5.71            & \multicolumn{1}{c|}{18.42}            & 14.97            & \multicolumn{1}{c|}{5.49}             & 4.40             & -                             \\ \cline{2-9} 
\multicolumn{1}{|c|}{}                                                                       & RI\%                                                                   & \multicolumn{1}{c|}{\textbf{1.01\%}} & \textbf{0.71\%} & \multicolumn{1}{c|}{\textbf{1.38\%}}  & \textbf{2.46\%}  & \multicolumn{1}{c|}{\textbf{4.17\%}}  & \textbf{4.27\%}  & \textbf{2.33\%}               \\ \hline
\multicolumn{1}{|c|}{\multirow{3}{*}{MSE}}                                                   & biased                                                                 & \multicolumn{1}{c|}{5.96}            & 4.95            & \multicolumn{1}{c|}{14.27}            & 12.39            & \multicolumn{1}{c|}{3.36}             & 2.68             & -                             \\ \cline{2-9} 
\multicolumn{1}{|c|}{}                                                                       & debiased                                                               & \multicolumn{1}{c|}{6.05}            & 5.00            & \multicolumn{1}{c|}{14.98}            & 12.64            & \multicolumn{1}{c|}{3.35}             & 2.61             & -                             \\ \cline{2-9} 
\multicolumn{1}{|c|}{}                                                                       & RI\%                                                                   & \multicolumn{1}{c|}{\textbf{1.51\%}} & \textbf{1.01\%} & \multicolumn{1}{c|}{\textbf{4.98\%}}  & \textbf{2.02\%}  & \multicolumn{1}{c|}{\textbf{-0.30\%}} & \textbf{-2.61\%} & \textbf{1.10\%}               \\ \hline
\multicolumn{1}{|c|}{\multirow{3}{*}{InfoNCE}}                                               & biased                                                                 & \multicolumn{1}{c|}{6.54}            & 5.36            & \multicolumn{1}{c|}{16.45}            & 13.43            & \multicolumn{1}{c|}{4.60}             & 3.56             & -                             \\ \cline{2-9} 
\multicolumn{1}{|c|}{}                                                                       & debiased                                                               & \multicolumn{1}{c|}{6.66}            & 5.45            & \multicolumn{1}{c|}{16.57}            & 13.72            & \multicolumn{1}{c|}{4.65}             & 3.66             & -                             \\ \cline{2-9} 
\multicolumn{1}{|c|}{}                                                                       & RI\%                                                                   & \multicolumn{1}{c|}{\textbf{1.83\%}} & \textbf{1.68\%} & \multicolumn{1}{c|}{\textbf{0.73\%}}  & \textbf{2.16\%}  & \multicolumn{1}{c|}{\textbf{1.09\%}}  & \textbf{2.81\%}  & \textbf{1.72\%}               \\ \hline
\multicolumn{1}{|c|}{\multirow{4}{*}{\begin{tabular}[c]{@{}c@{}}MINE\\ (ours)\end{tabular}}} & MINE                                                                   & \multicolumn{1}{c|}{6.56}            & 5.37            & \multicolumn{1}{c|}{16.93}            & 14.28            & \multicolumn{1}{c|}{5.00}             & 3.93             & -                             \\ \cline{2-9} 
\multicolumn{1}{|c|}{}                                                                       & \begin{tabular}[c]{@{}c@{}}RI\% \\ (w.r.t baised infoNCE)\end{tabular} & \multicolumn{1}{c|}{\textbf{0.31\%}} & \textbf{0.19\%} & \multicolumn{1}{c|}{\textbf{2.92\%}}  & \textbf{6.33\%}  & \multicolumn{1}{c|}{\textbf{8.70\%}}  & \textbf{10.39\%} & \textbf{4.80\%}               \\ \cline{2-9} 
\multicolumn{1}{|c|}{}                                                                       & MINE+                                                                  & \multicolumn{1}{c|}{7.12}            & 5.86            & \multicolumn{1}{c|}{18.13}            & 15.31            & \multicolumn{1}{c|}{5.18}             & 4.05             & -                             \\ \cline{2-9} 
\multicolumn{1}{|c|}{}                                                                       & \begin{tabular}[c]{@{}c@{}}RI\% \\ (w.r.t baised infoNCE)\end{tabular} & \multicolumn{1}{c|}{\textbf{8.87\%}} & \textbf{9.33\%} & \multicolumn{1}{c|}{\textbf{10.21\%}} & \textbf{14.00\%} & \multicolumn{1}{c|}{\textbf{12.61\%}} & \textbf{13.76\%} & \textbf{11.46\%}              \\ \hline
\end{tabular}
\end{scriptsize}
}
\end{table*}

In this section, we experimentally study the most widely used loss functions in the existing recommendation models and their corresponding debiased loss as well as newly proposed MINE loss under recommendation settings. This could help better understand and compare the resemblance and discrepancy between different loss functions, and understand their debiasing benefits.  Specifically, we would like to answer the following questions:
\begin{itemize}
\item Q1. How does the bias introduced by the existing negative sampling impact model performance and how do our proposed debiased losses perform compared to traditional biased ones?
\item Q2. How does our proposed MINE objective function in Section \ref{sec:ccl-softmax} compare to widely used ones?
\item Q3. What is the effect of different hyperparameters on the various debiased loss functions?
\end{itemize}

\subsection{Experimental Setup}

\subsubsection{Datasets}
We use three datasets, \textbf{Amazon-Books}, \textbf{Yelp2018} and \textbf{Gowalla} commonly used by a number of recent studies \citep{simplex,lightgcn,ngcf,sgl-ed}. We obtain the publicly available processed data and follow the same setting as \citep{ngcf,lightgcn,simplex}.

\subsubsection{Evaluation Metrics} Consistent with benchmarks \citep{ngcf,lightgcn,simplex}, we evaluate the performance by $Recall@20$ and $NDCG@20$ over all items \citep{walid@sample}.

\subsubsection{Models}
This chapter focuses on the study of loss functions, which is architecture agnostic. Thus, we evaluate our loss functions on top of the simplest one - Matrix Factorization (MF) in this study. By eliminating the influence of the architecture, we can compare the power of various loss functions fairly. 
Meantime, we compare the performance of the proposed methods with a few classical machine learning models as well as advanced deep learning models, including, iALS \citep{ials_revisiting}, MF-BPR\citep{bpr}, MF-CCL \citep{simplex}, YouTubeNet \citep{youtube}, NeuMF \citep{ncf}, LightGCN \citep{lightgcn}, NGCF \citep{ngcf}, CML \citep{cml}, PinSage \citep{pinsage}, GAT \citep{gat}, MultiVAE \citep{multi-vae}, SGL-ED \citep{sgl-ed}.
Those comparisons can help establish that the loss functions with basic MF can produce strong and robust baselines for the recommendation community to evaluate more advanced and sophisticated models. 

\subsubsection{Loss functions}
As the core of this chapter, we focus on some most widely used loss functions and their debiased counterparts, including InfoNCE \citep{infonce}, MSE \citep{ials@hu2008,ials_revisiting}, and state-of-the-art CCL \citep{simplex}, as well as new proposed MINE \citep{MINE}.

\subsubsection{Reproducibility}

To implement reliable MF-based baselines (with CCL, MSE, InfoNCE objectives), we search a few hyperparameters and present the best one. By default, we set the batch size of training as $512$. Adam optimizer learning rate is initially set as $1e-4$ and reduced by $0.5$ on the plateau, early stopped until it arrives at $1e-6$. For the cases where negative samples get involved, we would set it as $800$ by default. We search the global regularization weight between $1e-9$ to $1e-5$ with an increased ratio of $10$. 


\begin{table*}[]

\resizebox{1.0\textwidth}{!}{%
\begin{scriptsize}
\begin{tabular}{|ccccccc|}
\hline
\multicolumn{1}{|c|}{\multirow{2}{*}{Model}} & \multicolumn{2}{c|}{Yelp}                                               & \multicolumn{2}{c|}{Gowalla}                                              & \multicolumn{2}{c|}{Amazon-Books}                  \\ \cline{2-7} 
\multicolumn{1}{|c|}{}                       & \multicolumn{1}{c|}{Recall@20}     & \multicolumn{1}{c|}{NDCG@20}       & \multicolumn{1}{c|}{Recall@20}      & \multicolumn{1}{c|}{NDCG@20}        & \multicolumn{1}{c|}{Recall@20}     & NDCG@20       \\ \hline
\multicolumn{7}{|c|}{Deep Learning Based}                                                                                                                                                                                                               \\ \hline
\multicolumn{1}{|c|}{YouTubeNet* \citep{youtube}}            & \multicolumn{1}{c|}{6.86}          & \multicolumn{1}{c|}{\textbf{5.67}} & \multicolumn{1}{c|}{17.54}          & \multicolumn{1}{c|}{14.73}          & \multicolumn{1}{c|}{5.02}          & 3.88          \\ \hline
\multicolumn{1}{|c|}{NeuMF* \citep{ncf}}                 & \multicolumn{1}{c|}{4.51}          & \multicolumn{1}{c|}{3.63}          & \multicolumn{1}{c|}{13.99}          & \multicolumn{1}{c|}{12.12}          & \multicolumn{1}{c|}{2.58}          & 2             \\ \hline
\multicolumn{1}{|c|}{CML* \citep{cml}}                   & \multicolumn{1}{c|}{6.22}          & \multicolumn{1}{c|}{5.36}          & \multicolumn{1}{c|}{16.7}           & \multicolumn{1}{c|}{12.92}          & \multicolumn{1}{c|}{\textbf{5.22}} & 4.28          \\ \hline
\multicolumn{1}{|c|}{MultiVAE* \citep{multi-vae}}              & \multicolumn{1}{c|}{5.84}          & \multicolumn{1}{c|}{4.5}           & \multicolumn{1}{c|}{16.41}          & \multicolumn{1}{c|}{13.35}          & \multicolumn{1}{c|}{4.07}          & 3.15          \\ \hline
\multicolumn{1}{|c|}{LightGCN* \citep{lightgcn}}              & \multicolumn{1}{c|}{6.49}          & \multicolumn{1}{c|}{5.3}           & \multicolumn{1}{c|}{\textbf{18.3}}  & \multicolumn{1}{c|}{\textbf{15.54}} & \multicolumn{1}{c|}{4.11}          & 3.15          \\ \hline
\multicolumn{1}{|c|}{NGCF* \citep{ngcf}}                  & \multicolumn{1}{c|}{5.79}          & \multicolumn{1}{c|}{4.77}          & \multicolumn{1}{c|}{15.7}           & \multicolumn{1}{c|}{13.27}          & \multicolumn{1}{c|}{3.44}          & 2.63          \\ \hline
\multicolumn{1}{|c|}{GAT* \citep{gat}}                   & \multicolumn{1}{c|}{5.43}          & \multicolumn{1}{c|}{4.32}          & \multicolumn{1}{c|}{14.01}          & \multicolumn{1}{c|}{12.36}          & \multicolumn{1}{c|}{3.26}          & 2.35          \\ \hline
\multicolumn{1}{|c|}{PinSage* \citep{pinsage}}               & \multicolumn{1}{c|}{4.71}          & \multicolumn{1}{c|}{3.93}          & \multicolumn{1}{c|}{13.8}           & \multicolumn{1}{c|}{11.96}          & \multicolumn{1}{c|}{2.82}          & 2.19          \\ \hline
\multicolumn{1}{|c|}{SGL-ED* \citep{sgl-ed}}                & \multicolumn{1}{c|}{6.75}          & \multicolumn{1}{c|}{5.55}          & \multicolumn{1}{c|}{-}              & \multicolumn{1}{c|}{-}              & \multicolumn{1}{c|}{4.78}          & 3.79          \\ \hline
\multicolumn{7}{|c|}{MF based}                                                                                   \\ \hline
\multicolumn{1}{|c|}{iALS** \citep{ials_revisiting}}                & \multicolumn{1}{c|}{6.06}          & \multicolumn{1}{c|}{5.02}          & \multicolumn{1}{c|}{13.88}          & \multicolumn{1}{c|}{12.24}          & \multicolumn{1}{c|}{2.79}          & 2.25        \\
\hline
\multicolumn{1}{|c|}{MF-BPR** \citep{bpr}}                & \multicolumn{1}{c|}{5.63}          & \multicolumn{1}{c|}{4.60}          & \multicolumn{1}{c|}{15.90}          & \multicolumn{1}{c|}{13.62}          & \multicolumn{1}{c|}{3.43}          & 2.65         \\ \hline
\multicolumn{1}{|c|}{MF-CCL** \citep{simplex}}               & \multicolumn{1}{c|}{\textbf{6.91}} & \multicolumn{1}{c|}{\textbf{5.67}} & \multicolumn{1}{c|}{\textbf{18.17}} & \multicolumn{1}{c|}{14.61}          & \multicolumn{1}{c|}{\textbf{5.27}} & \textbf{4.22} \\ \hline
\multicolumn{1}{|c|}{MF-CCL-debiased (ours)} & \multicolumn{1}{c|}{\textbf{6.98}} & \multicolumn{1}{c|}{\textbf{5.71}} & \multicolumn{1}{c|}{\textbf{18.42}} & \multicolumn{1}{c|}{\textbf{14.97}} & \multicolumn{1}{c|}{\textbf{5.49}} & \textbf{4.40} \\ \hline
\multicolumn{1}{|c|}{MF-MINE+ (ours)}        & \multicolumn{1}{c|}{\textbf{7.12}} & \multicolumn{1}{c|}{\textbf{5.86}} & \multicolumn{1}{c|}{18.13}          & \multicolumn{1}{c|}{\textbf{15.31}} & \multicolumn{1}{c|}{5.18}          & \textbf{4.06} \\ \hline
\end{tabular}
\end{scriptsize}
}
\caption{Perfomance comparison to widely used models. We highlight the top-3 best models in each column. Results of models marked with $*$ are duplicated from \citep{simplex} for consistency. For a fair comparison of the MF-based models marked with $**$ are reproduced by ourselves with thorough parameter searching detailed in \cref{app:reproduce}.}
\label{tab:main}
\end{table*}

\subsection{Q1. (Biased) Loss vs Debiased Loss}
Here, we seek to examine the negative impact of bias introduced by negative sampling and assess the effectiveness of our proposed debiased loss functions. We utilize Matrix Factorization (MF) as the backbone and compare the performance of the original (biased) version and the debiased version for the loss functions of $CCL$ \citep{simplex}, $MSE$ \citep{ials@hu2008,ials_revisiting}, and $InfoNCE$ \citep{infonce}.   Table \ref{tab:debias} shows the biased loss vs debiased loss (on top of the MF) results. 

The results show that the debiased methods consistently outperform the biased ones for the CCL and InfoNCE loss functions in all cases. In particular, the debiased method exhibits remarkable improvements of up to over $4\%$ for the CCL loss function on the $Amazon-Books$ dataset. This highlights the adverse effect of bias in recommendation quality and demonstrates the success of applying contrastive learning based debiasing methods in addressing this issue.
In addition, despite the robustness with respect to the biased MSE loss for iALS and EASE (as being proof in \cref{sec:linear}, the debiased MSE still provides a slight gain over the existing MSE loss. Finally, we note that the performance gained by the debiasing losses is also consistent with the results over other machine learning tasks as shown in ~\cite{debiased}.

\subsection{Q2. Performance of MINE $\backslash$ MINE+}
Next, we delve into the examination of the newly proposed Mutual Information Neural Estimator (MINE) loss function. As we mentioned before, to our best knowledge, this is the first study utilizing and evaluating MINE loss in recommendation settings. We compare MINE with other loss functions, all set MF as the backbone, as displayed in Table \ref{tab:debias},. 
Our results show the surprising effectiveness of MINE-type losses.
In particular, the basic MINE is shown close to $5\%$ average lifts over the (biased) InfoNCE loss, a rather stand loss in recommendation models. MINE+ demonstrates comparable or superior performance compared to the state-of-the-art CCL loss function. Furthermore, MINE+ consistently achieves better results than InfoNCE with a maximum improvement of 14\% (with an average improvement over $11\%$. 

These findings highlight the potential of MINE as an effective objective function for recommendation models and systems.  
Noting that our MINE+ loss is model agnostic, which can be applied to any representation-based model.
But using it with the basic MF already provides strong and robust baselines for evaluating recommendation models (and systems). 

Furthermore, we compare MF-MINE+ as well as our debiased MF-CCL model with some advanced machine learning and deep learning models in \cref{tab:main}. We can find these two methods perform best or second best in most cases, affirming the superiority of the MINE+ loss and debiased CCL.


\subsection{Q3. Hyperparameter analysis}
Here, we analyze the impacts of three key hyperparameters in the debiased framework: negative weight $\lambda$, number of negative samples, and number of positive samples as shown in \cref{fig:hyper_pos_neg,fig:hyper_lambda}.

The left side of \cref{fig:hyper_pos_neg} reveals that, in general, increasing the number of negative samples leads to better performance in models. The MSE loss requires fewer negative samples to reach a stable plateau but performs worse when the number of negative samples increases. The right side of \cref{fig:hyper_pos_neg} shows that the number of positive samples in the debiased formula has a minor impact on performance compared to negative samples.

In \cref{fig:hyper_lambda}, we illustrate how the negative weight affects the performance of the debiased CCL, MSE, and InfoNCE losses. In general, we would expect the $\cap$ shaped curve for quality metric ($Recall@20$ here) for the parameters, which is consistent with our curves. The optimal negative weight for these loss functions locates at different orders of magnitude due to their intrinsic differences of loss functions. In general, the negative weights of MSE and CCLE are relatively small and the negative weights of InfoNCE are in the order of hundreds and related to the number of negative samples (\cref{app:reproduce}). 

In \cref{fig:temp}, we investigate how temperature can affect the performance of $MINE+$ loss function. In general, it would achieve optimal around $0.5$ for all datasets (see \cref{app:reproduce}).

\begin{figure}
    \centering
    \includegraphics[width=0.8\linewidth]{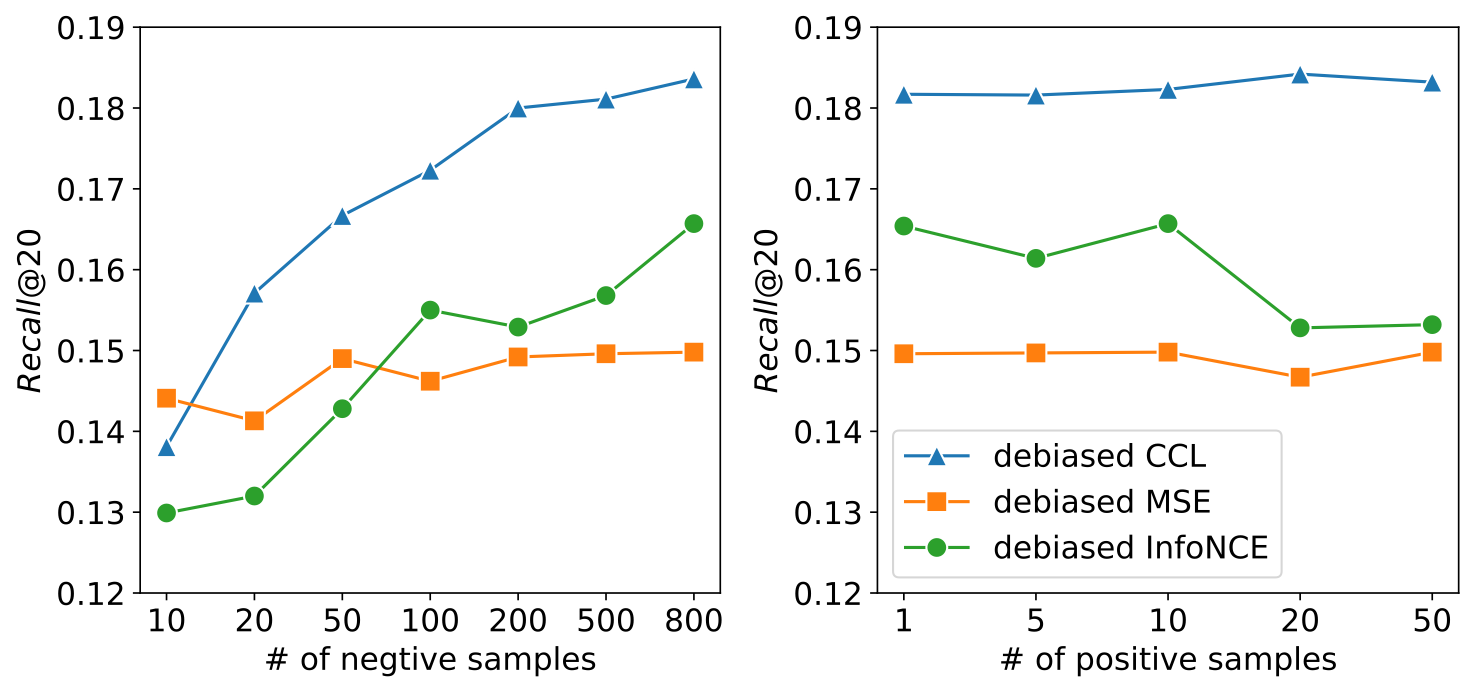}
    \vspace{-5pt}
    \caption{Effect of number of positive and negative samples on $Gowalla$}
    \label{fig:hyper_pos_neg}
\end{figure}

\begin{figure}
    \centering
    \includegraphics[width=0.8\linewidth]{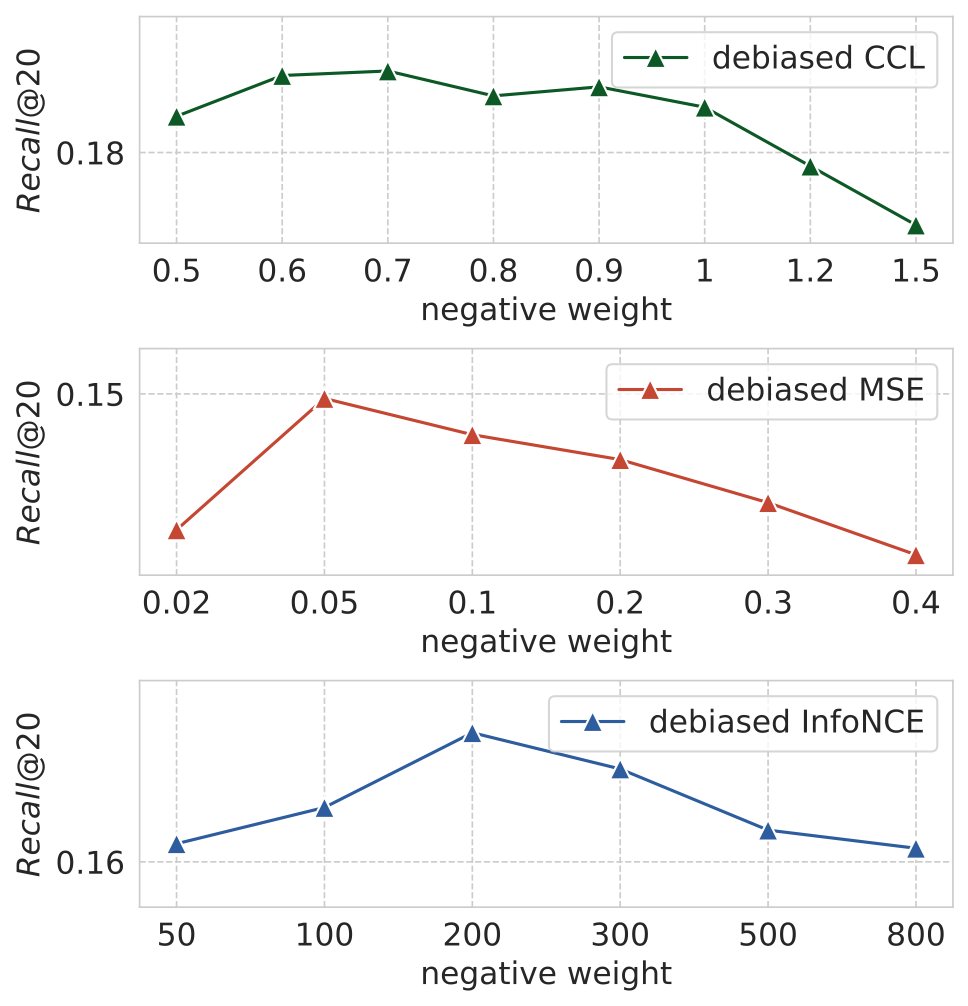}
    \vspace{-5pt}
    \caption{Effect of negative weight on $Gowalla$}
    \label{fig:hyper_lambda}
\end{figure}

\begin{figure}
    \centering
    \includegraphics[width=0.95\linewidth]{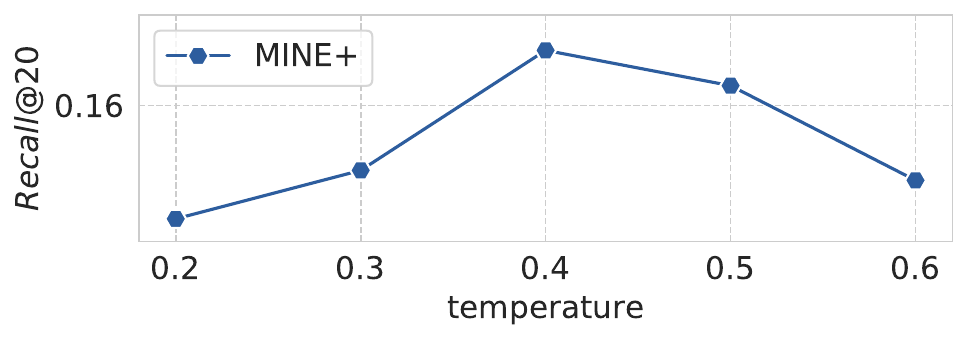}
    \caption{Effect of temperature on $Gowalla$}
    \label{fig:temp}
\end{figure}

\section{conclusion}
\label{sec:ccl-conclusion}
In this chapter, we conduct a comprehensive analysis of recommendation loss functions from the perspective of contrastive learning. 
We introduce a list of debiased losses and  new mutual information-based loss functions - MINE and MINE+ to recommendation setting. We  also show how the point-wise loss functions can be debiased and certified both iALS and EASE are inherently debiased. 
The empirical experimental results demonstrate the debiased losses and new information losses outperform the existing (biased) ones. 
In the future, we would like to investigate how the loss functions work with more sophisticated neural architectures and help discover more effective recommendation models by leveraging these losses.

\chapter{Related Work}\label{chpt:related_all}

\section{Introduction}
Over the past few decades, Machine Learning has gained widespread prominence and captivated attention across numerous domains. Its remarkable potential has led to its application in various fields, including human health \cite{jf@JMIR,junwu@fineehr,pelin}, finance, e-commerce, natural language processing \cite{xinjin@symlm,yubo@ci-poetry,yubo@text}, computer vision \cite{yuxin@remote,yuxin@segmentation,yuxin@cross,xinhong,xiaomin@emod}, autonomous vehicles, audio, and speech processing \cite{yubo@soundsieve}, robotics, wireless and ML \cite{caitao@deepMTL,caitao@localization}, human-computer interaction \cite{kong-eye-track,kong-multi,kong-ignore,kong-invis}, network \cite{caitao@networking,junwu@bottrinet,junwu@cellpad}, information retrieval \cite{yuxin@que2engage,yuxin@image}, supply chain optimization, sentiment analysis, Internet-of-things (IoT) \cite{xinjin@iot}, On-device ML \cite{yubo@device}, bioinformatics, and drug discovery. The versatility of Machine Learning \cite{prml}, especially deep learning \cite{xiaomin@parallel,xiaomin@spp} has allowed it to revolutionize these domains, enabling advancements in diagnostics, personalized recommendations, data analysis \cite{xiaobo@data,xiaobo@queue}, explainable ML \cite{ziheng@relax,ziheng@dark,ziheng@trans}, pattern recognition \cite{kong-group}, time series analysis \cite{xiaomin@tts}, and decision-making. With its continuous evolution and advancements, Machine Learning continues to shape and transform diverse industries, bringing forth new opportunities and possibilities.








In today's era of rapid e-commerce and thriving social media platforms, recommender systems have become indispensable tools for businesses \cite{deeprec,graph-rec}. These systems take on various forms across industries, from suggesting products on online marketplaces like Amazon and Taobao to generating personalized playlists for video and music services such as YouTube, Netflix, and Spotify. With the overwhelming amount of information available, users heavily rely on recommender systems to alleviate the burden of information overload and effortlessly explore their interests among a sea of options, including products, movies, news, and restaurants \cite{multi-vae}.

\section{Recommender System Related Work}
\label{sec:related-main}

\subsection{Graph Neural Networks}
Graph Neural Networks (GNNs) have indeed shown great promise in various fields, including machine learning~\cite{kipf2017semisupervised,yuwang@tree,yuwang@ssl,yuwang@fair,yuwang@imbal}, mobile computing \cite{guimin@graph,guimin@graph-icmla,guimin@graph-iot,guimin@mobile-graph} and recommender systems~\cite{graph-rec,xiaohan@dynamic,xiaohan@pretrain,xiaohan@hyperbolic,basket@xiaohan}. They are particularly effective for tasks that involve structured data, such as social networks, biological networks, or any kind of graph-structured data. Besides the most essential GNN-based work, which directly works on user-item bipartite graph \cite{chen2020revisiting,lightgcn,ngcf}, there are a few advanced techniques that enhance the graph recommendation. PinSage \cite{pinsage} utilized a random-walk-based sampling method to obtain the fixed size of neighborhoods to scale it up. KGAT \cite{kgat} included the item side information to enhance the graph-based recommendation.

\subsection{Objectives of Implicit Collaborative Filtering}

Implicit feedback has been popular for decades in Recommender System \citep{rendle2021item} since Hu et al. first proposed iALS \citep{ials@hu2008}, where a second-order pointwise objective - Mean Square Error (MSE) was adopted to optimize user and item embeddings alternatively. Due to its effectiveness and efficiency, there is a large number of works that take the MSE or its variants as their objectives, spanning from Matrix Factorization (MF) based models \citep{ials_revisiting,ials++,eals}, to regression-based models, including SLIM \citep{slim}, EASE \citep{ease,edlae}, etc. He et al. treat collaborative filtering  as a binary classification task and apply the pointwise objective - Binary Cross-Entropy Loss (BCE) onto it \citep{ncf}. CML utilized the pairwise hinge loss onto the collaborative filtering scenario \citep{cml}. MultiVAE \citep{multi-vae} utilizes multinomial likelihood estimation. Rendle et al. proposed a Bayesian perspective pairwise ranking loss - BPR, in the seminal work \citep{bpr}. YouTubeNet posed a recommendation as an extreme multiclass classification problem and apply the Softmax cross-entropy (Softmax) \citep{youtube,sample-softmax}. Recently, inspired by the largely used contrastive loss in the computer vision area, \citep{simplex} proposed a first-order based Cosine Contrastive Loss (CCL), where they maximize the cosine similarity between the positive pairs (between users and items) and minimize the similarity below some manually selected margin.

\subsection{Contrastive Learning}
Recently, Contrastive Learning (CL) has become a prominent optimizing framework in deep learning \citep{SimCLR,infonce}. The motivation behind CL is to learn representations by contrasting positive and negative pairs as well as maximize the positive pairs, including data augmentation \citep{SimCLR} and multi-view representations \citep{multi-cl}, etc. 
Chuang et al. proposed a new unsupervised contrastive representation learning framework, targeting to minimize the bias introduced by the selection of the negative samples. This debiased objective consistently improved its counterpart - the biased one in various benchmarks \citep{debiased}. Lately, by introducing the margin between positive and negative samples, Belghazi et al. present a supervised contrastive learning framework that is robust to biases \citep{unbiased}.
In addition, a mutual information estimator - MINE \citep{MINE}, closely related to the contrastive learning loss InfoNCE\citep{infonce}, has demonstrated its efficiency in various optimizing settings. \citep{alignment-uniformity} study the alignment and uniformity of user, item embeddings from the perspective of contrastive learning in the recommendation. CLRec \citep{cl-rec} design a new contrastive learning loss which is equivalent to using inverse propensity weighting to reduce exposure bias of a large-scale system. 



\subsection{Negative Sampling Strategies for Item Recommendation}
In most real-world scenarios, only positive feedback is available which brings demand for negative signals to avoid trivial solutions during training recommendation models. Apart from some non-sampling frameworks like iALS \citep{ials@hu2008}, Multi-VAE \citep{multi-vae}, the majority of models \citep{negative_sampling_review,lightgcn,ncf,bpr,rendle2021item,sampling_strategy} would choose to sample negative items for efficiency consideration. Uniform sampling is the most popular one \citep{bpr} which assumes uniform prior distribution. Importance sampling is another popular choice \citep{importance}, which chooses negative according to their frequencies. Adaptive sampling strategy keep tracking and picking the negative samples with higher scores \citep{aobpr,dynamic}. Another similar strategy is SRNS \citep{robustsample}, which tries to find and optimize false negative samples 
In addition, NCE \citep{nce} approach for negative sampling has also been adopted in the recommendation system\citep{sampling_strategy}. 


\subsection{Bias and Debias in Recommendation}
In the recommender system, there are various different bias \citep{bias-debias}, 
including item exposure bias, popularity bias, position bias, selection bias, etc.
These issues have to be carefully handled otherwise it would lead to unexpected results, such as inconsistency of online and offline performance, etc. Marlin et al. \citep{select-bias} validate the existence of selection bias by conducting a user survey, which indicates the discrepancy between the observed rating data and all data. Since users are generally exposed to a small set of items, unobserved data doesn't mean negative preference, this is so-called exposure bias \citep{expose-bias}. Position bias \citep{position-bias} is another commonly 
encountered one, where users are more likely to select items that display in a desirable place or higher position on a list.

Steck et al. \citep{steck-debias} propose an unbiased metric to deal with selection bias. iALS \citep{ials@hu2008} put some weights on unobserved data which helps improve the performance and deal with exposure bias. In addition, the sampling strategy would naturally diminish the exposure bias during model training \citep{bias-debias,bpr}. EXMF \citep{user-bias} is a user exposure simulation model that can also help reduce exposure bias.
To reduce position bias, some models \citep{dbcl,aecc} make assumptions about user behaviors and estimate true relevance instead of directly using clicked data. \citep{unbiased-pair} proposed an unbiased pairwise BPR estimator and further provide a practical technique to reduce variance. Another widely used method to correct bias is the utilization of propensity score \citep{unbiased-propensity,unbiased-rec-miss} where each item is equipped with some position weight.

\bibliographystyle{IEEEtran}
\bibliography{IEEEabrv, dissertation, kdd23,kdd21,more}

\appendix
\chapter{Appendix}
\appendix
\label{appendix}
\section{Reproducibility}
Generally, We follow the (strong generalization) experiment set-up in ~\cite{liang2018variational,DBLP:conf/nips/Steck20} and also the pre-processing of the three public available datasets, MovieLens 20 Million (ML-20M) \cite{ml20m}, Netflix Prize (Netflix) \cite{netflix}, and the Million Song Data (MSD)\cite{msddataset}.

\subsection{Experimental Set-up for Table \ref{tab:ml-20mbasic} and \ref{tab:netflixbasic}}
For the experiment of table \ref{tab:ml-20mbasic} and \ref{tab:netflixbasic}, we utilize the strong generalization protocol for EASE \cite{Steck_2019} and LRR methods. For Matrix Factorization based methods (MF and WMF/ALS), they are trained for with data (except the items to be evaluated in validation and test sets)  Note that this actually provides an advantage to the matrix factorization approaches as they prefer to learn the embeddings (latent factors) before its prediction. The experiment results present in table \ref{tab:ml-20mbasic} and \ref{tab:netflixbasic} are obtained by parameter grid search over the validation set according to $nDCG@100$, the same as \cite{liang2018variational}. The searching results are listed as following : table \ref{tab:ml20m:mf:search}, table \ref{tab:netflix:mf:search},
table \ref{tab:ml20m:lrr:search} and table \ref{tab:netflix-lrr-search}.

\begin{table}[H]
\resizebox{\linewidth}{!}{%
\begin{tabular}{|c|c|c|c|c|c|c|c|c|}
\hline
\multicolumn{2}{|c|}{\multirow{2}{*}{}} & \multicolumn{7}{c|}{$\lambda$}                                                \\ \cline{3-9} 
\multicolumn{2}{|c|}{}                  & 0       & 10      & \textbf{50}      & 100     & 200     & 500    & 1000   \\ \hline
\multirow{5}{*}{k}    & 128             & 0.29874 & 0.30951 & 0.36488          & 0.37826 & 0.30121 & 0.1901 & 0.1901 \\ \cline{2-9} 
                      & 256             & 0.22911 & 0.25104 & 0.37504          & 0.37826 & 0.30121 & 0.1901 & 0.1901 \\ \cline{2-9} 
                      & 512             & 0.1546  & 0.19782 & 0.39682          & 0.37826 & 0.30121 & 0.1901 & 0.1901 \\ \cline{2-9} 
                      & \textbf{1000}   & 0.09177 & 0.18242 & \textbf{0.39893} & 0.37826 & 0.30121 & 0.1901 & 0.1901 \\ \cline{2-9} 
                      & 1500            & 0.06089 & 0.20776 & 0.39893          & 0.37826 & 0.30121 & 0.1901 & 0.1901 \\ \hline
\end{tabular}%
}
\caption{ML-20M, MF, parameter search}
\label{tab:ml20m:mf:search}
\end{table}

\begin{table}[H]
\resizebox{\linewidth}{!}{%
\begin{tabular}{|c|c|c|c|c|c|c|c|c|}
\hline
\multicolumn{2}{|c|}{\multirow{2}{*}{}} & \multicolumn{7}{c|}{$\lambda$}                                                  \\ \cline{3-9} 
\multicolumn{2}{|c|}{}                  & 0       & 10      & 50      & \textbf{100}     & 200     & 500     & 1000    \\ \hline
\multirow{5}{*}{k}    & 128             & 0.31297 & 0.31542 & 0.32607 & 0.34103          & 0.34132 & 0.25831 & 0.17414 \\ \cline{2-9} 
                      & 256             & 0.25036 & 0.25536 & 0.28659 & 0.33521          & 0.34132 & 0.25831 & 0.17414 \\ \cline{2-9} 
                      & 512             & 0.17485 & 0.18314 & 0.26081 & 0.35157          & 0.34132 & 0.25831 & 0.17414 \\ \cline{2-9} 
                      & \textbf{1000}   & 0.12036 & 0.13766 & 0.28868 & \textbf{0.36414} & 0.34132 & 0.25831 & 0.17414 \\ \cline{2-9} 
                      & 1500            & 0.09147 & 0.12449 & 0.32103 & 0.36414          & 0.34132 & 0.25831 & 0.17414 \\ \hline
\end{tabular}%
}
\caption{Netflix, MF, parameter search}
\label{tab:netflix:mf:search}
\end{table}

\begin{table}[H]
\resizebox{\linewidth}{!}{%
\begin{tabular}{|c|c|r|r|r|r|r|r|r|}
\hline
\multicolumn{2}{|c|}{\multirow{2}{*}{}} & \multicolumn{7}{c|}{$\lambda$}                                                  \\ \cline{3-9} 
\multicolumn{2}{|c|}{} &
  \multicolumn{1}{c|}{8000} &
  \multicolumn{1}{c|}{9000} &
  \multicolumn{1}{c|}{\textbf{10000}} &
  \multicolumn{1}{c|}{11000} &
  \multicolumn{1}{c|}{12000} &
  \multicolumn{1}{c|}{13000} &
  \multicolumn{1}{c|}{14000} \\ \hline
\multirow{3}{*}{k}    & 1000            & 0.41063 & 0.41273 & 0.41432          & 0.41476 & 0.41515 & 0.41513 & 0.41478 \\ \cline{2-9} 
                      & \textbf{2000}   & 0.41332 & 0.41469 & \textbf{0.41533} & 0.41509 & 0.41499 & 0.41455 & 0.41394 \\ \cline{2-9} 
                      & 3000            & 0.41282 & 0.41397 & 0.41473          & 0.4146  & 0.41452 & 0.41413 & 0.41347 \\ \hline
\end{tabular}%
}
\caption{ML-20M, LRR, parameter search}
\label{tab:ml20m:lrr:search}
\end{table}

\begin{table}[H]
\resizebox{\linewidth}{!}{%
\begin{tabular}{|c|c|r|r|r|r|r|r|}
\hline
\multicolumn{2}{|c|}{\multirow{2}{*}{}} & \multicolumn{6}{c|}{$\lambda$}                                        \\ \cline{3-8} 
\multicolumn{2}{|c|}{} &
  \multicolumn{1}{c|}{10000} &
  \multicolumn{1}{c|}{20000} &
  \multicolumn{1}{c|}{30000} &
  \multicolumn{1}{c|}{\textbf{40000}} &
  \multicolumn{1}{c|}{50000} &
  \multicolumn{1}{c|}{60000} \\ \hline
\multirow{3}{*}{k}    & 2000            & 0.33632 & 0.37139 & 0.37856 & 0.37942          & 0.37828 & 0.37644 \\ \cline{2-8} 
                      & \textbf{3000}   & 0.34905 & 0.37441 & 0.37934 & \textbf{0.37949} & 0.37807 & 0.37617 \\ \cline{2-8} 
                      & 4000            & 0.35184 & 0.37468 & 0.37919 & 0.37931          & 0.37786 & 0.37601 \\ \hline
\end{tabular}%
}
\caption{Netflix, LRR, parameter search}
\label{tab:netflix-lrr-search}
\end{table}

\subsection{Experimental Set-up for Table \ref{table: maintable}}

In table \ref{table: maintable}, for LRR (closed-form) model ( described in equation \ref{eq:regressionpca}). For ML-20M dataset, we set $k=2000$, $\lambda = 8000$, $c = 1000$ (used to control range of weighted $\lambda_i$). For Netflix dataset the $\lambda = 8000$ , $\lambda = 40000$, $c = 5000$ . Noting that these hyper-parameters are not set as optimal ones (described in table \ref{tab:ml20m:lrr:search}, table \ref{tab:netflix-lrr-search}), which won't affect our claims. For EDLAE (including full rank) model, we obtain the similarity matrix by running the code from \cite{DBLP:conf/nips/Steck20}. For WMF/ALS model and EASE model, we set the hyper-parameters as table \ref{tab:ml-20mbasic} and table \ref{tab:netflixbasic}. Other models' results are obtained form \cite{liang2018variational}, \cite{Steck_2019} and \cite{DBLP:conf/nips/Steck20}.

For fast training augmented model, we sample part of training data. Generally, it takes 2.5 minutes per 100 batch (batch size is 2048) for training.

\subsection{MSD Dataset Results}

The table \ref{tab:msd_data} shows our experiment results carried out on the MSD dataset. Baseline models' results are obtained form \cite{liang2018variational}, \cite{Steck_2019} and \cite{DBLP:conf/nips/Steck20}.

\begin{table}[H]
\resizebox{\linewidth}{!}{%
\begin{tabular}{|c|c|c|c|c|}
\hline
\multicolumn{2}{|l|}{\multirow{2}{*}{}}               & \multicolumn{3}{c|}{MSD}                          \\ \cline{3-5} 
\multicolumn{2}{|l|}{}                                & Recall@20       & Recall@50       & nDCG@100      \\ \hline
\multicolumn{2}{|c|}{LRR}                             & 0.24769         & 0.33509         & 0.30127       \\ \cline{3-5} 
\multicolumn{2}{|c|}{LRR + $\lambda_i$}                    & 0.25083         & 0.33902         & 0.30372       \\ \hline
\multicolumn{2}{|c|}{EDLAE}                           & 0.26391         & 0.35465         & 0.31951       \\ \cline{3-5} 
\multicolumn{2}{|c|}{EDLAE Full Rank}                 & 0.33408         & 0.42948         & 0.39151       \\ \cline{3-5} 
\multicolumn{2}{|c|}{EDLAE Full Rank+HT}              & 0.33423         & 0.43134         & 0.38851       \\ \hline
\multicolumn{2}{|c|}{SLIM}                            & \multicolumn{3}{c|}{did not finished in \cite{slim01}} \\ \cline{3-5} 
\multicolumn{2}{|c|}{WMF}                             & 0.211           & 0.312           & 0.257         \\ \cline{3-5} 
\multicolumn{2}{|c|}{EASE}                            & 0.333           & 0.428           & 0.389         \\ \cline{3-5} 
\multicolumn{2}{|c|}{CDAE}                            & 0.188           & 0.283           & 0.237         \\ \cline{3-5} 
\multicolumn{2}{|c|}{MULT-DAE}                        & 0.266           & 0.363           & 0.313         \\ \cline{3-5} 
\multicolumn{2}{|c|}{MULT-VAE}                        & 0.266           & 0.364           & 0.316         \\ \hline
\multirow{3}{*}{dataset statistics} & \# items        & \multicolumn{3}{c|}{41140}                        \\ \cline{2-5} 
                                    & \# users        & \multicolumn{3}{c|}{571355}                       \\ \cline{2-5} 
                                    & \# interactions & \multicolumn{3}{c|}{34 millions}                  \\ \hline
\end{tabular}%
}
\caption{The performance comparison between models on MSD dataset.}
\label{tab:msd_data}
\end{table}



\section{reproducibility}\label{app:reproduce}


In \cref{tab:mine-repro}, we list the details of the hyperparameters for reproducing the results of $MINE+$ in \cref{tab:main,tab:debias}.

In \cref{tab:ccl-repro}, we list the details of the hyperparameters for reproducing the results of debiased $CCL$ in \cref{tab:debias}.

\begin{table}[]
\caption{Hyperparameter details of $MINE+$ loss}
\label{tab:mine-repro}
\begin{tabular}{|c|c|c|c|}
\hline
\multicolumn{1}{|l|}{} & Yelp2018 & Gowalla & Amazon-Books \\ \hline
negative weight        & 1.1      & 1.2     & 1.1          \\ \hline
temperature            & 0.5      & 0.4     & 0.4          \\ \hline
regularization         & 1e-8       & 1e-8      & 1e-8           \\ \hline
number of negative     & 800      & 800     & 800          \\ \hline
\end{tabular}
\end{table}

\begin{table}[]
\caption{Hyperparameter details of debiased $CCL$ loss}
\label{tab:ccl-repro}
\begin{tabular}{|c|c|c|c|}
\hline
\multicolumn{1}{|l|}{} & Yelp2018 & Gowalla & Amazon-Books \\ \hline
negative weight        & 0.4      & 0.7     & 0.6          \\ \hline
margin                 & 0.9      & 0.9     & 0.4          \\ \hline
regularization         & -9       & -9      & -9           \\ \hline
number of negative     & 800      & 800     & 800          \\ \hline
number of positive     & 10       & 20      & 50           \\ \hline
\end{tabular}
\end{table}

\section{Proof}\label{app:proof}

\subsection{Proof of \cref{th1}}

\begin{proof}
To derive the alternating least square solution, we have: 
\begin{itemize}
    \item Fixing item vectors, optimizing: 
    \begin{equation*}
    \begin{split}
        \mathcal{L}_u &=||\sqrt{c_u}{(H^{u}_S)}^T\mathbf{w}_u - \sqrt{c_u}\mathbf{y}_S||^2 + ||\sqrt{\alpha_0}H^T\mathbf{w}_u||^2 + ||\sqrt{\lambda_u}I \cdot \mathbf{w}_u||^2\\
        &-||\sqrt{c_u\alpha_0} {(H^{u}_S)}^T\mathbf{w}_u||^2, \ \ \ \ where \ \ \ \ \lambda_u = \lambda(|\mathcal{I}_u^+|+\alpha_0|\mathcal{I}|)^{\nu} 
    \end{split}
\end{equation*}
Note that the observed item matrix for user $u$ as $H^u_S$:

\begin{equation}
H^u_S=\begin{bmatrix}
\vrule & \vrule &        & \vrule\\
\mathbf{h}_1 & \mathbf{h}_s & \cdots & \mathbf{h}_{|\mathcal{I}_u^+|}\\
\vrule & \vrule &        & \vrule\\
\end{bmatrix}\in \mathbb{R}^{k\times |\mathcal{I}_u^+|},\quad s\in \mathcal{I}_u^+
\end{equation}

$H$ is the entire item matrix for all items. 
$y_S$ are all $1$ column (the same row dimension as $H^u_S$.

    \item Fixing User Vectors, optimizing: 
 \begin{equation*}
    \begin{split}
        \mathcal{L}_i &=||{(W^{i}_S \sqrt{C_u})}^T\mathbf{h}_i - \sqrt{C_u}\mathbf{y}_S||^2 + ||\sqrt{\alpha_0}W^T\mathbf{h}_i||^2 + ||\sqrt{\lambda_i}I \cdot \mathbf{h}_i||^2\\
        &-||\sqrt{\alpha_0} {(\sqrt{C_u}W^{i}_S)}^T\mathbf{h}_i||^2, \ \ \ \ where \ \ \lambda_i = \lambda(|\mathcal{U}_i^+|+\alpha_0|\mathcal{U}|)^{\nu}
    \end{split}
\end{equation*}
Here, $W^{i}_S$ are observed user matrix for item $i$, and $W$ are the entire user matrix, $C_u=diag(c_u)$ a $|\mathcal{U}|\times|\mathcal{U}|$ diagonal matrix  with $c_u$ on the diagonal.  
\end{itemize}

Solving $\mathcal{L}_u$ and $\mathcal{L}_i$, we have the following closed form solutions: 
\begin{equation*}
    \begin{split}
        & \mathbf{w}_u^*=\Big(c_u (1 -\alpha_0)H^{u}_S{(H^{u}_S)}^T + \alpha_0 HH^T + \lambda_u I\Big)^{-1}\cdot H^{u}_S\cdot \sqrt{c_u}\mathbf{y}_S \\ 
        & \mathbf{h}_i^*=\Big((1 -\alpha_0)W^{i}_SC_u{(W^{i}_S)}^T + \alpha_0 WW^T + \lambda_i I\Big)^{-1}\cdot W^{i}_S\cdot \sqrt{C_u}\mathbf{y}_S
    \end{split}
\end{equation*}

Assuming $c_u$ be a constant for all users, we get 

\begin{equation*}
    \begin{split}
        & \mathbf{w}_u^* \varpropto \Big(H^{u}_S{(H^{u}_S)}^T + \frac{\alpha_0}{(1-\alpha_0)c_u} HH^T + \frac{\lambda_u}{(1-\alpha_0)c_u} I\Big)^{-1}\cdot H^{u}_S\cdot \mathbf{y}_S \\ 
        & \mathbf{h}_i^* \varpropto \Big(W^{i}_S{(W^{i}_S)}^T + \frac{\alpha_0}{(1-\alpha_0)c_u} 
        WW^T + \frac{\lambda_i}{(1-\alpha_0)c_u} I\Big)^{-1}\cdot W^{i}_S \cdot \mathbf{y}_S
    \end{split}
\end{equation*}
Interestingly by choosing the right $\alpha_0$ and $\lambda$, the above solution is in fact the same solution (up to constant factor) for the original $\mathcal{L}_{iALS}$ ~\cite{hu2008collaborative}. Following the above analysis and let 
$\alpha_0^\prime=\frac{\alpha_0}{(1-\alpha_0)c_u}$ and 
$\lambda^\prime=\frac{\lambda}{(1-\alpha_0)c_u}$
\end{proof}

\section{ADMM}
\subsection{ADMM}
ADMM problem: 

$f,g$ convex. ($x,z$ are vectors)

\begin{equation}
    \begin{split}
        &\min f(x) + g(z)\\
        &s.t \quad Ax+Bz = c
    \end{split}
\end{equation}

Its equivalent to optimize:

\begin{equation}
    \begin{split}
        \mathcal{L}_{\rho}(x,y,z)&=f(x)+g(z)\\
        &+y^T(Ax+Bz-c)+(\rho/2)||Ax+Bz-c||^2_2
    \end{split}
\end{equation}

The alternating solution is:

\begin{equation}
\begin{aligned}
x^{k+1} & :=\operatorname{argmin}_x L_\rho\left(x, z^k, y^k\right) \\
z^{k+1} & :=\operatorname{argmin}_z L_\rho\left(x^{k+1}, z, y^k\right) \\
y^{k+1} & :=y^k+\rho\left(A x^{k+1}+B z^{k+1}-c\right)
\end{aligned}
\end{equation}

\subsection{Low-Rank EDLAE}
In \citep{edlae}, equation 10 is to optimize:

\begin{equation}
    \begin{split}
        &||X-X\cdot\{UV^T-dMat\big(diag(UV^T)\big)\}||^2_F\\
        &+||\Lambda^\frac{1}{2}\cdot \{UV^T-dMat\big(diag(UV^T)\big)\}||^2_F
    \end{split}
\end{equation}

This can be re-written as:

\begin{equation}
    \begin{split}
        &||X-XUV^T+X\cdot dMat(\beta)||^2_F+||\Lambda^\frac{1}{2}\cdot (UV^T-dMat(\beta)) ||^2_F\\
        &s.t. \quad \beta = diag(UV^T)
    \end{split}
\end{equation}

Using augmented Lagrangian:

\begin{equation}
    \begin{split}
       & ||X-XUV^T+X\cdot dMat(\beta)||^2_F+||\Lambda^{\frac{1}{2}}UV^T||-||dMat(\beta)||^2_F\\
       &+2\gamma^T\Omega (\beta-diag(UV^T))+||\Omega^\frac{1}{2}(\beta-diag(UV^T))||^2_F
    \end{split}
\end{equation}

$\gamma$ is the vector of Lagrangian multipliers, amd $\Omega = \Lambda + \omega I$ where $\omega>0$ is a scalar.

\begin{equation}
\begin{aligned}
\hat{\mathbf{U}} & \leftarrow\left(\mathbf{X}^{\top} \mathbf{X}+\Lambda+\Omega\right)^{-1}\left(\mathbf{X}^{\top} \mathbf{X} \cdot \operatorname{dMat}(\mathbf{1}+\hat{\beta})+\Omega \cdot \operatorname{dMat}(\hat{\beta}-\hat{\gamma})\right) \hat{\mathbf{V}}\left(\hat{\mathbf{V}}^{\top} \hat{\mathbf{V}}\right)^{-1} \\
\hat{\mathbf{V}}^{\top} & \leftarrow\left(\hat{\mathbf{U}}^{\top}\left(\mathbf{X}^{\top} \mathbf{X}+\Lambda+\Omega\right) \hat{\mathbf{U}}\right)^{-1} \hat{\mathbf{U}}^{\top}\left(\mathbf{X}^{\top} \mathbf{X} \cdot \operatorname{dMat}(\mathbf{1}+\hat{\beta})+\Omega \cdot \operatorname{dMat}(\hat{\beta}-\hat{\gamma})\right) \\
\hat{\beta} & \leftarrow \frac{\operatorname{diag}\left(\mathbf{X}^{\top} \mathbf{X} \hat{\mathbf{U}} \hat{\mathbf{V}}^{\top}\right)-\operatorname{diag}\left(\mathbf{X}^{\top} \mathbf{X}\right)+\operatorname{diag}(\Omega) \odot\left(\operatorname{diag}\left(\hat{\mathbf{U}} \hat{\mathbf{V}}^{\top}\right)+\hat{\gamma}\right)}{\operatorname{diag}\left(\mathbf{X}^{\top} \mathbf{X}\right)+\operatorname{diag}(\Omega-\Lambda)} \\
\hat{\gamma} & \leftarrow \hat{\gamma}+\operatorname{diag}\left(\hat{\mathbf{U}} \hat{\mathbf{V}}^{\top}\right)-\hat{\beta}
\end{aligned}
\end{equation}

\clearpage
\newpage

\section{EASE-debiased}

\begin{equation*}
    \begin{split}
&\mathcal{L}^{Debiased}_{mse}=\sum\limits_{u} \Bigg[ \sum\limits_{i\in \mathcal{I}_u^+}[c_u(\hat{y}_{ui}-1)^2 -c_u\lambda \hat{y}^2_{ui}]+ \lambda\sum\limits_{t\in\mathcal{I}} \hat{y}_{ut}^2 \Bigg]\\
        &=\sum\limits_{u} \Bigg[ \sum\limits_{i\in \mathcal{I}_u^+}[c_u(\hat{y}_{ui}-1)^2 -c_u\lambda \hat{y}^2_{ui}]+ \lambda\sum\limits_{t\in\mathcal{I}_u^+} \hat{y}_{ut}^2 + \lambda\sum\limits_{p\in\mathcal{I}\backslash \mathcal{I}_u^+} \hat{y}_{ut}^2 \Bigg]\\
&=\sum\limits_{u} \Bigg[\sum\limits_{i\in \mathcal{I}_u^+}[c_u(\hat{y}_{ui}-1)^2+c_u \sum\limits_{p\in\mathcal{I}\backslash \mathcal{I}_u^+} \hat{y}_{ut}^2 \Bigg]\\
&+ \lambda (1-c_u)\sum\limits_{u}\sum\limits_{i\in \mathcal{I}_u^+}\hat{y}^2_{ui} + (\lambda-c_u) \sum\limits_{p\in\mathcal{I}\backslash \mathcal{I}_u^+} \hat{y}_{ut}^2 \\
 &=||\sqrt{C_u}(X-XW)||^2_F \\
        &-\lambda||X\odot \sqrt{C_u-I}XW||^2_F -||(1-X)\odot \sqrt{C_u-\lambda I}XW||^2_F\\
    \end{split}
\end{equation*}

For Regression based $W = UV^T$
\begin{equation*}
\begin{split}
 &L=||\sqrt{C_u}(X-XUV^T)||^2_F \\
        &-\lambda||X\odot \sqrt{C_u-I}XUV^T||^2_F -||(1-X)\odot \sqrt{C_u-\lambda I}XUV^T||^2_F\\
        &+\gamma ||UV||^2_F\\
        &s.t. \quad diag(UV)=0\\
    \end{split}
\end{equation*}
\section{Mapping Function}
\subsection{ Properties of Recurrent function $f$}
\label{app_subsec:property}

In the following, we enumerate a list of interesting properties of this recurrent formula of $f$ based on Beta distribution. 

\begin{lemma}[Location of Last Point]
For any $a$, all $f(n)$ converge to $N$: $f(n)=N$. 
\end{lemma}
\begin{proof}
We note that 
\begin{small}
\begin{equation*}
\begin{split}
    &\sum\limits_{k=0}^{n-1}{\binom{n-1}{k}\mathcal{B}(a+k, n-k)} 
    = \int_{0}^{1}{\sum\limits_{k=0}^{n-1}{\binom{n-1}{k}} t^{a+k-1}(1-t)^{n-k-1}dt} \\
    &= \int_{0}^{1}{ t^ {a-1}\left[\sum\limits_{k=0}^{n-1}{\binom{n-1}{k}} t^{k}(1-t)^{n-k-1} \right]dt} 
    = \int_{0}^{1}{ t^ {a-1}} = \frac{1}{a}
    \end{split}
\end{equation*}
\end{small}
Add up \Cref{eq:l2r} from $k=1$ to $n-1$, we have:

\begin{equation*}\label{eq:nN}
\begin{split}
&\frac{[f(n)-1]^a-[f(1)-1]^a}{a[N-1]^a} = \sum\limits_{k=1}^{n-1}{\binom{n-1}{k}\mathcal{B}(a+k, n-k)}\\
&=\frac{1}{a}-\binom{n-1}{0}\mathcal{B}(a,n) \quad 
\text{then, we have } f(n) = N 
\end{split} 
\end{equation*}

\end{proof}

\noindent{\bf Uniform Distribution and Linear Map:} 
When the parameter $a=1$, the Beta distribution degenerates to the uniform distribution. 
From equations \Cref{eq:f1,eq:fk}, we have another simple linear map:
\begin{equation}\label{eq:uni_fk}
    f(k;a=1) = k\frac{N-1}{n} + 1
\end{equation}

Even though the user rank distribution is quite different from the uniform distribution, we found that this formula provides a reasonable approximation for the mapping function, and generally, better than the Naive formula \Cref{eq:fk1}. 
More interestingly, we found that when $a$ ranges from $0$ to $1$ (as they express an exponential-like distribution), they actually are quite close to this linear formula.

\noindent{\bf Approximately Linear:}
When we take a close look at the $f(k;a)$ sequences ($f(1;a), f(2;a), \cdots f(k;a)$ for different parameters $a$ from $0$ to $1$, \footnote{This also holds when $a>1$, but since the $Recall$, aka the user rank distribution, is typically very different from these settings, we do not discuss them here.} we find that when $k$ is large, $f(k;a)$ all gets very close to $f(k;1)$ (the linear map function for the uniform distribution). Figures~\ref{fig:rela_a1} show the relative difference of all $f(k;a)$ sequences for $a=0.2,0.6,0.8$ with respect to $a=1$, i.e., $[f(k;a)-f(k;a=1)]/f(k;a=1)$. Basically, they all converge quickly to $f(k;a=1)$ as $k$ increases.

To observe this, 
  let us take a look at their $f(K)$ locations when $k$ is getting large. To simplify our discussion, let $g(k) = f(K)-1$, and then we have 

\begin{equation*}\label{eq:LinearAp}
\begin{split}
     &[g(k+1)]^a - [g(k)]^a = a(N-1)^a\frac{\Gamma(n)}{\Gamma(n+a)}\frac{\Gamma(k+a)}{\Gamma(k+1)}\\
     &[\frac{g(k+1)}{g(k)}]^a = 1 +[\frac{(N-1)k}{g(k)n}]^a\frac{a}{k} \\
     & \text{when $n$ and $k$ are large,}  \lim_{n\rightarrow \infty}{\frac{\Gamma(n)}{\Gamma(n+a)}} = \frac{1}{n^a}  \\
     &[\frac{g(k+1)}{g(k)}] = \left(1 +[\frac{(N-1)k}{g(k)n}]^a\frac{a}{k}\right)^{1/a}\approx 1 +[\frac{(N-1)k}{g(k)n}]^a\frac{1}{k}
\end{split}
\end{equation*}

When $g(k) = (N-1)\frac{k}{n}$, the above equation holds $\frac{g(k+1)}{g(k)} = 1+\frac{1}{k}$, and this suggests they are all quite similar to the linear map $f(k;a=1)$ for the uniform distribution. 

By looking at the difference $f(k;a)-f(k-1;a)$, we notice we will get very close to the constant $\frac{N-1}{n}=f(k;1)-f(k-1;1)$ even when $k$ is small. To verify this, let $y_{k+1}$
\begin{small}
$$ = [f(k+1;a)-1]^a - [f(k;a) - 1]^a = a[N-1]^a\frac{\Gamma(n)}{\Gamma(n+a)}\frac{\Gamma(k+a)}{\Gamma(k+1)}$$
\end{small}
Then we immediately observe: 

\begin{equation*}
\begin{split}
    &\frac{y_{k+1}}{y_k}=\frac{
    a[N-1]^a\frac{\Gamma(n)}{\Gamma(n+a)}\frac{\Gamma(k+a)}{\Gamma(k+1)}}{a[N-1]^a\frac{\Gamma(n)}{\Gamma(n+a)}\frac{\Gamma(k-1+a)}{\Gamma(k)}}
    =1 + \frac{a-1}{k}
    \end{split}
\end{equation*}

Thus, after only a few iterations for $f(k;a)$, we have found that their (powered) difference will get close to being a constant. 
\section{Fix-Size-Sampling Estimation: EM Algorithm}\label{app_sec:fix-em}
In this section, we give the details of the EM algorithm for $MLE$ estimator.
Recalling \Cref{eq:mb_mle}
Weighted log-likelihood function is:
\begin{equation*}
    \log{\mathcal{L}} = \sum\limits_{u=1}^{M}{w_u \cdot\log{\sum\limits_{k=1}^{N}{p(x_u,z_{uk}|\theta_k)}}}
\end{equation*}

\noindent\textbf{E-step}

\begin{equation*}
\mathcal{Q}(\pmb{\pi},\pmb{\pi}^{old})=\sum\limits_{u=1}^{M}{w_u \sum\limits_{k=1}^{N}{\gamma(z_{uk}) \log{p(x_u,z_{uk}|\theta_k)}}}
\end{equation*}

where
\begin{equation*}
    \gamma(z_{uk})=p(z_{uk}|x_u,\pmb{\pi}^{old})=\frac{{\pi}^{old}_k p(x_u|\theta_k)}{\sum\limits_{j=1}^{N}{\pi^{old}_j p(x_u|\theta_j)}}
\end{equation*}

\noindent\textbf{M-step}
\begin{equation}\label{eq:Q_prime}
   \mathcal{Q}^\prime(\pmb{\pi},\pmb{\pi}^{old})=\mathcal{Q}(\pmb{\pi},\pmb{\pi}^{old}) + \lambda (1-\sum\limits_{k=1}^{N}{\pi_k})
\end{equation}

\begin{equation*}
    \frac{\partial \mathcal{Q}^\prime(\pmb{\pi},\pmb{\pi}^{old})}{\partial \pi_k}=\sum\limits_{u=1}^{M}{ \frac{w_u\gamma(z_{uk}) }{\pi_k} -\lambda}=0
\end{equation*}

\begin{equation*}
    \lambda\pi_k=\sum\limits_{u=1}^M{w_u\cdot\gamma(z_{uk})}
\end{equation*}

\begin{equation}
    \lambda=\sum\limits_{k=1}^{N}\sum\limits_{u=1}^M{w_u\cdot\gamma(z_{uk})}=\sum\limits_{u=1}^M{w_u}
\end{equation}

\begin{equation*}
    \pi^{new}_k=\frac{\sum\limits_{u=1}^M{w_u\cdot\gamma(z_{uk})}}{\sum\limits_{u=1}^M{w_u} }
\end{equation*}

\section{Linear Combination of Coefficients}\label{sec:app_linear}
Considering $X=(X_1,\dots,X_n)$ is the random variables of a sample $M$ times multinomial distribution with $n$ cells probabilities $(\theta_1,\dots,\theta_n)$. We have:$\frac{X_i}{M}\rightarrow \theta_i, \text{ when } M\rightarrow \infty$
\begin{equation}
    \begin{split}
        &\mathbb{E}[{X_i}] = M\theta_i\quad Var[X_i] = M\theta_i(1-\theta_i)\\
        &Cov(X_i, X_j) = -M\theta_i\theta_j
    \end{split}
\end{equation}
Considering the a new random variable deriving from the linear combination: $\mathcal{A} = \sum\limits_{i=1}^n{w_iX_i}$, where the ${w_i}$ are the constant coefficients.
\begin{small}
\begin{equation*}\label{eq:A_equation}
    \begin{split}
    \mathbb{E}[\mathcal{A}]&=  M\cdot\sum\limits_{i=1}^n{w_i\theta_i}\\
        Var[\mathcal{A}] &= \mathbb{E}[\mathcal{A}^2]-(\mathbb{E}[\mathcal{A}])^2=\sum\limits_i^{n}{w^2_i[M\theta_i-M\theta_i^2]}-2\sum\limits_{i\neq j}{w_iw_j[M\theta_i\theta_j]}\\
        &=M\sum\limits_{i}^n{w^2_i\theta_i}-M\cdot\Big(\sum\limits_{i}^n{w^2_i\theta^2_i}+2\sum\limits_{i\neq j}w_iw_j\theta_i\theta_j\Big)\\
        &=M\cdot\Big(\sum\limits_{i}^n{w^2_i\theta_i}-\big(\sum\limits_{i}^n{w_i\theta_i}\big)^2 \Big)
    \end{split}
\end{equation*}
\end{small}

\section{Rewriting of $L_2$}\label{sec:l2_re}

\begin{equation}
    \begin{split}
    \mathcal{L}_2&=\sum\limits_{R=1}^N{P(R)\cdot \mathbb{E}[\sum_{r=1}^n \tilde{P}(r|R) \widehat{\mathcal M}(r)- \sum_{r=1}^n  P(r|R) \widehat{\mathcal M}(r)]^2 }\\
        &=\sum\limits_{R=1}^N{P(R)\cdot \mathbb{E}[\sum_{r=1}^n \frac{X_r}{M\cdot P(R)} \widehat{\mathcal M}(r)- \sum_{r=1}^n  \frac{\mathbb{E}[X_r]}{M\cdot P(R)} \widehat{\mathcal M}(r)]^2 }\\
        &=\sum\limits_{R=1}^N{\frac{1}{M^2\cdot P(R)}\cdot \mathbb{E}[\sum_{r=1}^n {X_r}\widehat{\mathcal M}(r)- \sum_{r=1}^n  {\mathbb{E}[X_r]} \widehat{\mathcal M}(r)]^2 }\\
 &= \sum\limits_{R=1}^{N}{\frac{1}{M}\cdot \Big(\sum\limits_{r}^n{\widehat{\mathcal M}^2(r) P(r|R)}-\big(\sum\limits_{r}^n{\widehat{\mathcal M}(r)P(r|R)}\big)^2 \Big) } \\
 &= \sum\limits_{R=1}^{N}{\frac{1}{M} Var(\widehat{\mathcal M}(r)|R)} \nonumber
    \end{split}
\end{equation}

\section{Adaptive EM steps}\label{sec:app_em}
\noindent{\textbf{E-step}}\\
\begin{equation}
    \begin{split}
        \log\mathcal{L}&=\sum\limits_{u=1}^{M}\log\sum\limits_{k=1}^N{P(r_u, R_{uk}; \boldsymbol{\pi})}\\
        &= \sum\limits_{u=1}^{M}\log\sum\limits_{k=1}^N{\phi(R_{uk})\cdot \frac{P(r_u, R_{uk}; \boldsymbol{\pi})}{\phi(R_{uk})}}\\
       &\ge \sum\limits_{u=1}^{M}\sum\limits_{k=1}^N{\phi(R_{uk})\cdot\log P(r_u, R_{uk}; \boldsymbol{\pi})} + constant\\
       &\triangleq \sum\limits_{u=1}^{M} Q_u(\boldsymbol{\pi}, \boldsymbol{\pi}^{old}) = Q(\boldsymbol{\pi}, \boldsymbol{\pi}^{old})
       \end{split}
       \end{equation}
where
       \begin{equation}
           \begin{split}
        \phi(R_{uk}) &= P(R_u = k|r_u;\boldsymbol{\pi}^{old})\\
        &=\frac{\pi^{old}_k\cdot P(r_u|R_u = k; n_u)}{\sum\limits_{j=1}^N\pi^{old}_j\cdot P(r_u|R_u = j; n_u) }
    \end{split}
\end{equation}
where $\phi$ is the posterior, $\boldsymbol{\pi}^{old}$ is the known probability distribution and $\boldsymbol{\pi}$ is parameter.

\noindent{\textbf{M-step}:} Derived from the Lagrange maximization procedure:  
\begin{equation}
    \begin{split}
        \pi^{new}_k=\frac{1}{M}\sum\limits_{u=1}^M\phi(R_{uk})
    \end{split}
\end{equation}

\end{document}